\documentclass{article}
\usepackage{ijcai22}

\usepackage{times}

\usepackage{soul}
\usepackage{url}
\usepackage[hidelinks]{hyperref}
\usepackage[utf8]{inputenc}
\usepackage[small]{caption}
\usepackage{graphicx}
\usepackage{amsmath}
\usepackage{booktabs}
\usepackage{tikz}
\usepackage{multirow}

\usepackage{latexsym}
\usepackage{amsfonts}
\usepackage{amssymb}

\usepackage{amsthm}
\newtheorem{theorem}{Theorem}

\newtheorem{example}[theorem]{Example}
\newtheorem{proposition}[theorem]{Proposition}
\newtheorem{lemma}[theorem]{Lemma}

\usepackage{tikz}
\usetikzlibrary{automata}
\usetikzlibrary{positioning,arrows,fit,calc}
\usetikzlibrary{decorations.pathmorphing}
\usetikzlibrary{shapes}
\usetikzlibrary{backgrounds}
\usepackage{multirow}
\usepackage{booktabs}
\usepackage{xcolor,colortbl}
\usepackage{setspace}
\usetikzlibrary{decorations.pathreplacing}
\usetikzlibrary{matrix}

\usepackage{thmtools,thm-restate}

\usepackage{pifont}

\newcommand{\avec}[1]{\boldsymbol{#1}}

\newcommand{\LTL}{\textsl{LTL}}

\newcommand{\U}{\mathbin{\mathsf{U}}}
\newcommand{\Us}{\U_{s}}

\newcommand{\QBE}{\mathsf{QBE}}
\newcommand{\QBEbp}{\QBE^{\mathsf{b}\text{+}}}
\newcommand{\QBEbm}{\QBE_{\mathsf{b}\!\text{--}}}
\newcommand{\QBEba}{\QBE^{\mathsf{b}\text{+}}_{\mathsf{b}\text{-}}}

\newcommand{\unt}{\mathbin{\nabla}}
\newcommand{\sub}{\mathit{sub}}

\newcommand{\D}{\mathcal{D}}
\newcommand{\nxt}{{\ensuremath\raisebox{0.25ex}{\text{\scriptsize$\bigcirc$}}}}
\newcommand{\Rnext}{\nxt_{\!\scriptscriptstyle F}}

\newcommand{\Xallop}{^{\smash{\Box\raisebox{1pt}{$\scriptscriptstyle\bigcirc$}}}}
\newcommand{\Xbd}{^{\smash{\hspace*{-2pt}\Box\Diamond}}}

\newcommand{\horn}{\textit{horn}}

\newcommand{\op}{{\boldsymbol{o}}}

\newcommand{\Abox}{\mathcal D}

\newcommand{\q}{\avec{q}}
\newcommand{\TO}{\mathcal O}

\newcommand{\sig}{\mathsf{sig}}

\newcommand{\Qp}{\mathcal{Q}_p}
\newcommand{\dep}{\mathit{tdp}}

\newcommand{\tp}{\textit{tp}}

\newcommand{\leg}{\mathfrak L}
\newcommand{\lt}{\mathfrak t}

\usepackage[mathscr]{euscript}
 \let\mathscr\relax
\usepackage[scr]{rsfso}

\tikzset{
  basic box/.style = {
    shape = rectangle,
    align = center,
    draw  = #1,
    rounded corners},
  header node/.style = {
    font          = \strut\Large\ttfamily,
    text depth    = +0pt,
    fill          = white,
    draw},
  header/.style = {%
    inner ysep = +1.5em,
    append after command = {
      \pgfextra{\let\TikZlastnode\tikzlastnode}
      node [header node] (header-\TikZlastnode) at (\TikZlastnode.north) {#1}
    }
  },
  hv/.style = {to path = {-|(\tikztotarget)\tikztonodes}},
  vh/.style = {to path = {|-(\tikztotarget)\tikztonodes}},
  fat blue line/.style = {ultra thick, blue}
}

\newcommand{\slin}[3]
{
{
\scriptsize
\draw[-, very thick] ($({#1}) +(0,0.09)$) -- node[above=0.1] {#2}  node[below=0.1] {#3} ($({#1}) +(0,-0.09)$);
}
}

\newcommand{\NP}{\textsc{NP}}
\newcommand{\coNP}{\textsc{coNP}}

\newcommand{\PTime}{\textsc{P}}
\newcommand{\NExpTime}{\textsc{NExpTime}}
\newcommand{\ExpSpace}{\textsc{ExpSpace}}
\newcommand{\ExpTime}{\textsc{ExpTime}}
\newcommand{\PSpace}{\textsc{PSpace}}

\newcommand{\I}{{\mathcal{I}}}
\newcommand{\J}{{\mathcal{J}}}

\newcommand{\sm}{\lessdot}

\newcommand{\Type}{{\boldsymbol{T}}}

\newcommand{\C}{\mathcal{C}}
\newcommand{\Q}{\mathcal{Q}}
\newcommand{\M}{\boldsymbol{M}}
\newcommand{\B}{\mathsf{b}}

\newcommand{\Jmc}{\mathcal{J}}

\newcommand{\qa}{q_{\textit{acc}}}

\newcommand{\conf}{\mathfrak c}

\title{Reverse Engineering of Temporal Queries Mediated by LTL Ontologies}

\author{
Marie Fortin$^1$ \and
Boris Konev$^2$ \and
Vladislav Ryzhikov$^3$\and
Yury Savateev$^4$\and\\
Frank Wolter$^2$\And
Michael Zakharyaschev$^{3}$
\affiliations
$^1$Universit\'e Paris Cit\'e, CNRS, IRIF, France\\
$^2$Department of Computer Science, University of Liverpool, UK\\
$^3$Department of Computer Science and Information Systems, Birkbeck, University of London, UK\\
$^4$School of Electronics and Computer Science, University of Southampton, UK\\
\emails
mfortin@irif.fr, \{boris.konev,wolter\}@liverpool.ac.uk, \{vlad,michael\}@dcs.bbk.ac.uk, y.savateev@soton.ac.uk
}

\begin{document}


\maketitle


\begin{abstract}
In reverse engineering of database queries, we aim to construct a query from a given set of  answers and non-answers; it can then be used to explore the data further or as an explanation of the answers and non-answers. We investigate this query-by-example problem for queries formulated in positive fragments of linear temporal logic \LTL{} over timestamped data, focusing on the design of suitable query languages and the combined and data complexity of deciding whether there exists a query in the given language that separates the given answers from non-answers. We consider both plain \LTL{} queries and those mediated by \LTL-ontologies.
\end{abstract}


\section{Introduction}\label{intro}

Supporting users of databases by constructing a query from examples of answers and non-answers to the query has been a major research area since the 2000s~\cite{martins2019reverse}. In the database community, research has focussed on standard query languages such as SQL, graph query languages, and SPARQL~\cite{DBLP:conf/sigmod/ZhangEPS13,DBLP:conf/pods/WeissC17,kalashnikov2018fastqre,deutch2019reverse,DBLP:conf/icdt/StaworkoW12,DBLP:conf/icdt/Barcelo017,DBLP:journals/tods/CohenW16,DBLP:conf/www/ArenasDK16}. The KR community has been concerned with constructing queries from examples under the open world semantics and with background knowledge given by an ontology~\cite{GuJuSa-IJCAI18,DBLP:conf/gcai/Ortiz19,DBLP:conf/cikm/CimaCL21,kr2021restr,KR}.
A fundamental problem that has been investigated by both communities is known as \emph{separability} or \emph{query-by-example} (QBE), a term coined by Zloof~[\citeyear{DBLP:journals/ibmsj/Zloof77}]:
\begin{description}
\item[Given:] sets $E^{+}$ and $E^{-}$ of pairs $(\mathcal{D},\avec{d})$ with a database instance $\mathcal{D}$ and a tuple $\avec{d}$ in $\mathcal{D}$, a (possibly empty) ontology $\mathcal{O}$, and a query language $\mathcal{Q}$.

\item[Problem:] decide whether there exists a query $\q\in\mathcal{Q}$ \emph{separating} $(E^{+},E^{-})$ in the sense that $\mathcal{O},\mathcal{D}\models \q(\avec{d})$ for all $(\mathcal{D},\avec{d})\in E^{+}$ and $\mathcal{O},\mathcal{D}\not\models \q(\avec{d})$ for all $(\mathcal{D},\avec{d})\in E^{-}$.
\end{description}
If such a $\q$ exists, then $(E^{+},E^{-})$ is often called  \emph{satisfiable} w.r.t.~$\mathcal{Q}$ under $\TO$, and the construction of $\q$ is called \emph{learning}.

In many applications, the input data is timestamped and queries are naturally formulated in languages with temporal operators. In this paper, we investigate temporal query-by-example by focusing on the basic but very useful case where data $\Abox$ is a set of timestamped atomic  propositions.
Our query languages are positive fragments of \emph{linear temporal logic} \LTL{} with the temporal operators $\Diamond$ (eventually), $\nxt$ (next), and $\U$ (until) interpreted under the strict semantics~\cite{DBLP:books/cu/Demri2016}. To enforce generalisation, we do not admit $\lor$.
Our most expressive query language $\mathcal{Q}[\U]$ is thus defined as the set of formulas constructed from atoms using $\land$ and $\U$ (via which $\nxt$ and $\Diamond$ are expressible); the fragments $\mathcal{Q}[\Diamond]$ and $\mathcal{Q}[\nxt,\Diamond]$ are defined analogously.
Ontologies can be given in full $\LTL$ or its fragments $\LTL\Xbd$ (known as the \emph{Prior logic}~\cite{prior:1956b}), which only uses the operators $\Box$ (always in the future) and $\Diamond$, and the Horn fragment $\LTL_\horn\Xallop$ containing axioms of the form $C_1 \land \dots \land C_k \to C_{k+1}$, where the $C_i$ are atoms possibly prefixed by $\Box$ and $\nxt$ for $i\leq k+1$, and also by $\Diamond$ for $i\leq k$.
Ontology axioms are supposed to hold at all times.
In fact, already this basic `one-dimensional' temporal ontology-mediated querying formalism provides enough expressive power in those real-world situations where the interaction among individuals in the object domain is not important and can be disregarded in data modelling; see~\cite{DBLP:journals/ai/ArtaleKKRWZ21} and also Example~\ref{ex:00} and the references before it.

Within
this temporal setting, we take a broad view of the potential applications of the QBE problem. On the one hand, there are non-expert users who would like to explore data via queries but are not familiar with temporal logic. They usually are, however, capable of providing data examples illustrating the queries they are after. QBE supports such users in the construction of those queries. On the other hand, the positive and negative data examples might come from an application, and the user is interested in possible explanations of the examples. Such an explanation is then provided by a temporal query separating the positive examples from the negative ones. In this case, our goal is similar to recent work on learning \LTL{}  formulas in explainable planning and program synthesis~\cite{lemieux2015general,DBLP:conf/fmcad/NeiderG18,DBLP:conf/aips/CamachoM19,DBLP:conf/icgi/FijalkowL21,DBLP:conf/tacas/RahaRFN22,DBLP:conf/kr/FortinKRSWZ22}.


%
%
%
%
%

\begin{example}\label{ex:00}\em
Imagine an engineer whose task is to explain the behaviour of the monitored equipment (say, why an engine stops) in terms of qualitative sensor data such as `low temperature' ($T$), `strong vibration' ($V$), etc. Suppose the engine stopped after the runs $\Abox^+_1$ and $\Abox^+_2$ below but did not stop after the runs $\Abox^-_1$, $\Abox^-_2$, $\Abox^-_3$, where we assume the runs to start at $0$ and measurements to be recorded at moments $0,1,2,\dots$:
\begin{multline*}
\Abox^+_1 = \{T(2), V(4) \},  \Abox^+_2 = \{T(1), V(4)\},\\ \Abox^-_1 = \{T(1)\},  \Abox^-_2 = \{ V(4)\},  \Abox^-_3 = \{V(1), T(2)\}.
\end{multline*}
The $\Diamond$-query $\q = \Diamond (T \land \Diamond \Diamond V)$ is true at $0$ in the $\Abox^+_i$, false in $\Abox^-_i$, and so gives a possible explanation of what could cause the engine failure. The example set $(\{\Abox^+_3,\Abox^+_4\},\{\Abox^-_4\})$ with
\begin{multline*}
\Abox^+_3 = \{T(1), V(2) \}, \ \ \Abox^+_4 = \{T(1),T(2), V(3)\}, \\ \Abox^-_4 = \{T(1),V(3)\}
\end{multline*}
is explained by the $\U$-query $T \U V$.
Using background knowledge, we can compensate for sensor failures resulting in incomplete data. To illustrate, suppose $\mathcal{E}^+_1 = \{H(3), V(4) \}$, where $H$ means `heater is on'\!. If an ontology $\TO$ has the axiom $\nxt H \to T$ saying that a heater can only be triggered by the low temperature at the previous moment, then the same $\q$ separates $\{\mathcal{E}^+_1, \Abox^+_2\}$ from $\{\Abox^-_1, \Abox^-_2,\Abox^-_3\}$ under $\TO$.
\hfill $\dashv$
\end{example}
Query $\q$ in Example~\ref{ex:00} is of a particular `linear' form,
in which the order of atoms is fixed and not left
open as, for instance, in the `branching' $\Diamond T \wedge \Diamond V$. More precisely, \emph{path $\nxt\Diamond$-queries} in the class $\smash{\mathcal{Q}_{p}[\nxt,\Diamond]}$ take the form
\begin{align}\label{dnpath}
\q = \rho_0 \land \op_1 (\rho_1 \land \op_2 (\rho_2 \land \dots \land \op_n \rho_n) ),
\end{align}
where $\op_i \in \{\nxt, \Diamond \}$ and $\rho_i$ is a conjunction of atoms; $\mathcal{Q}_{p}[\Diamond]$
restricts $\op_i$ to $\{\Diamond \}$; and \emph{path $\U$-queries} $\Qp[\U]$ look like
\begin{align}\label{upath}
	\q = \rho_0 \land (\lambda_1 \U (\rho_1 \land ( \lambda_2 \U ( \dots (\lambda_n \U \rho_n) \dots )))),
\end{align}
where $\lambda_i$ is a conjunction of atoms or $\bot$. Path queries are motivated by two observations. First, if a query language admits conjunctions of queries---unlike our classes of path queries---then, dually to  overfitting for $\lor$, multiple negative examples become redundant: if $\q_{\mathcal{D}}$ separates $(E^{+},\{\mathcal{D}\})$, for each $\mathcal{D} \in E^{-}$, then $\bigwedge_{\mathcal{D}\in E^{-}}\q_\mathcal{D}$ separates $(E^{+},E^{-})$.
%
%
Second, numerous natural query types known from applications can be captured by path queries. For example, the existence of a \emph{common subsequence} of
the positive examples (regarded as words) that is not a subsequence of any negative one corresponds to the existence of a separating $\mathcal{Q}_{p}[\Diamond]$-query with $\rho_0=\top$ and
$\rho_{i}\not=\top$ for $i>0$, and the existence of a \emph{common subword} of the positive examples that is not a subword of any negative one corresponds to the existence of
a separating query of the form $\Diamond(\rho_{1} \land \nxt (\rho_2 \land \dots \land \nxt  \rho_n) )$. These and similar queries are the basis of data comparison programs with numerous applications in computational linguistics, bioinformatics, and revision control systems~\cite{878178,4609376,DBLP:journals/cor/BlumDSJLMR21}.

While path queries express the intended separating pattern of events in many applications, branching queries are needed if the order of events is irrelevant  for separation.

\begin{example}\em
In the setting of Example~\ref{ex:00}, the positive examples $\{T(2), V(4) \}$ and $\{V(1), T(4)\}$ are separated from the negative $\{T(1)\}$ and $\{ V(4)\}$
%
%
by the branching $\mathcal{Q}[\Diamond]$-query $\Diamond T \land \Diamond V$ while no path query is capable of doing this. \hfill $\dashv$
\end{example}

Branching $\mathcal{Q}[\nxt,\Diamond]$-queries express transparent existential conditions and can be regarded as \LTL{}  CQs. However, branching $\mathcal{Q}[\U]$-queries with nestings of $\U$ on the left-hand side correspond to complex first-order formulas with multiple alternations of quantifiers $\exists$ and $\forall$, which are hard to comprehend. So we also consider the language $\mathcal{Q}[\Us]\supseteq \mathcal{Q}_{p}[\U]$ of `simple' $\mathcal{Q}[\U]$-queries without such nestings.



%
%
%

In this paper, we take the first steps towards understanding the complexity and especially feasibility of the query-by-example problems $\mathsf{QBE}(\mathcal{L},\mathcal{Q})$ with $\mathcal{L}$ an ontology and $\mathcal{Q}$ a query language.
We are particularly interested in whether there is a difference in complexity between path and branching queries and whether it can be reduced by bounding the number of positive or negative examples. Our results in the ontology-free case
\vspace*{-1mm}
\begin{table}[h]
\centering
\begin{tabular}{c|c|c|c}
		\hline
		$\QBE$ for & $\mathsf{b}+, \mathsf{b}-$ & $\mathsf{b}+$ & $\mathsf{b}-$ or unbounded \\
\hline
$\mathcal{Q}_{p}[\Diamond]/\mathcal{Q}_{p}[\nxt,\Diamond]$ & ${\small \le} \PTime$ & ${\small =} \NP$ & ${\small =} \NP$    \\
$\mathcal{Q}[\Diamond]/\mathcal{Q}[\nxt,\Diamond]$ & ${\small \le} \PTime$ & ${\small \le} \PTime$ & ${\small =} \NP$ \\\hline
$\Qp[\U]$ & \multicolumn{3}{c}{${\small =} \NP$} \\\hline
$\mathcal{Q}[\Us]$ & ${\small \le} \PTime$ & ${\small \le} \PTime$ & ${\small \ge} \NP$, ${\small \le} \PSpace$  \\\hline
$\mathcal{Q}[\U]$ & \multicolumn{3}{c}{${\small \le} \PSpace$} \\ \hline
\end{tabular}%
\caption{Complexity in the ontology-free case.}
\label{table:free}
\end{table}
\vspace*{-2mm}
\noindent
are summarised in Table~\ref{table:free}, where $\mathsf{b}+$\,/\,$\mathsf{b}-$ indicate that the number of positive\,/\,negative examples is bounded\footnote{We do not consider queries with $\nxt$ only as separability is trivially in \PTime{} and does not detect any useful patterns.}\!\!. Note that path queries are indeed harder than branching ones when the number of positive examples is bounded but not in the unbounded case.  Our proof techniques range from reductions to common subsequence existence problems~\cite{DBLP:journals/jacm/Maier78,DBLP:journals/tcs/Fraser96} and dynamic programming to mimicking separability by path and branching $\U$-queries in terms of containment and simulation of  transition systems~\cite{DBLP:conf/cav/KupfermanV96a}. The key to NP upper bounds is the \emph{polynomial separation property} (PSP) of the respective languages: any separable example set is separated by a polynomial-size query.
The complexity for $\mathcal{Q}_{p}[\Diamond]$, $\mathcal{Q}[\Diamond]$ can also be obtained from~\cite{DBLP:conf/icgi/FijalkowL21} who studied separability by $\mathcal{Q}[\Diamond]$-queries of bounded size.



In the presence of ontologies, we distinguish between the combined complexity of $\mathsf{QBE}(\mathcal{L},\mathcal{Q})$, when both data and ontology are regarded as input, and the data complexity, when the ontology is deemed fixed or negligibly small compared with the data.
%
We obtain encouraging results:  $\mathcal{Q}_{p}[\Diamond]$- and $\mathcal{Q}[\Diamond]$-queries  mediated by $\LTL\Xbd$-ontologies and all of our  queries mediated by $\LTL_\horn\Xallop$-ontologies enjoy the same data complexity as in Table~\ref{table:free}.
%
%
The combined complexity results for queries with $\LTL_\horn\Xallop$-ontologies we have obtained so far  are given in Table~\ref{table:Horn}. Interestingly, QBE for query classes with $\Diamond$ and $\nxt$ only is \PSpace-complete---
\begin{table}[h]
\centering
\begin{tabular}{c|c}\hline
$\mathcal{Q}[\Diamond]$\,/\,$\Qp[\Diamond]$ & \multirow{2}{*}{= \PSpace}  \\
$\mathcal{Q}[\nxt, \Diamond]$\,/\,$\Qp[\nxt, \Diamond]$ &  \\
$\mathcal{Q}[\Us]$ & $\ge \PSpace$, $\leq \textsc{ExpTime}$ \\
$\Qp[\U]$ & $\geq \NExpTime, \leq \ExpSpace$ \\
$\mathcal{Q}[\U]$ & $\ge \PSpace$, $\le 2\textsc{ExpTime}$ \\\hline
\end{tabular}
\caption{Combined complexity of $\QBE(\LTL_\horn\Xallop, \mathcal{Q})$ in both bounded and unbounded cases.}
\label{table:Horn}
\end{table}
\noindent
not harder than satisfiability. The upper bound is proved by establishing the \emph{exponential separation property} for all of these classes of queries and using the canonical (aka minimal)  model property of Horn \LTL. The upper bounds for $\U$-queries are by reduction to the simulation and containment problems for exponential-size transition systems. For arbitrary \LTL-ontologies, this technique only gives a $2\ExpTime{}$ upper bound for $\mathcal{Q}[\Us]$ and a $2\ExpSpace{}$ one for $\Qp[\U]$.
Separability by (path) $\Diamond$-queries under $\LTL\Xbd$ ontologies turns out to be $\Sigma_{2}^{p}$-complete, where the upper bound is shown by establishing the PSP.

Compared with non-temporal QBE, our results are very encouraging: QBE is \textsc{coNExpTime}-complete for conjunctive queries (CQs) over standard relational databases \cite{DBLP:conf/cp/Willard10,DBLP:conf/icdt/CateD15} and even undecidable for CQs under $\mathcal{ELI}$ or $\mathcal{ALC}$ ontologies~\cite{DBLP:conf/ijcai/FunkJLPW19,AAAI20}.





\section{Further Related Work}

We now briefly comment on a few other related research areas.
One of them is concept learning in description logic (DL), as  proposed by~\cite{DBLP:conf/ilp/BadeaN00} who, inspired by inductive logic
programming, used refinement operators to construct a concept separating
positive and negative examples in a DL ABox. 
There has been significant interest in this approach~\cite{DBLP:conf/ilp/LehmannH09,DBLP:journals/ml/LehmannH10,Lisi15,DBLP:conf/aaai/SarkerH19,DBLP:conf/dlog/Lisi12,DBLP:journals/fgcs/RizzoFd20}.
Prominent systems include the {\sc DL Learner}~\cite{DBLP:journals/ws/BuhmannLW16}, {\sc
DL-Foil}~\cite{DBLP:conf/ekaw/Fanizzi0dE18} and its extension {\sc DL-Focl} \cite{DBLP:conf/ekaw/0001FdE18}, SPaCEL~\cite{DBLP:journals/jmlr/TranDGM17}, {\sc
YinYang}~\cite{DBLP:journals/apin/IannonePF07}, {\sc pFOIL-DL}~\cite{Straccia15}, and {\sc EvoLearner}~\cite{DBLP:conf/www/HeindorfBDWGDN22}. However, this work has not considered the complexity of separability. Also closely related is the work on the separability of two formal (e.g., regular) languages  using a weaker (e.g., FO-definable) language~\cite{DBLP:journals/corr/PlaceZ14,DBLP:conf/icdt/HofmanM15,DBLP:conf/fsttcs/PlaceZ22}. When translated into a logical separability problem, the main difference to our results is that one demands $\TO,\mathcal{D}\models\neg \q(\avec{d})$---and not just $\TO,\mathcal{D}\not\models\q(\avec{d})$---for all $(\mathcal{D},\avec{d})\in E^{-}$.



\section{Preliminaries}\label{prelims}

\LTL-\emph{formulas} are built from \emph{atoms} $A_i$, $i < \omega$, using the Booleans and  (future-time) temporal operators  $\nxt$, $\Diamond$, $\Box$, $\U$, which we interpret under the \emph{strict semantics}~\cite{gkwz,DBLP:books/cu/Demri2016}. An \LTL-\emph{interpretation} $\I$ identifies those atoms $A_i$ that are \emph{true} at each time instant $n \in \mathbb N$, written $\I, n \models A_i$. The truth-relation for atoms is extended inductively to \LTL-formulas by taking  $\I,n \models \varphi \U \psi$ iff $\I,m \models \psi$, for some  $m > n$, and $\I, k \models \varphi$ for all $k \in (n,m)$, and using the standard clauses for the Booleans and equivalences $\nxt\varphi \equiv \bot \U \varphi$, $\Diamond\varphi \equiv \top \U \varphi$ and $\Box\varphi \equiv \neg\Diamond\neg \varphi$ with Boolean \emph{constants} $\bot$ and $\top$ for `false' and `true'\!.

An \emph{\LTL-ontology}, $\TO$, is any finite set of \LTL-formulas, called the \emph{axioms} of $\TO$. An interpretation $\I$ is a \emph{model} of $\TO$ if all axioms of $\TO$ are true at \emph{all times} in $\I$. 
As mentioned in the introduction, apart from full $\LTL$ we consider its Prior $\Box\Diamond$-fragment $\LTL\Xbd$ and $\LTL_\horn\Xallop$ whose axioms take the form
\begin{equation}\label{axiomH}
C_1 \land \dots \land C_k ~\to~ C_{k+1}
\end{equation}
with $C_i$ given by $C ::= A_i \mid \bot \mid \Box C  \mid  \nxt C$. In fact, we could allow $\Diamond$ on the left-hand side of~\eqref{axiomH} as $\Diamond C \to C'$ can be replaced by $\nxt C \to A$, $\nxt A \to A$, $A \to C'$ with fresh $A$.


A \emph{data instance} is a finite set $\Abox$ of atoms $A_i(\ell)$ with a \emph{time\-stamp} $\ell \in \mathbb N$; $\max\Abox$ is the maximal timestamp in $\Abox$.
We access data by means of \LTL{} analogues of conjunctive queries: our \emph{queries}, $\varkappa$, are constructed from atoms, $\bot$ and $\top$ using $\land$, $\nxt$, $\Diamond$ and $\U$.
The class of queries that only use operators from $\Phi \subseteq \{\nxt, \Diamond, \U\}$ is denoted by $\mathcal{Q}[\Phi]$; $\Qp[\Phi]$ is its subclass of \emph{path-queries}, which take the form~\eqref{dnpath} or~\eqref{upath}; and $\mathcal{Q}[\Us]$ comprises \emph{simple queries} in $\mathcal{Q}[\U]$ that do not contain subqueries $\varkappa_1 \U \varkappa_2$ with an occurrence of $\U$ in $\varkappa_1$. Note that $\Qp[\U] \subseteq \mathcal{Q}[\Us]$. The \emph{temporal depth} $\dep(\varkappa)$ of $\varkappa$ is the maximum number of nested temporal operators in $\varkappa$.

An interpretation $\I$ is a \emph{model} of a data instance $\Abox$ if \mbox{$\I,\ell \models A_i$} for all $A_i(\ell) \in \Abox$. $\TO$ and $\Abox$ are \emph{consistent} if they have a model.
We call $k \le \max\Abox$ a (\emph{certain}) \emph{answer} to the \emph{ontology-mediated query} $(\TO, \varkappa)$ over $\Abox$ and write \mbox{$\TO, \Abox \models \varkappa(k)$} if $\I,k \models \varkappa$ in all models $\I$ of $\TO$ and $\Abox$.
%


Let $\mathcal{L}$ and $\mathcal{Q}$ be an ontology and query language defined above. The \emph{query-by-example problem} $\mathsf{QBE}(\mathcal{L},\mathcal{Q})$ we are concerned with in this paper is formulated as follows:
\begin{description}
\item[given] an $\mathcal{L}$-ontology $\TO$ and an \emph{example set} $E = (E^+, E^-)$ with finite sets $E^+$ and $E^-$ of \emph{positive} and, respectively, \emph{negative} data instances,

\item[decide] whether $E$ is $\mathcal{Q}$-\emph{separable} under $\mathcal{O}$ in the sense that there is a $\mathcal{Q}$-query $\varkappa$ with $\mathcal{O}, \Abox \models \varkappa(0)$ for all $\Abox \in E^+$ and $\mathcal{O}, \Abox \not\models \varkappa(0)$ for all $\Abox \in E^-$.
\end{description}
If $\mathcal L = \emptyset$, we shorten $\mathsf{QBE}(\emptyset,\mathcal{Q})$ to $\mathsf{QBE}(\mathcal{Q})$. We also consider the $\mathsf{QBE}$ problems with the input example sets having a \emph{bounded} number of positive and/or negative examples, denoted
$\QBEbp(\mathcal{L},\mathcal{Q})$, $\QBEbm(\mathcal{L},\mathcal{Q})$, or $\QBEba(\mathcal{L},\mathcal{Q})$. Notations like $\QBE^{2\text{+}}_{1\text{--}}(\mathcal{L},\mathcal{Q})$ should be self-explanatory. The size of $\TO$, $E$, $\varkappa$, denoted $|\TO$, $|E|$, $|\varkappa|$, respectively, is the number of symbols in it with the timestamps given in unary.


The next example illustrates the definitions and relative expressive power of queries with different temporal operators.

\begin{example}\em
(a) Let $E = (\{\Abox_{1}\},\{\Abox_{2}\})$ with $\Abox_{1}= \{A(1)\}$, $\Abox_{2}=\{A(2)\}$.  Then $\nxt A$ separates $E$ but no $\mathcal{Q}[\Diamond]$-query does. $E$ is not separable under $\TO = \{ \nxt A \to A \}$ by any query $\varkappa$ as $\TO,\Abox_1 \models \varkappa(0)$ implies $\TO,\Abox_2 \models \varkappa(0)$.
	
	
(b) Let $E = (\{ \Abox_1,\Abox_2\}, \{\Abox_3\})$ with $\Abox_1 = \{A(1),B(2)\}$, $\Abox_2 = \{A(2),B(3)\}$, $\Abox_3 = \{A(3),B(5)\}$. Then the query $\Diamond(A \wedge \nxt B)$ separates $E$ but no query in $\mathcal{Q}[\Diamond]$ does.
	
(c) $A \U B$ separates $(\{ \{B(1)\}, \{A(1),B(2)\} \}, \{ \{B(2)\}\})$ but no $\mathcal{Q}[\nxt,\Diamond]$-query does. \hfill $\dashv$
\end{example}
We now establish a few important \emph{polynomial-time reduc\-tions}, $\le_p$, among the $\QBE$-problems for various query classes, including $\mathcal{Q}_{p}^{\circ}[\Diamond]$-\emph{queries} of the form
\begin{align}\label{normF}
\varkappa = \rho_0 \land \Diamond (\rho_1 \land \Diamond (\rho_2 \land \dots \land \Diamond \rho_n) ),
\end{align}
where each $\rho_{i}$ is a $\mathcal{Q}_p[\nxt]$-query  (i.e., $\Diamond$-free $\mathcal{Q}_p[\nxt,\!\Diamond]$-query).

\begin{theorem}\label{th:data-reduction}
The following polynomial-time reductions hold\textup{:}
\begin{description}
\item[(\emph{i}.1)] $\QBE(\mathcal{L},\mathcal{Q}) \le_p \QBE_{1\text{--}}(\mathcal{L},\mathcal{Q})$, for any $\mathcal{Q}$ closed under~$\land$, and any $\mathcal{L}$ \textup{(}including $\mathcal{L} = \emptyset$\textup{)},

\item[(\emph{i}.2)] $\QBE(\mathcal{L},\mathcal{Q}) \le_p \QBE^{2\text{+}}(\mathcal{L},\mathcal{Q})$, for $\mathcal{L} \in \{\LTL, \LTL\Xbd\}$,

\item[(\emph{i}.3)] $\QBE(\mathcal{L},\mathcal Q[\nxt,\Diamond]) \le_p \QBE(\mathcal{L},\mathcal{Q}_{p}^{\circ}[\Diamond]))$ and\\ $\mathsf{QBE}(\mathcal{L},\mathcal{Q}[\Diamond]) \le_p \mathsf{QBE}(\mathcal{L},\mathcal{Q}_{p}[\Diamond])$, for any $\mathcal{L}$,

\item[(\emph{ii}.1)] $\QBE(\Qp[\Diamond]) \le_p \QBE(\Qp[\nxt,\Diamond])$ and\\ $\QBE(\Qp[\Diamond]) \le_p \QBE(\Qp[\U]) \le_p \QBE_{1\text{--}}(\Qp[\U])$,

\item[(\emph{ii}.2)] $\mathsf{QBE}(\mathcal{Q}[\nxt,\Diamond]) =_p \QBE(\mathcal{Q}[\Diamond]) \le_p \QBE(\mathcal Q[\Us])$.
\end{description}
Reductions $(i.1)$--$(i.3)$ work for combined complexity\textup{;} $(i.1)$, $(i.3)$ also work for data complexity. The reductions preserve boundedness of the number of positive/negative examples.
\end{theorem}
\begin{proof}
In (\emph{i}.1), $(E^+,E^-)$ with $E^- = \{\Abox_1,\dots,\Abox_n\}$ is $\mathcal{Q}$-separable under $\TO$ iff each $(E^+,\{\Abox_i\})$ is because if $\varkappa_i$ separates $(E^+,\{\Abox_i\})$, then $\varkappa_1 \land \dots \land \varkappa_n$ separates $(E^+,E^-)$.

In (\emph{i}.2), $(E^+,E^-)$ with $E^+=\{\Abox_1,\dots,\Abox_n\}$, $n>1$, is $\mathcal{Q}$-separable under $\TO$ iff $(E'^+, E^-)$ is $\mathcal{Q}$-separable under $\TO'$ that extends $\TO$ with the following axioms simulating $E^+$:
\begin{align*}
& S_1\to A_1\lor\dots \lor A_n, \quad S_2\to A_1\lor\dots \lor A_n,\\
&C_i\land\Diamond A_j\to X, \quad D_i\land\Diamond A_j\to X, \ \ \text{for $X(i)\in\Abox_j$,}
\end{align*}
where $S_1$, $S_2$, $A_k$, $C_l$, $D_l$, for $l \le n' =\max_i\max\Abox_i$, are fresh  and $E'^+$ consists of $\{C_0(0),\dots,C_k(n'),S_1(n'+1)\}$ and $\{D_0(0),\dots,D_k(n'),S_2(n'+1)\}$.

(\emph{i}.3) Using $[\rho_{0} \wedge \Diamond(\rho_{1} \wedge \bigwedge_{i}\Diamond \varkappa_{i})] \equiv [\rho_{0} \wedge \bigwedge_{i} \Diamond(\rho_{1} \wedge \Diamond \varkappa_{i})]$, $\nxt \Diamond \varkappa \equiv \Diamond\nxt \varkappa$ and $\nxt (\varkappa \land \varkappa') \equiv (\nxt \varkappa \land \nxt \varkappa')$ we convert, in polytime, each $\mathcal{Q}[\nxt,\Diamond]$-query to an equivalent conjunction of $\mathcal{Q}_{p}^{\circ}[\Diamond]$-queries. Thus, there is $\q\in \mathcal{Q}[\nxt,\Diamond]$ separating $(E^{+},E^{-})$ iff there are polysize $\q_{\mathcal{D}}\in\mathcal{Q}_{p}^{\circ}[\Diamond]$ separating $(E^{+},\{\mathcal{D}\})$, for each $\mathcal{D}\in E^{-}$.

(\emph{ii}.1) The first two reductions are shown by adding to $E^+ \ni \Abox$, for some $\Abox$, the data instance $\mathcal{D}'=\{ A(mn) \mid  A(n)\in \mathcal{D}\}$ with $m=\max\mathcal{D}+2$. Now, if $\Abox\models \varkappa(0)$ and $\Abox'\models \varkappa(0)$, for $\varkappa \in \Qp[\U]$, then $\varkappa$ is equivalent to a $\Qp[\Diamond]$-query.
The third reduction, illustrated below for $E^+= \{\Abox^+_1, \Abox^+_2\}$ and $E^- = \{\Abox^-_1, \Abox^-_2\}$, transforms $E$ into two positive and \emph{one negative} example using `pads' of  fresh atoms $B$, $C$. We show \\
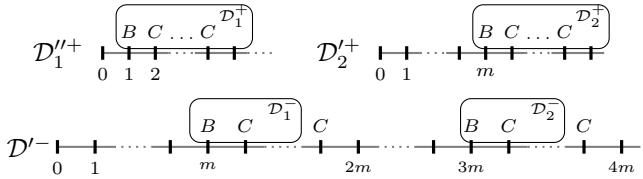
\centerline{
\begin{tikzpicture}[nd/.style={draw,thick,circle,inner sep=0pt,minimum size=1.5mm,fill=white},xscale=0.35]
\node[label=left:{$\Abox''^+_1$}] at (0,0) {};
\node   at (1.2,0.3) (x1)  {};
\node at (5,0.5) {\tiny $\Abox_1^+$};
\node  at (5,0.4) (x2) {};
\node[fit = (x1)(x2), basic box = black]  {};
\draw[thick,gray,-] (0,0) -- (2.5,0);
\draw[thick,gray,dotted] (2.5,0) -- (3.5,0);
\draw[thick,gray,-] (3.5,0) -- (5.5,0);
\draw[thick,gray,dotted] (5.5,0) -- (6.5,0);
\slin{0,0}{}{$0$};
\slin{1,0}{$B$}{$1$};
\slin{2,0}{$C$}{$2$}
\node at (3,0.2) {\tiny$\ldots$};
\slin{4,0}{$C$}{};
\slin{5,0}{}{};
\end{tikzpicture}
\quad
\begin{tikzpicture}[nd/.style={draw,thick,circle,inner sep=0pt,minimum size=1.5mm,fill=white},xscale=0.35]
\node[label=left:{$\Abox'^+_2$}] at (0,0) {};
\node   at (4.2,0.3) (x1)  {};
\node at (8,0.5) (D) {\tiny $\Abox_2^+$};
\node  at (8,0.4) (x2) {};
\node[fit = (x1)(x2), basic box = black]  {};
\draw[thick,gray,-] (0,0) -- (1.5,0);
\draw[thick,gray,dotted] (1.5,0) -- (2.5,0);
\draw[thick,gray,-] (2.5,0) -- (5.5,0);
\draw[thick,gray,dotted] (5.5,0) -- (6.5,0);
\draw[thick,gray,-] (6.5,0) -- (8.5,0);
\slin{0,0}{}{$0$};
\slin{1,0}{}{$1$};
\slin{3,0}{}{};
\slin{4,0}{$B$}{$m$};
\slin{5,0}{$C$}{};
\node at (6,0.2) {\tiny$\ldots$};
\slin{7,0}{$C$}{};
\slin{8,0}{}{};
\end{tikzpicture}
}

\centerline{
\begin{tikzpicture}[nd/.style={draw,thick,circle,inner sep=0pt,minimum size=1.5mm,fill=white},xscale=0.5]
\node[label=left:{$\Abox'^-$}\!\!\!] at (0,0) {};
\node   at (4,0.3) (x1)  {};
\node at (6,0.5) {\tiny $\Abox_1^-$};
\node  at (6,0.4) (x2) {};
\node at (11.2,0.3) (x3)  {};
\node at (13,0.5) {\tiny $\Abox^-_2$};
\node  at (13,0.4) (x4) {};
\node[fit = (x1)(x2), basic box = black]  {};
\node[fit = (x3)(x4), basic box = black]  {};
\draw[thick,gray,-] (0,0) -- (1.5,0);
\draw[thick,gray,dotted] (1.5,0) -- (2.5,0);
\draw[thick,gray,-] (2.5,0) -- (5.5,0);
\draw[thick,gray,dotted] (5.5,0) -- (6.5,0);
\draw[thick,gray,-] (6.5,0) -- (8.5,0);
\draw[thick,gray,dotted] (8.5,0) -- (9.5,0);
\draw[thick,gray,-] (9.5,0) -- (12.5,0);
\draw[thick,gray,dotted] (12.5,0) -- (13.5,0);
\draw[thick,gray,-] (13.5,0) -- (15.5,0);
\slin{0,0}{}{$0$};
\slin{1,0}{}{$1$};
\slin{3,0}{}{};
\slin{4,0}{$B$}{\tiny$m$};
\slin{5,0}{$C$}{};
\slin{7,0}{$C$}{};
\slin{8,0}{}{\tiny$2m$};
\slin{10,0}{}{};
\slin{11,0}{$B$}{\tiny$3m$};
\slin{12,0}{$C$}{};
\slin{14,0}{$C$}{};
\slin{15,0}{}{\tiny$4m$};
\end{tikzpicture}
}

\noindent
that $E$ is $\Qp[\U]$-separable iff $(\{\Abox''^+_1,\Abox'^+_2\}, \{\Abox'^-\})$ is.

(\emph{ii}.2) The first reduction is established by modifying every $\Abox$ in the given $E$ as illustrated below using fresh atoms $A_i$ and $B_j$ that encode $\nxt^iA$ and $\nxt^j B$, respectively:
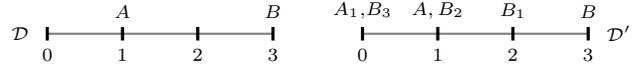
\centerline{
\begin{tikzpicture}[nd/.style={draw,thick,circle,inner sep=0pt,minimum size=1.5mm,fill=white},xscale=1]
\node[label=left:{\scriptsize$\Abox$}] at (0,0) {};
\draw[thick,gray,-] (0,0) -- (3,0);
\slin{0,0}{}{\scriptsize$0$};
\slin{1,0}{\scriptsize$A$}{\scriptsize$1$};
\slin{2,0}{}{\scriptsize$2$};
\slin{3,0}{\scriptsize$B$}{\scriptsize$3$};
\end{tikzpicture}
\quad \
\begin{tikzpicture}[nd/.style={draw,thick,circle,inner sep=0pt,minimum size=1.5mm,fill=white},xscale=1]
\draw[thick,gray,-] (0,0) -- (3,0);
\slin{0,0}{\scriptsize$A_1,\!B_3$}{\scriptsize$0$};
\slin{1,0}{\scriptsize$A,B_2$}{\scriptsize$1$};
\slin{2,0}{\scriptsize$B_1$}{\scriptsize$2$};
\slin{3,0}{\scriptsize$B$}{\scriptsize$3$};
\node[label=right:{\scriptsize$\Abox'$}] at (3,0) {};
\end{tikzpicture}
}\\
Then $E$ is $\mathcal{Q}[\nxt,\!\Diamond]$-separable iff $E'$ is $\mathcal{Q}[\Diamond]$-separable. The converse and the second reduction are similar to (\emph{ii}.1).
\end{proof}




%

%

\section{QBE without Ontologies}\label{sec:no-ont}

We start investigating the complexity of the $\QBE$ problems for \LTL{} by considering queries without mediating ontologies.

\begin{theorem}\label{thm:qbewithout}
The $\QBE$-problems for the classes of queries defined above \textup{(}with the empty ontology\textup{)} belong to the complexity classes shown in Table~\ref{table:free}.
%
\end{theorem}

We comment on the proof in the remainder of this section.

\smallskip
\noindent
\textbf{$\nxt\Diamond$-queries.}
\NP-hardness is established by reduction of the consistent subsequence existence problems~\cite[Theorems 2.1, 2.2]{DBLP:journals/tcs/Fraser96} in tandem with Theorem~\ref{th:data-reduction}; membership in NP follows from the fact that separating queries, if any, can always be taken of polynomial size.

Tractability is shown using dynamic programming. We explain the idea for $\QBEba(\mathcal{Q}_{p}[\nxt,\Diamond])$,  $E^{+}=\{\mathcal{D}_{1}^{+},\mathcal{D}_{2}^{+}\}$ and
$E^{-}=\{\mathcal{D}_{1}^{-},\mathcal{D}_{2}^{-}\}$. Suppose $\varkappa$ takes the form~\eqref{dnpath} with $\rho_{n}\ne \top$.
Then $\mathcal{D}\models \varkappa(0)$ iff there is a strictly monotone map $f \colon [0,n] \rightarrow [0,\max \mathcal{D}]$ with $f(0)=0$, $f(i+1)=f(i)+1$
if $\op_{i}=\nxt$, and $\rho_{i}\subseteq t_{\mathcal{D}}(f(i)) = \{A \mid A(f(i)) \in \Abox\}$.
We call such an $f$ a \emph{satisfying assignment} for $\varkappa$ in $\mathcal{D}$.
Let $S_{i,j}$ be the set of tuples $(k,\ell_1,\ell_2,n_{1},n_{2})$ such that $\ell_{1} \leq i\leq \max \mathcal{D}_{1}^{+}$, \mbox{$\ell_{2} \leq j\leq \max \mathcal{D}_{2}^{+}$}, and there is
$
\varkappa=\rho_0 \land \op_1 (\rho_1 \land \dots \land \op_k \rho_k)
$
for which $(i)$ there are satisfying assignments $f_{1},f_{2}$ in $\mathcal{D}_{1}^{+}$ and $\mathcal{D}_{2}^{+}$ with $f_{1}(k)=\ell_{1}$ and $f_{2}(k)=\ell_{2}$, respectively, and $(ii)$ $n_{1}$ is minimal with a satisfying assignment $f$ for $\varkappa$ in $\mathcal{D}_{1}^{-}$ having  $f(k)=n_{1}$, and $n_{1}=\infty$ if there is no such $f$; and similarly for $n_{2}$, $\mathcal{D}_{2}^{-}$.
It suffices to compute $S_{\max \mathcal{D}_{1}^{+},\max \mathcal{D}_{2}^{+}}$
in polytime. This can be done incrementally by initially observing that $S_{0,j}$ can only contain $(0,0,0,0,0)$, which is the case if there is $\rho_{0}\subseteq t_{\mathcal{D}_{1}^{+}}(0)$,
$\rho_{0}\subseteq t_{\mathcal{D}_{2}^{+}}(0)$ and
$\rho_{0}\not\subseteq t_{\mathcal{D}_{1}^{-}}(0)$,
$\rho_{0}\not\subseteq t_{\mathcal{D}_{2}^{-}}(0)$ (and similarly for $S_{i,0}$).
%

\smallskip
\noindent
\textbf{$\U$-queries.}
NP-hardness for $\Qp[\U]$, $\mathcal{Q}[\Us]$ follows from Theorem~\ref{th:data-reduction} $(ii.1)$, $(ii.2)$ and  NP-hardness for $\nxt\Diamond$-queries.

The upper bounds are shown by reduction of $\mathcal{Q}_{p}[\U]$- and $\mathcal{Q}[\Us]$-separability to the simulation and containment problems for transition systems~\cite{DBLP:conf/cav/KupfermanV96a}.
A \emph{transition system}, $S$, is a digraph each of whose nodes and edges is labelled by some set of symbols from a node or, respectively, edge alphabet; $S$ also has a designated set $S_0$ of \emph{start nodes}.
A \emph{run} of $S$ is a path in digraph $S$, starting in $S_0$, together with all of its labels.
The \emph{computation tree} of $S$ is the tree unravelling $\mathfrak T_S$ of $S$.
For systems $S$ and $S'$ over the same alphabets, we say that $S$ is \emph{contained} in $S'$ if, for every run $r$ of $S$, there is a run $r'$ of $S'$ such that $r$ and $r'$ have the same length and the labels on the states and edges in $r$ are subsumed by the corresponding labels in $r'$.
$S$ is \emph{simulated} by $S'$ if $\mathfrak T_{S}$ is \emph{finitely embeddable} into $\mathfrak T_{S'}$ in the sense that every finite subtree\footnote{A \emph{subtree} is a convex subset of $\mathfrak T_S$'s nodes with some start node.} of $\mathfrak T_{S}$ can be homomorphically mapped into $\mathfrak T_{S'}$ preserving (subsumption of) node and edge labels.

%

Now, let $E= (E^+, E^-)$ with $E^\sigma = \{\Abox_i \mid i \in I^\sigma\}$, for \mbox{$\sigma \in \{+,-\}$} and disjoint $I^+$ and $I^-$, and let  $\Sigma$ be the signature of $E$. For each $i \in I^+ \cup I^-$, we take a transition system $S^i$ with states $0^i, \dots, (\max \D_i+1)^i$, where $(\max \D_i+1)^i$ is labelled with $\emptyset$ and the remaining $j^i$ by $\{A \mid A(j) \in \D_i \}$. Transitions are $j^i \to k^i$, for $0 \leq j < k \leq \max \D_i+1$, that are labelled by $\{ A \in \Sigma \cup \{\bot\} \mid A(n) \in \D_i, n \in (j,k)  \}$ and $(\max \D_i+1)^i \to (\max \D_i+1)^i$ with label $\Sigma^\bot =\Sigma \cup \{\bot\}$. Thus, $\D_i$ shown on the left below gives rise to $S^i$ on the right:\\[-3pt]
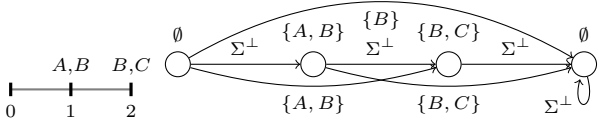
\centerline{
\begin{tikzpicture}[nd/.style={draw,thick,circle,inner sep=0pt,minimum size=1.5mm,fill=white},xscale=0.8]
\draw[thick,gray,-] (0,0) -- (2,0);
\slin{0,0}{}{$0$};
\slin{1,0}{$A,\!B$}{$1$};
\slin{2,0}{$B,\!C$}{$2$};
\end{tikzpicture}
\begin{tikzpicture}[nd/.style={draw,thick,circle,inner sep=0pt,minimum size=1.5mm,fill=white}, node distance=18mm]
\tikzset{every state/.style={minimum size=0pt}}
%
%
%
\node[state, label=above:{\scriptsize $\emptyset$}] (s0) {};
\node[state, right of = s0, label=above:{\scriptsize $\{A,B\}$}] (s1) {} edge[above, <-] node{\scriptsize $\Sigma^\bot$} (s0);
\node[state, right of = s1, label=above:{\scriptsize $\{B, C\}$}] (s2) {} edge[above, <-] node{\scriptsize $\Sigma^\bot$} (s1);
\draw[->] (s0) edge [bend right=15] node[below]  {\scriptsize $\{A, B\}$} (s2);
\node[state, right of = s2, label=above:{\scriptsize $\emptyset$}] (sink) {} edge[above, <-] node{\scriptsize $\Sigma^\bot$} (s2)
edge [bend left=15, <-] node[below]  {\scriptsize $\{B, C\}$} (s1)
edge [bend right=30, <-] node[below]  {\scriptsize $\{B\}$} (s0)
edge [loop below, <-] node[left]  {\scriptsize $\Sigma^\bot$} (s2)
;
\end{tikzpicture}
}\\
We form the direct product (synchronous composition) $\mathfrak P$ of $\{S^i \mid i \in I^+\}$, for $I^+ = \{1, \dots, l\}$, whose states are vectors $(s_1, \dots, s_l)$ of states $s_i \in S^i$, which are labelled by the intersection of the labels of $s_i$ in $S^i$, with transitions $(s_1, \dots, s_l) \to (p_1, \dots, p_l)$, if $s_i \to p_i$ in $S^i$ for all $i$, also labelled by the intersection of the component transition labels. On the other hand, we take the disjoint union $\mathfrak N$ of $S^i$, for $i \in I^-$, and establish the following separability criterion:
\begin{theorem}\label{criterionUsA}
$(i)$ $E$ is not $\Q[\Us]$-separable iff $\mathfrak P$ is simulated by $\mathfrak N$.
$(ii)$ $E$ is not $\Qp[\U]$-separable iff $\mathfrak P$ is contained in $\mathfrak N$.
\end{theorem}

\begin{example}\em \label{ex:prod-unrav}
For the example set depicted below, in which the only negative instance is on the right-hand side,\\
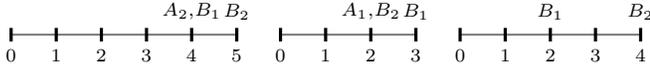
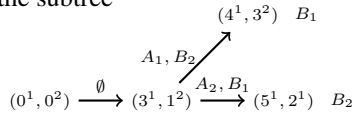
\centerline{
\begin{tikzpicture}[nd/.style={draw,thick,circle,inner sep=0pt,minimum size=1.5mm,fill=white},xscale=0.6]
\draw[thick,gray,-] (0,0) -- (5,0);
\slin{0,0}{}{$0$};
\slin{1,0}{}{$1$};
\slin{2,0}{}{$2$};
\slin{3,0}{}{$3$};
\slin{4,0}{$A_2,\!B_1$}{$4$};
\slin{5,0}{$B_2$}{$5$};
\end{tikzpicture}
\
\begin{tikzpicture}[nd/.style={draw,thick,circle,inner sep=0pt,minimum size=1.5mm,fill=white},xscale=0.6]
\draw[thick,gray,-] (0,0) -- (3,0);
\slin{0,0}{}{$0$};
\slin{1,0}{}{$1$};
\slin{2,0}{$A_1,\!B_2$}{$2$};
\slin{3,0}{$B_1$}{$3$};
\end{tikzpicture}
\
\begin{tikzpicture}[nd/.style={draw,thick,circle,inner sep=0pt,minimum size=1.5mm,fill=white},xscale=0.6]
\draw[thick,gray,-] (0,0) -- (4,0);
\slin{0,0}{}{$0$};
\slin{1,0}{}{$1$};
\slin{2,0}{\scriptsize $B_1$}{$2$};
\slin{3,0}{}{$3$};
\slin{4,0}{$B_2$}{$4$};
\end{tikzpicture}
}\\
$\mathfrak T_{\mathfrak P}$ contains the subtree \\[-7pt]
\centerline{
\begin{tikzpicture}[->,thick,node distance=1.8cm, transform shape,scale=0.9]\footnotesize
\node[rectangle] (s00) {\scriptsize$(0^1,0^2)$};
\node[rectangle, right of = s00] (s31) {\scriptsize$(3^1,1^2)$} edge[above, <-] node{\scriptsize$\emptyset$} (s00);
\node[rectangle, above right of = s31, label=right:{\scriptsize$B_1$}] (s43) {\scriptsize$(4^1,3^2)$} edge[left, <-] node{\scriptsize$A_1, B_2$} (s31);
\node[rectangle, right of = s31, label=right:{\scriptsize$B_2$}] (s52) {\scriptsize$(5^1,2^1)$} edge[above, <-] node{\scriptsize$A_2, B_1$} (s31);
\end{tikzpicture}
}\\
where only the last $\mathfrak P$-node of a $\mathfrak T_{\mathfrak P}$-node (a sequence) is indicated together with the atoms that are true at nodes and on edges. Intuitively, $\mathfrak T_{\mathfrak P}$ `represents' all possible $\Q[\Us]$-queries and its paths represent $\Qp[\U]$-queries $\varkappa$ such that $\TO, \D \models \varkappa(0)$ for all $\D \in E^+$. The $\Q[\Us]$-query given by the subtree above is $\varkappa = \Diamond \big( ((A_1 \land B_2) \U B_1) \land ((A_2 \land B_1) \U B_2)\big)$. The subtree is not embeddable into $\mathfrak T_\mathfrak{N}$ (obtained for the negative instance), so $\varkappa$ separates $E$. Observe that every path in $\mathfrak T_{\mathfrak P}$ (and in the subtree above) is embeddable into $\mathfrak T_\mathfrak{N}$.
 \end{example}

By inspecting the structure of $\mathfrak P$ and $\mathfrak N$ we observe that if $\mathfrak P$ has a run that is not embeddable into any run of $\mathfrak N$, then we can find such a run of length $\le M = \min \{ \max \Abox_i \mid i \in I^+\}$ (any longer run has $\emptyset$-labels on its states after the $M$th one). Thus, we can guess the required run and check in \PTime{} if it is correct, establishing the \NP{} upper bound for $\mathcal{Q}_{p}[\U]$. To show the \PSpace{} upper bound for $\mathcal{Q}[\Us]$, we notice that if there is a finite subtree of $\mathfrak T_{\mathfrak P}$ that is not embeddable into $\mathfrak T_{\mathfrak N}$, then the full subtree $\smash{\mathfrak T_{\mathfrak P}^M}$ of depth $M$ is not embeddable into $\mathfrak T_{\mathfrak N}$, which can be checked by constructing $\smash{\mathfrak T_{\mathfrak P}^M}$ branch-by-branch while checking all possible embeddings of these branches into $\mathfrak T_{\mathfrak N}$. Finally, we have the \PTime{} upper bound for $\mathcal{Q}[\Us]$ with a bounded number of positive examples because $\mathfrak P$ is  constructible in polytime and checking simulation between transition systems is \PTime{}-complete~\cite{DBLP:conf/cav/KupfermanV96a}. Interestingly, the smallest separating query we can construct in this case is of the same size as $\smash{\mathfrak T_{\mathfrak P}^M}$, i.e., exponential in $|E^+|$; however, we can check its existence in polytime.

The \PSpace{} upper bound for $\mathcal{Q}[\U]$ requires a more sophisticated notion of simulation between transition systems.

\begin{example}\em \label{ex:u-path-not-treeB}
The example set below, where only the rightmost instance is negative, is separated by the $\mathcal{Q}[\U]$-query\\
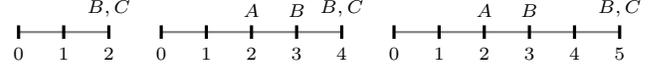
\centerline{
\begin{tikzpicture}[nd/.style={draw,thick,circle,inner sep=0pt,minimum size=1.5mm,fill=white},xscale=0.6]
\draw[thick,gray,-] (0,0) -- (2,0);
\slin{0,0}{}{\scriptsize$0$};
\slin{1,0}{}{\scriptsize$1$};
\slin{2,0}{\scriptsize$B,C$}{\scriptsize$2$};
\end{tikzpicture}
\
\begin{tikzpicture}[nd/.style={draw,thick,circle,inner sep=0pt,minimum size=1.5mm,fill=white},xscale=0.6]
\draw[thick,gray,-] (0,0) -- (4,0);
\slin{0,0}{}{\scriptsize$0$};
\slin{1,0}{}{\scriptsize$1$};
\slin{2,0}{\scriptsize$A$}{\scriptsize$2$};
\slin{3,0}{\scriptsize$B$}{\scriptsize$3$};
\slin{4,0}{\scriptsize$B,C$}{\scriptsize$4$};
\end{tikzpicture}
\
\begin{tikzpicture}[nd/.style={draw,thick,circle,inner sep=0pt,minimum size=1.5mm,fill=white},xscale=0.6]
\draw[thick,gray,-] (0,0) -- (5,0);
\slin{0,0}{}{\scriptsize$0$};
\slin{1,0}{}{\scriptsize$1$};
\slin{2,0}{\scriptsize$A$}{\scriptsize$2$};
\slin{3,0}{\scriptsize$B$}{\scriptsize$3$};
\slin{4,0}{}{\scriptsize$4$};
\slin{5,0}{\scriptsize$B,C$}{\scriptsize$5$};
\end{tikzpicture}
}\\
$(A \U B) \U C$ but is not $\mathcal{Q}[\Us]$-separable by Theorem~\ref{criterionUsA}. \hfill $\dashv$
\end{example}
%

We prove a $\mathcal{Q}[\U]$-inseparability criterion using transition systems whose non-initial/sink states are \emph{pairs} of sets of numbers, and transitions are of \emph{two} types. The  picture below shows a data instance and the induced transition system (where $z$ has incoming arrows labelled by $\Sigma^\bot$ from all states \\
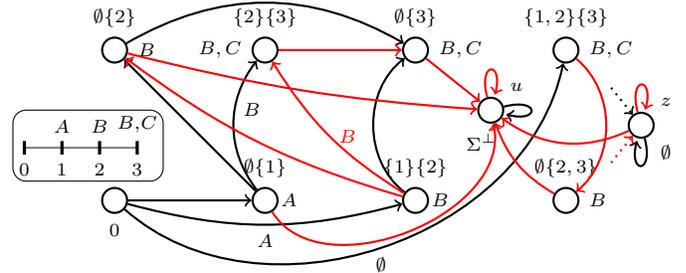
\centerline{
\begin{tikzpicture}[thick, transform shape]
\begin{scope}[nd/.style={draw,thick,circle,inner sep=0pt,minimum size=1.5mm,fill=white}, xshift = -1.2cm, yshift = .7cm, local bounding box=b]
\draw[thick,gray,-] (0,0) -- (1.5,0);
\slin{0,0}{}{\scriptsize $0$};
\slin{.5,0}{\scriptsize $A$}{\scriptsize $1$};
\slin{1,0}{\scriptsize $B$}{\scriptsize $2$};
\slin{1.5,0}{\scriptsize $B,\!C$}{\scriptsize $3$};
\end{scope}
\draw[thin, rounded corners=5pt] (b.south west) rectangle (b.north east);
\tikzset{every state/.style={minimum size=0pt}}
%
%
%
\node[state, label = below:{\scriptsize $0$}] at (0,0) (zero) {};
\node[state, label = above:{\scriptsize $\emptyset \{ 1 \}$},  label = right:{\scriptsize $\!\!A$}] at (2,0)(e1) {};
\node[state, label = above:{\scriptsize $\{ 1 \}\{ 2 \}$}, label = right:{\scriptsize $\!\!B$}] at (4,0) (one2) {};
\node[state, label = above:{\scriptsize $\emptyset\{ 2 \}$}, label = right:{\scriptsize $B$}] at (0,2) (e2) {};
\node[state, label = above:{\scriptsize $\{ 2 \}\{ 3 \}$}, label = left:{\scriptsize $B, C$}] at (2,2) (two3) {};
\node[state, label = above:{\scriptsize $\emptyset\{ 3 \}$}, label = right:{\scriptsize $B, C$}] at (4,2) (e3) {};
\node[state, label = above:{\scriptsize $\{ 1,2 \}\{ 3 \}$}, label = right:{\scriptsize $B, C$}] at (6,2) (onetwo3) {};
\node[state, label = above:{\scriptsize $\emptyset\{2, 3 \}$}, label = right:{\scriptsize $B$}] at (6,0) (e23) {};
\node[state, label = above right:{\scriptsize $u$}, label = below:{\scriptsize $\!\!\!\!\!\Sigma^\bot$}] at (5,1.2) (u) {};
\node[state, label = above right:{\scriptsize $z$}, label = below right:{\scriptsize $\emptyset$}] at (7,1) (z) {};
%
%
\path (zero) edge[->] (e1);
\path (zero) edge[->, bend right=15] node[below] {\scriptsize$A$}(one2);
\path (zero) edge[->, bend right=60] node[below] {\scriptsize$\emptyset$}(onetwo3);
%
%
\path (e1) edge[->] (e2);
\path (e1) edge[->, bend left=38] node[above right] {\scriptsize$B$}(two3);
%
%
%
\path (e2) edge[->, bend left=30] (e3);
\path (one2) edge[->, bend left=53] (e3);
\path[red] (two3) edge[->] (e3);
\path[red] (one2) edge[->, bend left=13] node[right] {\scriptsize$B$}(two3);
\path[red] (onetwo3) edge[->, bend left=50] (e23);
\path[red] (one2) edge[->, bend left=10] (e2);
\path[dotted] (6.6,1.5) edge[->] (z);
\path[dotted, red] (6.6,0.5) edge[->] (z);
\path[red] (e2) edge[->, bend right=5] (u);
\path[red] (e3) edge[->, bend right=0] (u);
\path[red] (e1) edge[->, bend right=80] (u);
\path[red] (e23) edge[->, bend left=20] (u);
\path[red] (u) edge[->, loop above] (u);
\path (u) edge[->, loop right] (u);
\path[red] (z) edge[->, bend left=30] (u);
\path[red] (z) edge[->, loop above] (z);
\path (z) edge[->, loop below] (z);
\end{tikzpicture}
}\\[-6mm]
but $u$, which are all omitted).
Each arrow from $0$ leads to a state $\{1, \dots, n-1 \}\{n\}$; it represents a formula $\varphi \U \psi$ that is true at $0$, with the arrow label indicating the non-nested atoms of $\varphi$ and the state label indicating the atoms of $\psi$. Each black (resp., red) arrow from $\avec{s}_1\avec{s}_2$ to $\avec{s}'_1\avec{s}'_2$ represents a $\U$-formula $\alpha_{\avec{s}_2 \to \avec{s}'_1\avec{s}'_2}$ (resp., $\alpha_{\avec{s}_1 \to \avec{s}'_1\avec{s}'_2}$) that is true at all points in $\avec{s}_2$ (resp., $\avec{s}_1$).
The black and red transitions are arranged in such a way that a transition from $\avec{s}''_1\avec{s}''_2$ to $\avec{s}_1\avec{s}_2$ with an arrow label $\lambda$ and $\avec{s}_1\avec{s}_2$-label $\mu$ represents the $\U$-formula $(\lambda \land \bigwedge \alpha_{\avec{s}_1 \to \avec{s}'_1\avec{s}'_2}) \U (\mu \land \bigwedge \alpha_{\avec{s}_2 \to \avec{s}'_1\avec{s}'_2})$ and similarly for the transitions from $0$. A version of Theorem~\ref{criterionUsA} for $\mathcal{Q}[\U]$ and a \PSpace-algorithm are given in the full paper.

%
%
%
%

\section{QBE with $\LTL_{\smash{\horn}}\Xallop$-Ontologies}
\label{sec:horn}
Recall from~\cite{DBLP:journals/ai/ArtaleKKRWZ21} that, for any $\LTL_\horn\Xallop$-ontology $\TO$ and data instance $\Abox$ consistent with $\TO$, there is a \emph{canonical model} $\mathcal C_{\TO, \Abox}$ of $\TO$ and $\Abox$ such that, for any query $\varkappa$ and any $k \in \mathbb N$, we have
$\TO, \Abox \models \varkappa(k)$ iff $\mathcal C_{\TO, \Abox} \models \varkappa(k)$.

Let $\sub_{\TO}$ be the set of subformulas of the $C_i$ in the axioms~\eqref{axiomH} of $\TO$ and their negations. A \emph{type for} $\TO$ is any maximal subset $\tp \subseteq \sub_{\TO}$ consistent with $\TO$. Let $\Type$ be the set  of all types for $\TO$. Given an interpretation $\I$, we denote by $\tp_{\I}(n)$ the type for $\TO$ that holds at $n \in \mathbb N$ in $\I$.
For $\TO$ consistent with $\Abox$, we abbreviate $\tp_{\C_{\TO,\Abox}}$ to $\tp_{\TO,\Abox}$. The   canonical models have a periodic structure in the following sense:
\begin{proposition}\label{prop:period}
For any $\LTL_{\smash{\horn}}\Xallop$ ontology $\TO$ and any data instance  $\Abox$ consistent with $\TO$, there are $s_{\TO,\Abox} \le 2^{|\TO|}$ and $p_{\TO,\Abox} \le 2^{2|\TO|}$ such that $\tp_{\TO,\Abox}(n) = \tp_{\TO,\Abox}(n + p_{\TO,\Abox})$, for all $n \ge \max \Abox + s_{\TO,\Abox}$.
%
%
Deciding $\mathcal C_{\TO, \Abox} \models \xi(\ell)$, for a binary $\ell$ and a conjunction of atoms $\xi$, is in $\PSpace$\,/\,$\PTime$ for combined\,/\,data complexity.
\end{proposition}

We now show that the combined complexity of $\QBE$ with $\Diamond$- and $\nxt,\Diamond$-queries is \PSpace-complete in both bounded and unbounded cases, i.e., as complex as $\LTL_{\smash{\horn}}\Xallop$ reasoning.

\begin{theorem}\label{thm:Horn-diam}
Let $\mathcal{Q}\in \{\mathcal{Q}[\nxt,\Diamond],\mathcal{Q}[\Diamond],\mathcal{Q}_{p}[\nxt, \Diamond],\mathcal{Q}_{p}[\Diamond]\}$. Then $\mathsf{QBE}(\LTL_{\smash{\horn}}\Xallop,\mathcal{Q})$ and $\QBEba\!(\LTL_{\smash{\horn}}\Xallop,\mathcal{Q})$ are both \PSpace-complete for combined complexity.
\end{theorem}
\begin{proof}
\PSpace-hardness is inherited from that of $\LTL_{\smash{\horn}}\Xallop$. We briefly sketch the proof of the matching upper bound for $\mathcal{Q}[\nxt,\Diamond]$ using the reduction of Theorem~\ref{th:data-reduction} $(i.3)$. We can assume that $\TO$ and $\D$ are consistent for any $\D \in E^+ \cup E^-$. For if $\TO$ and $\D \in E^-$ are inconsistent, then $E$ is not $\Q$-separable under $\TO$ as $\TO, \D \models \varkappa(0)$ for all $\varkappa \in \Q$; if $\TO$ and $\D \in E^+$ are inconsistent, then $E$ is separable iff $(E^+\setminus \{\D\}, E^-)$ is. Checking consistency is known to be $\PSpace{}$-complete.

Given an $\LTL_{\smash{\horn}}\Xallop$-ontology $\TO$ and an example set $E$, let
$$
k = \max_{\mathcal{D} \in E^+ \cup E^-} (\max \mathcal{D} + s_{\TO,\mathcal{D}}), \quad \textstyle m = \prod_{\mathcal{D} \in E^+ \cup E^-} p_{\TO,\mathcal{D}},
$$
where $s_{\TO,\mathcal{D}}$ and $p_{\TO,\mathcal{D}}$ in $\mathcal{C}_{\TO,\mathcal{D}}$ are from  Proposition~\ref{prop:period}. We show that if $E$ is $\mathcal{Q}[\nxt,\Diamond]$-separable under $\TO$, then it is separated by a conjunction of $|E^-|$-many $\varkappa\in\mathcal{Q}_{p}^{\circ}[\Diamond]$ of $\Diamond$-depth $\leq k+1$ and $\nxt$-depth $\leq k+m$ in~\eqref{normF}. Indeed, in this case any $(E^{+},\{\mathcal{D}\})$, for $\mathcal{D}\in E^{-}$, is separated under $\TO$ by some $\varkappa$ of the form~\eqref{normF} with the $\rho_l$ of $\nxt$-depth $\le k+m$ because $\rho_l = \bigwedge_{i=0}^\ell \nxt^i \lambda_i$ with $\ell > k+m$ can be replaced by
$$
\textstyle\bigwedge_{i=0}^k \nxt^i \lambda_i \land \bigwedge_{j=1}^m \nxt^{k+j} 
\bigwedge_{i \le \ell, j=(i-k) \ \text{mod} \ m} \lambda_i.
$$
In addition, if $n > k$ in~\eqref{normF}, then $(E^{+},\{\mathcal{D}^{-}\})$ is separated by
$$
\textstyle\rho_0 \land \Diamond (\rho_1 \land \Diamond (\rho_2 \land \dots \land \Diamond \rho_k)) \land \bigwedge_{i=k+1}^n \Diamond^{k+1} \rho_i ,
$$
and so by some
$\rho_0 \land \Diamond (\rho_1 \land \dots \land \Diamond (\rho_k \wedge \Diamond \rho_{j}))$ with  \mbox{$k<j\leq n$}. Our nondeterministic \PSpace-algorithm incrementally guesses the $\rho_l$ and checks if they are satisfiable in the relevant part of the relevant $\mathcal{C}_{\mathcal{O},\mathcal{D}}$ bounded by $k + m$.
\end{proof}

The situation is quite different for queries with $\U$:

\begin{theorem}\label{thm:Horn-U}
$\mathsf{QBE}(\LTL_\horn\Xallop, \mathcal{Q}[\Us])$ is in \ExpTime{} for combined complexity,
$\mathsf{QBE}(\LTL_\horn\Xallop, \mathcal{Q}_{p}[\U])$ is in \ExpSpace{}, and
$\smash\QBEba(\LTL_\horn\Xallop, \mathcal{Q}_{p}[\U])$ is \NExpTime{}-hard.
\end{theorem}
\begin{proof}
For the upper bounds, we again assume that $\TO$ and $\D$ are consistent for all $\D \in E^+ \cup E^-$.
Observe that Theorem~\ref{criterionUsA} continues to hold in the presence of $\LTL_\horn\Xallop$ ontologies $\TO$ but we need a different construction of transition systems $S^i$ that represent all $\mathcal{Q}[\Us]$-queries mediated by $\TO$ over $\D_i$. We illustrate it for $\TO = \{ A \to C \land \nxt B,\, B \to \nxt^2 B,\, B \to \nxt C \}$ and $\D_i = \{ A(0)\}$  below, where the picture on the left shows the canonical model of $\TO, \D_i$ (see Proposition~\ref{prop:period}) and next to it is $S^i$ (the omitted labels on transitions are $\Sigma^\bot$).

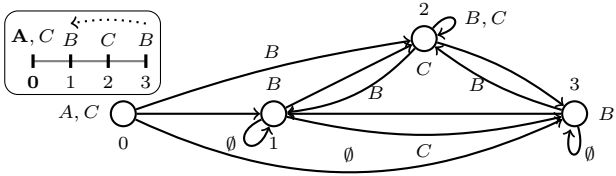
\centerline{
\begin{tikzpicture}[thick, transform shape]
\begin{scope}[nd/.style={draw,thick,circle,inner sep=0pt,minimum size=1.5mm,fill=white}, xshift = -1.2cm, yshift = .7cm, local bounding box=b]
\draw[thick,gray,-] (0,0) -- (1.5,0);
\slin{0,0}{\scriptsize
$\mathbf{A}, C$
}{\scriptsize $\mathbf{0}$};
\slin{.5,0}{\scriptsize $B$}{\scriptsize $1$};
\slin{1,0}{\scriptsize $C$}{\scriptsize $2$};
\slin{1.5,0}{\scriptsize $B$}{\scriptsize $3$};
\path(1.5,.5) edge[dotted,->, bend right=10] node{} (.5,.5);
\end{scope}
\draw[thin, rounded corners=5pt] (b.south west) rectangle (b.north east);
\tikzset{every state/.style={minimum size=0pt}}
%
%
%
\node[state, label = below:{\scriptsize $0$},  label = left:{\scriptsize $A, C$}] at (0,0) (zero) {};
\node[state, label = below:{\scriptsize $1$},  label = above:{\scriptsize $B$}] at (2,0)(one) {};
\node[state, label = above:{\scriptsize $2$}, label = below:{\scriptsize $C$}] at (4,1) (two) {};
\node[state, label = above:{\scriptsize $3$}, label = right:{\scriptsize $B$}] at (6,0) (three) {};
\path (zero) edge[->] (one);
\path (one) edge[->] (two);
\path (two) edge[->, bend left=10] (three);
\path (three) edge[->] (one);
\path (zero) edge[->, bend left=5] node[above] {\scriptsize$B$}(two);
\path (zero) edge[->, bend right=25] node[above] {\scriptsize$\emptyset$}(three);
\path (two) edge[->, bend left=20] node[below,right] {\scriptsize$B$}(one);
\path (three) edge[->, bend left=10] node[above, left] {\scriptsize$B$}(two);
\path (one) edge[->, bend right=13] node[below] {\scriptsize$C$}(three);
\path (one) edge[->, out=205, in=250, looseness=10] node[left] {\scriptsize$\emptyset$} (one);
\path (two) edge[->, out=50, in=15, looseness=10] node[right] {\scriptsize$B,C$} (two);
\path (three) edge[->, loop below] node[right] {\scriptsize$\emptyset$} (three);
%
\end{tikzpicture}}

In general, the size of $S^i$ is $|\Abox_i| + O(2^{|\TO|})$ and the product of $S^i$, $\Abox_i \in E^+$ is of size $O(2^{|\TO|+|E^+|})$. The upper bounds now follow from \PTime{} and \PSpace{} completeness of checking simulation and containment for transition systems.

Now we sketch the proof of the lower bound. Let $\M$ be a non-deterministic Turing machine that accepts 
$\Sigma$-words $\boldsymbol{x}=x_1\dots x_n$ in \mbox{$N=2^{\textit{poly}(|\boldsymbol{x}|)}$} steps  and erases the tape after a successful computation.
%
%
%
We represent configurations $\conf$ of a computation of $\M$ on $\boldsymbol{x}$ by an $N-1$-long word (with sufficiently many blanks at the end), in which $y$ in the active cell is replaced by $(q,y)$ with the current state $q \in Q$.
An accepting computation of $\M$ on $\boldsymbol{x}$ is encoded by the $N^2$-long word $w=\sharp \conf_1 \, \sharp \, \conf_2 \, \sharp\,  \dots \, \sharp \, \conf_{N-1} \, \sharp \, \conf_{N}$   over $\Xi=\Sigma\cup(Q\times\Sigma)\cup\{\sharp\}$.
Thus, a word $w$ of length $N^2$ encodes an accepting computation iff it starts with the initial configuration $\conf_1$ preceded by $\sharp$, ends with the accepting configuration $\conf_{\it acc}$, and every two length 3 subwords at distance $N$ apart form a \emph{legal tuple}~\cite[Theorem 7.37]{DBLP:books/daglib/0086373}.

We define $\TO$ and $E=(\{\Abox_1^+,\Abox_2^+\},\{\Abox^-\})$ so that their canonical models look as follows, for $\Xi=\{a_1,\ldots,a_k\}$:\\
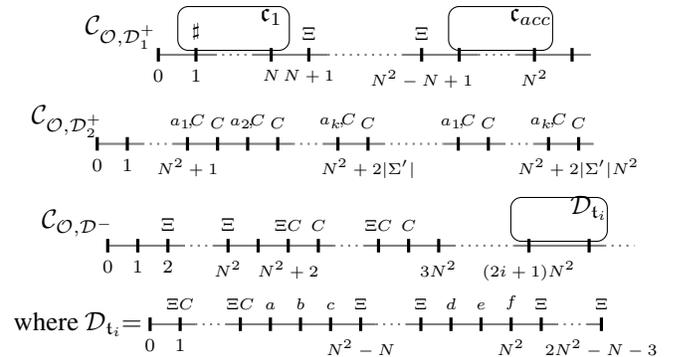
\centerline{
\begin{tikzpicture}[nd/.style={draw,thick,circle,inner sep=0pt,minimum size=1.5mm,fill=white},xscale=0.5]
\node at (-1,0.3) {$\C_{\TO,{\Abox_1^+}}$};
\node   at (1,0.3) (x1)  {};
\node at (3,0.5) {$\conf_1$};
\node  at (3,0.4) (x2) {};
\node at (8.2,0.3) (x3)  {};
\node at (9.9,0.5) {$\conf_{acc}$};
\node  at (10,0.4) (x4) {};
\node[fit = (x1)(x2), basic box = black]  {};
\node[fit = (x3)(x4), basic box = black]  {};
\draw[thick,gray,-] (0,0) -- (1.5,0);
\draw[thick,gray,dotted] (1.5,0) -- (2.5,0);
\draw[thick,gray,-] (2.5,0) -- (4.5,0);
\draw[thick,gray,dotted] (4.5,0) -- (6.5,0);
\draw[thick,gray,-] (6.5,0) -- (8.5,0);
\draw[thick,gray,dotted] (8.5,0) -- (9.5,0);
\draw[thick,gray,-] (9.5,0) -- (11.5,0);
\slin{0,0}{}{$0$};
\slin{1,0}{$\sharp$}{$1$};
\slin{3,0}{}{\tiny$N$};
\slin{4,0}{${\Xi}$}{\tiny$N+1$};
\slin{7,0}{${\Xi}$}{\tiny$N^2-N+1$};
\slin{8,0}{}{};
\slin{10,0}{}{\tiny$N^2$};
\slin{11,0}{}{};
\end{tikzpicture}
}\\
\centerline{
\begin{tikzpicture}[nd/.style={draw,thick,circle,inner sep=0pt,minimum size=1.5mm,fill=white},xscale=0.4]
\node at (-1,0.3) {$\C_{\TO,{\Abox_2^+}}$};
\draw[thick,gray,-] (0,0) -- (1.5,0);
\draw[thick,gray,dotted] (1.5,0) -- (2.5,0);
\draw[thick,gray,-] (2.5,0) -- (6.5,0);
\draw[thick,gray,dotted] (6.5,0) -- (7.5,0);
\draw[thick,gray,-] (7.5,0) -- (9.5,0);
\draw[thick,gray,dotted] (9.5,0) -- (11.5,0);
\draw[thick,gray,-] (11.5,0) -- (13.5,0);
\draw[thick,gray,dotted] (13.5,0) -- (14.5,0);
\draw[thick,gray,-] (14.5,0) -- (16.5,0);
\slin{0,0}{}{$0$};
\slin{1,0}{}{$1$};
\slin{3,0}{\tiny $a_1\!,\!\!C$}{\tiny$N^2+1$};
\slin{4,0}{\tiny $C$}{};
\slin{5,0}{\tiny $a_2\!,\!\!C$}{};
\slin{6,0}{\tiny $C$}{};
\slin{8,0}{\tiny$a_k\!,\!\!C$}{};
\slin{9,0}{\tiny$C$}{\tiny$N^2+2|\Sigma'|$};
\slin{12,0}{\tiny $a_1\!,\!\!C$}{};
\slin{13,0}{\tiny $C$}{};
\slin{15,0}{\tiny $a_k\!,\!\!C$}{};
\slin{16,0}{\tiny $C$}{\tiny$N^2+2|\Sigma'|N^2$};
\end{tikzpicture}
}\\
\centerline{
\begin{tikzpicture}[nd/.style={draw,thick,circle,inner sep=0pt,minimum size=1.5mm,fill=white},xscale=0.4]
\node at (-1,0.3) {$\C_{\TO,{\Abox^-}}$};
\node   at (14,0.3) (x7)  {};
\node at (16,0.5) {$\Abox_{\lt_i}$};
\node  at (16,0.4) (x8) {};
\node[fit = (x7)(x8), basic box = black]  {};
\draw[thick,gray,-] (0,0) -- (2.5,0);
\draw[thick,gray,dotted] (2.5,0) -- (3.5,0);
\draw[thick,gray,-] (3.5,0) -- (7.5,0);
\draw[thick,gray,dotted] (7.5,0) -- (8.5,0);
\draw[thick,gray,-] (8.5,0) -- (11.5,0);
\draw[thick,gray,dotted] (11.5,0) -- (13.5,0);
\draw[thick,gray,-] (13.5,0) -- (16.5,0);
\draw[thick,gray,dotted] (16.5,0) -- (17.5,0);
\slin{0,0}{}{$0$};
\slin{1,0}{}{$1$};
\slin{2,0}{$\Xi$}{$2$};
\slin{4,0}{$\Xi$}{\tiny$N^2$};
\slin{5,0}{}{};
\slin{6,0}{\tiny $\Xi C$}{\tiny$N^2+2$};
\slin{7,0}{\tiny $C$}{};
\slin{9,0}{\tiny $\Xi C$}{};
\slin{10,0}{\tiny $C$}{};
\slin{11,0}{}{\tiny$3N^2$}
\slin{14,0}{}{\tiny $(2i+1)N^2$};
\slin{16,0}{}{};
\end{tikzpicture}
}
\centerline{
\begin{tikzpicture}[nd/.style={draw,thick,circle,inner sep=0pt,minimum size=1.5mm,fill=white},xscale=0.4]
\node at (-2.4,0) {where $\Abox_{\lt_i}\!\! =$};
\draw[thick,gray,-] (0,0) -- (1.5,0);
\draw[thick,gray,dotted] (1.5,0) -- (2.5,0);
\draw[thick,gray,-] (2.5,0) -- (7.5,0);
\draw[thick,gray,dotted] (7.5,0) -- (8.5,0);
\draw[thick,gray,-] (8.5,0) -- (13.5,0);
\draw[thick,gray,dotted] (13.5,0) -- (14.5,0);
\draw[thick,gray,-] (14.5,0) -- (15,0);
\slin{0,0}{}{$0$};
\slin{1,0}{\tiny$\Xi C$}{$1$};
\slin{3,0}{\tiny$\Xi C$}{};
\slin{4,0}{\tiny$a$}{};
\slin{5,0}{\tiny$b$}{};
\slin{6,0}{\tiny$c$}{};
\slin{7,0}{\tiny $\Xi$}{\tiny$N^2-N$};
\slin{9,0}{\tiny $\Xi $}{};
\slin{10,0}{\tiny $d$}{};
\slin{11,0}{\tiny $e$}{}
\slin{12,0}{\tiny $f$}{\tiny$N^2$}
\slin{13,0}{\tiny$\Xi$}{};
\slin{15,0}{\tiny$\Xi$}{\tiny$2N^2-N-3$};
\end{tikzpicture}
}
and $\lt_i=(a,b,c,d,e,f)$ is the lexicographically $i$-th illegal tuple. The parts of the canonical models shown above are of exponential size; however, due to their repetitive nature, they can be described by a polynomial-size $\LTL_{\smash{\horn}}\Xallop$ ontology $\TO$ as in~\cite{DBLP:conf/time/RyzhikovSZ21}.
We show  that the $\Qp[\U]$-query
$$
\varkappa =\Diamond(\rho_1\land C\U(\rho_2 \land C \U (\dots(\rho_{N^2-1}\land (C\U \rho_{N^2}))\dots))),
$$
where $\rho_1\dots\rho_{N^2}$ encodes an accepting computation of $\M$ on $\boldsymbol{x}$, is the only type of query that can separate $E$ under $\TO$.
\end{proof}

As for \emph{data} complexity, we show that $\LTL_{\smash{\horn}}\Xallop$-ontologies come essentially for free:

\begin{theorem}\label{thm:datahorn}
The results of Theorem~\ref{thm:qbewithout} continue to hold for queries mediated by a fixed $\LTL_{\smash{\horn}}\Xallop$-ontology.
\end{theorem}

Intuitively, the reason is that, given a fixed $\LTL_{\smash{\horn}}\Xallop$-ontology $\TO$, we can compute the types of the canonical model $\C_{\TO,\Abox}$, for consistent $\TO$ and $\Abox$, in polynomial time in $\Abox$ by Proposition~\ref{prop:period}, with the length $M$ from Section~\ref{sec:no-ont} being polynomial in $E$. Checking consistency of $\D$ and fixed $\TO$ is known to be in $\PTime$~\cite{DBLP:journals/ai/ArtaleKKRWZ21}.



\section{QBE with $\LTL\Xbd$-Ontologies}
\label{sec:boxdiamond}

In this section, we investigate separability by $\Diamond$-queries under $\LTL\Xbd$-ontologies. Remarkably, we show that, for data complexity, $\LTL\Xbd$-ontologies also come for free
despite admitting arbitrary Boolean operators; cf.,~\cite{Schaerf93}.



%
\begin{theorem}\label{thm:dim}
Let $\mathcal{Q} \in \{\mathcal{Q}_{p}[\Diamond],\mathcal{Q}[\Diamond]\}$. If $E$ is $\mathcal{Q}$-separable under an $\LTL\Xbd$-ontology $\TO$, then $E$ can be separated under $\TO$ by a $\mathcal{Q}$-query of polysize in $E$ and $\TO$. $\mathsf{QBE}(\LTL\Xbd,\mathcal{Q})$ and $\QBEba(\LTL\Xbd,\mathcal{Q})$ are $\Sigma_{2}^{p}$-complete for combined complexity. The presence of $\LTL\Xbd$-ontologies has no effect on the data complexity, which remains the same as in Theorem~\ref{thm:qbewithout}.
\end{theorem}


We comment on the proof of this theorem for $\mathcal{Q}_{p}[\Diamond]$. Taking  into account $\NP$-completeness of checking if $\TO$ is consistent with $\D$ and tractability of this problem for a fixed $\TO$~\cite{DBLP:journals/ai/ArtaleKKRWZ21}, we can assume, as in Theorem~\ref{thm:Horn-diam}, that $\TO$ and $\D$ are consistent for each $\D \in E^+ \cup E^-$. Observe first that if $E$ is separated by $\varkappa \in \Qp[\Diamond]$ of the form~\eqref{dnpath} under an $\LTL\Xbd$-ontology $\TO$, then, as follows from~\cite{ono1980on},  for any $\Abox \in E^-$, there is a model $\mathcal{J}_{\mathcal{D}}\not\models \varkappa(0)$ of $\TO$ and $\Abox$ whose types form a sequence
%
\begin{equation}\label{periodA}
	\!\tp_{0},\dots,\tp_{k},\tp_{k+1},\dots,\tp_{k+l},\dots,\tp_{k+1},\dots,\tp_{k+l},\dots
\end{equation}
with $\max \mathcal{D} \leq k \leq \max \mathcal{D}+|\TO|$ and $l\leq |\TO|$.
%
%
This allows us to find a separating $\varkappa$ of polysize in $E$, $\TO$. Indeed, let $K$ be the maximal $k$ in~\eqref{periodA} over all $\Abox \in E^-$. If the depth $n$ of $\varkappa$ is $\le K$,  we are done. If $n > K$, we shorten $\varkappa$ as follows. Consider the prefix $\varkappa'$ of $\varkappa$ formed by $\rho_0,\dots,\rho_K$. If $\mathcal{J}_{\mathcal{D}}\not\models \varkappa'(0)$ for all $\Abox \in E^-$, we are done. Otherwise, for each $\Abox \in E^-$, we pick a $\rho_i$,   $i > K$, with $\rho_i \not\subseteq \tp_{k + j}$ for any $j \le l$; it must exist as $\mathcal{J}_{\mathcal{D}}\not\models \varkappa(0)$. Then we construct $\varkappa''$ by omitting from $\varkappa$ all $\rho_l$ that are different from those in $\varkappa'$ and the chosen $\rho_i$ with $i > K$. Clearly, $\varkappa''$ is as required.

A $\Sigma_{2}^{p}$-algorithm guesses $\varkappa$ and $\J_{\Abox}$, for $\Abox \in E^-$, and checks in polytime that $\J_{\Abox} \models \TO,\Abox$ and $\J_{\Abox} \not\models \varkappa(0)$ and in \coNP~\cite{ono1980on} that $\TO, \Abox \models  \varkappa(0)$ for all $\Abox \in E^+$. The lower bound is shown by reduction of the validity problem for fully quantified Boolean formulas $\exists \avec{p}\, \forall \avec{q}\,\psi$, where $\avec{p}=p_{1},\dots,p_{k}$ and 	$\avec{q}= q_{1},\dots,q_{m}$ are all propositional variables in $\psi$. We can assume that $\psi$ is not a tautology and $\lnot\psi\not\models x$ for $x\in\{p_i, \lnot p_i, q_j, \lnot q_j \mid i\leq k, j \leq m\}$.
	Let $E = (E^+,E^-)$ with $E^+ = \{\Abox_1, \Abox_2\}$, $E^- = \{\Abox_3\}$, where
	$$
	\mathcal{D}_{1}=\{B_1(0)\}, \ \mathcal{D}_{2}=\{B_2(0)\}, \ \mathcal{D}_{3}=\{q_{1}(0),\dots,q_{m}(0)\},
	$$
	and let $\mathcal{O}$ contain the following axioms with fresh  atoms $B_{1},B_{2},A_{i},\bar{A}_{i}$, for $i=1,\dots,k$:
	\begin{align*}
		B_1 \lor B_2 \rightarrow \neg\psi,\quad 		%
		& \textstyle p_{i} \rightarrow \Diamond \big(\bar{A}_{i} \wedge \bigwedge_{j\not=i} (A_{j} \wedge \bar{A}_{j})\big), \\
		& \hspace*{-2.5mm} \textstyle\neg p_{i} \rightarrow \Diamond \big(A_{i} \wedge \bigwedge_{j\not=i} (A_{j} \wedge \bar{A}_{j})\big).
		%
	\end{align*}
	Then $\exists \avec{p}\, \forall \avec{q}\, \psi$ is valid iff
	$E$ is $\Qp[\Diamond]$-separable under $\mathcal{O}$.

We obtain the \NP{} upper bounds in data complexity using the same argument as for the $\Sigma_{2}^{p}$-upper bound and observing that checking
$\mathcal{O},\mathcal{D}\models \varkappa(0)$ is in \PTime{} in data complexity. The \NP{} lower bounds are inherited from the ontology-free case. The proof of the \PTime{} upper bounds is more involved. We illustrate the idea for $\TO$ with arbitrary Boolean but without temporal operators.
In this case, one can show (which is non-trivial) that $\TO,\Abox\models \varkappa(0)$ iff $\I_{\TO,\Abox}\models \varkappa(0)$, where $\I_{\TO,\Abox}$ is the \emph{completion} of $\mathcal{D}$: it  contains $A(\ell)$ iff $\TO\cup \{B \mid B(\ell)\in \Abox\} \models A$. For example, if $\TO=\{A \vee B\}$
and $\Abox = \{A(1),B(1),A(3),B(3)\}$, the completion $\I_{\TO,\Abox}$ is just $\Abox$ regarded as an interpretation (so $\I_{\TO,\Abox}$ does not have to be a model of $\TO$). It can be constructed in polytime in $\Abox$ and, due to the equivalence above, used to prove the \PTime{} upper bounds using dynamic programming. That equivalence does not hold for $\LTL\Xbd$, but the technique can be extended by applying it to data sets enriched by certain types.

Note that the completion technique does not work for $\nxt,\!\Diamond$-queries. For example, $\TO,\Abox\models \Diamond(A \land \nxt B)$ for $\mathcal{D}$ and $\TO$ defined above, and so the equivalence does not hold. In fact, the complexity of separability by $\nxt\Diamond$-queries remains open.


\section{QBE with \LTL{}-Ontologies}
\label{sec:full}
For ontologies with arbitrary \LTL-axioms, we obtain:

\begin{theorem}\label{thm:fullLTL}
$(i)$ $\QBE(\LTL,\mathcal{Q})$ is in $2\ExpTime{}$, for any $\mathcal{Q} \in \{\, \mathcal{Q}[\Diamond], \mathcal{Q}[\nxt, \Diamond], \mathcal{Q}[\Us]\,\}$.
$(ii)$ $\mathsf{QBE}(\LTL, \mathcal{Q})$ is in $2\ExpSpace{}$, for any $\mathcal{Q} \in \{\, \mathcal{Q}_p[\Diamond], \mathcal{Q}_p[\nxt, \Diamond], \mathcal{Q}_p[\U]\,\}$.
%
\end{theorem}

The proof requires a further modification of the transition systems $S^i$ in Theorem~\ref{criterionUsA}. We illustrate it by an example.
Let
$\TO = \{ A \to \Diamond B,\, \top \to A \lor B,\, A \land B \to \bot\}$ with the set of $\TO$-types $\avec{T}_\TO = \{\tp_1, \tp_2, \tp_3 \}$, where $\tp_1 = \{ A, \neg B, \Diamond B\}$, $\tp_2 = \{ \neg A, B, \neg \Diamond B\}$, $\tp_3 = \{ A, \neg B, \neg \Diamond B\}$, and $\tp_4 = \{\neg A, B, \Diamond B\}$, from which we omitted subformulas such as $A \lor B$ that are true or false in all types. For non-empty sets $\avec{T}_1, \avec{T}_2 \subseteq \avec{T}_\TO$ and $\Gamma \subseteq \Sigma^\bot$, we take the relation $\avec{T}_1 \to_\Gamma \avec{T}_2$, which, intuitively, says that if there are instants $n_\I$ in all models $\I$ of $\TO, \Abox$ such that $\{\tp_\I(n_\I) \mid \I \models \TO, \Abox\} = \avec{T}_1$, then there exist $m_\I > n_\I$ with $\{\tp_\I(m_\I) \mid \I \models \TO, \D\} = \avec{T}_2$ and $\Gamma = \{ A \in \Sigma^\bot \mid \I, m \models A \text{ for all }\I \text{ and } n_\I < m < m_\I \}$. In our  example, we have $\{\tp_1, \tp_3\} \to_{\Sigma^\bot} \{\tp_1, \tp_2, \tp_3, \tp_4\}$ and $\{\tp_1, \tp_3\} \to_{\{B\}} \{\tp_1, \tp_3, \tp_4\}$ (among others). Then we construct the following transition system $S^i$ for, say, $\Abox_i = \{ A(0)\}$, which reflects all $\Q[\U_s]$-queries over $\TO, \Abox_i$ using $\avec{T}' \subseteq \avec{T}_\TO$ as states (the initial state is $\{\tp_1, \tp_3\}$ since $A(0) \in \Abox_i$):\\[-10pt]
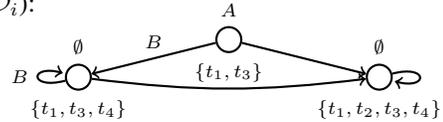
\centerline{
\begin{tikzpicture}[thick, transform shape]
\tikzset{every state/.style={minimum size=0pt}}
%
%
%
\node[state, label = below:{\scriptsize $\{t_1, t_3\}$},  label = above:{\scriptsize $A$}] at (2,.5) (one) {};
\node[state, label = below:{\scriptsize $\{t_1, t_3, t_4\}$},  label = above:{\scriptsize $\emptyset$}] at (0,0) (two) {};
\node[state, label = below:{\scriptsize $\{t_1, t_2, t_3, t_4\}$},  label = above:{\scriptsize $\emptyset$}] at (4,0) (three) {};
\path (one) edge[->]  node[above] {\scriptsize$B$} (two);
\path (two) edge[->, bend right=7]  (three);
\path (one) edge[->] (three);
\path (three) edge[->, loop right] (three);
\path (two) edge[->, loop left] node[left] {\scriptsize$B$} (two);
\end{tikzpicture}
}

The $S^i$ can be constructed in $2\ExpTime{}$ in $|\Abox_i|+|\TO|$ (checking $\avec{T}_1 \to_\Gamma \avec{T}_2$, for given $\avec{T}_1$, $\avec{T}_2$ and $\Gamma$, can be done in \ExpSpace{}). Also, the product of the $S^i$, for $\Abox_i \in E^+$,  can be constructed in $2\ExpTime{}$ in $|\Abox_i|+|E^+|$.



\section{Conclusions}\label{conclusions}

We have started an investigation of the computational complexity of query-by-example for principal classes of \LTL{}-queries, both with and without mediating ontologies. Our results are encouraging as we exhibit important cases that are tractable for data complexity and not harder than satisfiability for combined complexity. Many interesting and technically challenging problems remain open. Especially intriguing are queries with $\U$. For example, we still need to pinpoint the size of minimal separating $\mathcal{Q}[\Us]$- and $\Qp[\U]$-queries under a Horn ontology. 
The tight complexity of QBE for unrestricted $\U$-queries is also open. In general, such queries could be too perplexing for applications; however, they can express useful disjunctive patterns such as `in at most $n$ moments of time'\!.
Note also that sparse data instances with large gaps between timestamps may require binary representations thereof, for which the proofs of some of our results do not go through.

Our results and techniques provide a good starting point for studying QBE with (ontology-mediated) queries over temporal databases with a full relational component~\cite{DBLP:journals/tods/ChomickiTB01,DBLP:reference/db/ChomickiT18a,DBLP:journals/jair/ArtaleKKRWZ22} and also for the construction of separating queries satisfying additional conditions such as being a longest/shortest separator~\cite{DBLP:journals/cor/BlumDSJLMR21,DBLP:conf/icgi/FijalkowL21} or a most specific/general one~\cite{DBLP:journals/corr/abs-2206-05080}.


%
%

%


\section*{Acknowledgements}
This work was supported by EPSRC UK grants EP/S032207, EP/S032282, and EP/W025868.


\bibliographystyle{named}



\appendix
\onecolumn
\noindent{\Large \bf Appendix: Proofs}


\section{Proofs for Section~\ref{prelims}}
\bigskip
\noindent
\textbf{Theorem~\ref{th:data-reduction}.}
{\em The following polynomial-time reductions hold\textup{:}
\begin{description}
\item[(\emph{i}.1)] $\QBE(\mathcal{L},\mathcal{Q}) \le_p \QBE_{1\text{--}}(\mathcal{L},\mathcal{Q})$, for any $\mathcal{Q}$ closed under~$\land$, and any $\mathcal{L}$ \textup{(}including $\mathcal{L} = \emptyset$\textup{)},

\item[(\emph{i}.2)] $\QBE(\mathcal{L},\mathcal{Q}) \le_p \QBE^{2\text{+}}(\mathcal{L},\mathcal{Q})$, for $\mathcal{L} \in \{\LTL, \LTL\Xbd\}$,

\item[(\emph{i}.3)] $\QBE(\mathcal{L},\mathcal Q[\nxt,\Diamond]) \le_p \QBE(\mathcal{L},\mathcal{Q}_{p}^{\circ}[\Diamond]))$ and $\mathsf{QBE}(\mathcal{L},\mathcal{Q}[\Diamond]) \le_p \mathsf{QBE}(\mathcal{L},\mathcal{Q}_{p}[\Diamond])$, for any $\mathcal{L}$,

\item[(\emph{ii}.1)] $\QBE(\Qp[\Diamond]) \le_p \QBE(\Qp[\nxt,\Diamond])$ and $\QBE(\Qp[\Diamond]) \le_p \QBE(\Qp[\U]) \le_p \QBE_{1\text{--}}(\Qp[\U])$

\item[(\emph{ii}.2)] $\mathsf{QBE}(\mathcal{Q}[\nxt,\Diamond]) =_p \QBE(\mathcal{Q}[\Diamond]) \le_p \QBE(\mathcal Q[\Us])$.
%
\end{description}
Reductions $(i.1)$--$(i.3)$ work for combined complexity and $(i.1)$ and $(i.3)$ also work for data complexity. The reductions preserve boundedness of the number of positive/negative example.
}
\begin{proof}
(\emph{i}.1) Observe that if $\mathcal{Q}$ is closed under $\land$, $E = (E^+,E^-)$ and $E^- = \{\Abox_1^-,\dots,\Abox_n^-\}$, then $E$ is $\mathcal{Q}$-separable under $\TO$ iff each $(E^+,\{\Abox_i^-\})$, $1 \le i \le n$, is. Indeed, if $\varkappa_i$ separates $(E^+,\{\Abox_i^-\})$, then $\varkappa_1 \land \dots \land \varkappa_n$ separates $E$. For such $\mathcal{Q}$, we can thus assume that $E^-$ consists of a single data instance. Note that $\Qp[\nxt,\Diamond]$ and $\Qp[\U]$ are not closed under $\land$.

(\emph{i}.2) Let $E^+=\{\Abox_1,\dots,\Abox_n\}$ and let $k=\max_i\max\Abox_i$. We construct an ontology $\TO'$ by taking the fresh atoms $A_1,\ldots,A_n$, $C_0,\ldots,C_k$, $D_0,\ldots, D_k$, $S_1,S_2$ and adding the following axioms to the given ontology $\TO$:
\begin{align*}
& S_1\to A_1\lor\dots \lor A_n, \quad S_2\to A_1\lor\dots \lor A_n,\\
&C_i\land\Diamond A_j\to X, \quad D_i\land\Diamond A_j\to X, \ \ \text{for $X(i)\in\Abox_j$.}
\end{align*}
Let $E'^+$ consist of $\Abox'_1=\{C_0(0),\dots,C_k(n'),S_1(n'+1)\}$ and $\Abox'_2=\{D_0(0),\dots,D_k(n'),S_2(n'+1)\}$. Then every model of $\TO',\Abox'_1$ or $\TO',\Abox'_2$ contains a model of at least one of the $\TO,\Abox_i$ and, conversely, every model of any $\TO,\Abox_i$ can be converted into a model of $\TO',\Abox_1'$ or $\TO',\Abox_2'$ by adding only the newly introduced symbols. So, if there is a separating query for $(E^+,E^-)$ under the ontology $\TO$, then the same query separates $(E'^+,E^-)$ under the ontology $\TO'$. And if there is a separating query for $(E'^+,E^-)$ under the ontology $\TO'$, then it cannot contain any symbols from $\sig  (\TO') \setminus\sig (\TO)$, and so it separates $(E^+,E^-)$ under the ontology $\TO$.

(\emph{i}.3) Recall from the main part of the paper that $[\rho_{0} \wedge \Diamond(\rho_{1} \wedge \bigwedge_{i}\Diamond \varkappa_{i})] \equiv [\rho_{0} \wedge \bigwedge_{i} \Diamond(\rho_{1} \wedge \Diamond \varkappa_{i})]$, $\nxt \Diamond \varkappa \equiv \Diamond\nxt \varkappa$ and $\nxt (\varkappa \land \varkappa') \equiv (\nxt \varkappa \land \nxt \varkappa')$, and so each $\mathcal{Q}[\nxt,\Diamond]$-query can be equivalently transformed in polynomial time into a conjunction of $\mathcal{Q}_{p}^{\circ}[\Diamond]$-queries. If $E^{-}$ is a singleton, then a conjunction of queries in $\mathcal{Q}_{p}^{\circ}[\Diamond]$ separates $(E^{+},E^{-})$ under an ontology $\TO$ iff a single conjunct separates $(E^{+},E^{-})$ under $\TO$. Thus, there is $\q\in \mathcal{Q}[\nxt,\Diamond]$ separating an arbitrary $(E^{+},E^{-})$ under $\TO$ iff there are polysize $\q_{\mathcal{D}}\in\mathcal{Q}_{p}^{\circ}[\Diamond]$ separating $(E^{+},\{\mathcal{D}\})$ under $\TO$, for each $\mathcal{D}\in E^{-}$.
The second reduction is obtained by dropping $\nxt$ from the argument above.

In (\emph{ii}.1), the first two reductions are proved by adding to $E^+ \ni \Abox$, for some $\Abox$, the data instance $\mathcal{D}'=\{ A(mn) \mid  A(n)\in \mathcal{D}\}$ with $m=\max\mathcal{D}+2$. Now, if $\Abox\models \varkappa(0)$ and $\Abox'\models \varkappa(0)$, for $\varkappa \in \Qp[\U]$, then $\varkappa$ is equivalent to a $\Qp[\Diamond]$-query.

For the third reduction, we observe first that without loss of generality one can assume that the positive examples do not contain atoms of the form $X(0)$. 
%
%
Indeed, suppose $E = (E^+, E^-)$, $E^+=\{\Abox^+_1,\ldots,\Abox^+_n\}$, and $E^-=\{\Abox^-_1,\dots,\Abox^-_k\}$. 
Let $\rho=\bigcap_{i=1}^n\{X\mid X(0)\in\Abox^+_i\}$, $E'^-=\{\Abox\in E^-\mid \Abox\not\models\rho\}$, and let $D=(D^+,D^-)$, where $D^+=\{\Abox\setminus\{X(0)\mid X\in \sig(E)\}\mid\Abox\in E^+\}$ and $D^-=\{\Abox\setminus\{X(0)\mid X\in \sig(E)\}\mid\Abox\in E'^-\}$. If $\varphi=\rho_0\land \lambda_1\U(\rho_1\land \lambda_2\U ( \dots(\rho_{l-1}\land (\lambda_{l}\U \rho_{l}))\dots))$ separates $E$ then $\varphi'$, which is $\varphi$ with $\rho_0$ replaced by $\emptyset$, separates $D$. If $\psi=\rho_0\land \lambda_1\U(\rho_1\land \lambda_2\U ( \dots(\rho_{l-1}\land (\lambda_{l}\U \rho_{l}))\dots))$ separates $D$, then $\rho_0=\emptyset$ and $\psi'$, which is $\psi$ with $\rho_0$ replaced by $\rho$, separates $E$.

Now, suppose $E = (E^+, E^-)$, $E^+=\{\Abox^+_1,\ldots,\Abox^+_n\}$, $n>1$, and $E^-=\{\Abox^-_1,\dots,\Abox^-_k\}$, and $X(0) \notin \Abox_i^+$, for any $X$ and $i$. Let $B,C$ be fresh atoms and $m=\max_{\Abox\in E^+\cup E^-}(\max(\Abox))+2$.
We set $E'^+=\{\Abox''^+_1,\Abox_2'^+,\dots,\Abox'^+_n\}$, where 
\begin{align*}
& \Abox_i'^+ = \{X(j+m)\mid X(j)\in\Abox^+_i\}
\cup\{B(m)\}\cup\{C(j)\mid m< j< m+\max(\Abox_i^+)\},\\
& \Abox_1''^+ = \{X(j+1)\mid X(j)\in\Abox^+_1\}\cup\{B(1)\}
\cup\{C(j)\mid 1< j< 1+\max(\Abox_1^+)\}.
\end{align*}
We also set $E'^-=\{\Abox'^-\}$ with 
$$
\Abox'^-=\{X((2i-1)m+j)\mid X(j)\in\Abox_i^-, i\in[1,k]\}
\cup\{B((2i-1)m)\mid i\in[1,k]\}\cup \{C(i)\mid (2i-1)m< i<2im,i\in[1,k]\}.
$$
See the picture below for an illustration.\\
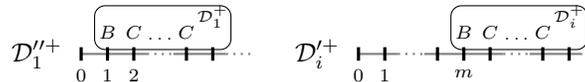
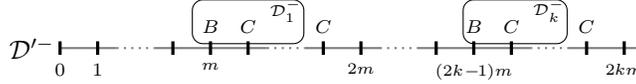
\centerline{
\begin{tikzpicture}[nd/.style={draw,thick,circle,inner sep=0pt,minimum size=1.5mm,fill=white},xscale=0.35]
\node[label=left:{$\Abox''^+_1$}] at (0,0) {};
\node   at (1.2,0.3) (x1)  {};
\node at (5,0.5) {\tiny $\Abox_1^+$};
\node  at (5,0.4) (x2) {};
\node[fit = (x1)(x2), basic box = black]  {};
\draw[thick,gray,-] (0,0) -- (2.5,0);
\draw[thick,gray,dotted] (2.5,0) -- (3.5,0);
\draw[thick,gray,-] (3.5,0) -- (5.5,0);
\draw[thick,gray,dotted] (5.5,0) -- (6.5,0);
\slin{0,0}{}{$0$};
\slin{1,0}{$B$}{$1$};
\slin{2,0}{$C$}{$2$}
\node at (3,0.2) {\tiny$\ldots$};
\slin{4,0}{$C$}{};
\slin{5,0}{}{};
\end{tikzpicture}
\quad
\begin{tikzpicture}[nd/.style={draw,thick,circle,inner sep=0pt,minimum size=1.5mm,fill=white},xscale=0.35]
\node[label=left:{$\Abox'^+_i$}] at (0,0) {};
\node   at (4.2,0.3) (x1)  {};
\node at (8,0.5) (D) {\tiny $\Abox_i^+$};
\node  at (8,0.4) (x2) {};
\node[fit = (x1)(x2), basic box = black]  {};
\draw[thick,gray,-] (0,0) -- (1.5,0);
\draw[thick,gray,dotted] (1.5,0) -- (2.5,0);
\draw[thick,gray,-] (2.5,0) -- (5.5,0);
\draw[thick,gray,dotted] (5.5,0) -- (6.5,0);
\draw[thick,gray,-] (6.5,0) -- (8.5,0);
\slin{0,0}{}{$0$};
\slin{1,0}{}{$1$};
\slin{3,0}{}{};
\slin{4,0}{$B$}{$m$};
\slin{5,0}{$C$}{};
\node at (6,0.2) {\tiny$\ldots$};
\slin{7,0}{$C$}{};
\slin{8,0}{}{};
\end{tikzpicture}
}

\centerline{
\begin{tikzpicture}[nd/.style={draw,thick,circle,inner sep=0pt,minimum size=1.5mm,fill=white},xscale=0.5]
\node[label=left:{$\Abox'^-$}\!\!\!] at (0,0) {};
\node   at (4,0.3) (x1)  {};
\node at (6,0.5) {\tiny $\Abox_1^-$};
\node  at (6,0.4) (x2) {};
\node at (11.2,0.3) (x3)  {};
\node at (13,0.5) {\tiny $\Abox^-_k$};
\node  at (13,0.4) (x4) {};
\node[fit = (x1)(x2), basic box = black]  {};
\node[fit = (x3)(x4), basic box = black]  {};
\draw[thick,gray,-] (0,0) -- (1.5,0);
\draw[thick,gray,dotted] (1.5,0) -- (2.5,0);
\draw[thick,gray,-] (2.5,0) -- (5.5,0);
\draw[thick,gray,dotted] (5.5,0) -- (6.5,0);
\draw[thick,gray,-] (6.5,0) -- (8.5,0);
\draw[thick,gray,dotted] (8.5,0) -- (9.5,0);
\draw[thick,gray,-] (9.5,0) -- (12.5,0);
\draw[thick,gray,dotted] (12.5,0) -- (13.5,0);
\draw[thick,gray,-] (13.5,0) -- (15.5,0);
\slin{0,0}{}{$0$};
\slin{1,0}{}{$1$};
\slin{3,0}{}{};
\slin{4,0}{$B$}{\tiny$m$};
\slin{5,0}{$C$}{};
\slin{7,0}{$C$}{};
\slin{8,0}{}{\tiny$2m$};
\slin{10,0}{}{};
\slin{11,0}{$B$}{\tiny$(2k\!-\!\!1)m$};
\slin{12,0}{$C$}{};
\slin{14,0}{$C$}{};
\slin{15,0}{}{\tiny$2km$};
\end{tikzpicture}
}

\noindent

Let $E = (E'^+, E'^-)$.
We prove equiseparability of $E$ and $E'$. 

($\Rightarrow$) Suppose $\varphi =\rho_0\land \lambda_1\U(\rho_1\land \lambda_2\U ( \dots(\rho_{l-1}\land (\lambda_{l}\U \rho_{l}))\dots))$ with $l<m$ separates $E$. Then $\rho_0=\top$ (since $X(0) \notin \Abox_i^+$, for any $X$). Consider the query $\Diamond\varphi'$, where $\varphi'$ is $\varphi$ in which $\rho_0$ is replaced by $B$ 
and $\lambda_i\ne\bot$ are replaced with $\lambda_i\land C$. Since $\Abox^+_j\models\varphi$ we have $\Abox\models\varphi'$ for all $\Abox\in E'^+$. If $\Abox'^-\models\Diamond\varphi'$, then $\Abox'^-\models\varphi'((2j-1)m)$ for some $j$. Since all $\lambda'$s contain $C$ and $l<m$, we have $\Abox^-_j\models\varphi$, which is impossible. Therefore, $\Diamond\varphi'$ separates $E'$.

($\Leftarrow$) Suppose that $\psi =\rho_0\land \lambda_1\U(\rho_1\land \lambda_2\U ( \dots(\rho_{l-1}\land (\lambda_{l}\U \rho_{l}))\dots))$ with $\rho_l\ne\top$ 
separates $E'$. Since $\Abox''^+_1\models \psi$, we have $l<m$. Find the smallest $i$ such that $\rho_i\ne\top$. As $\Abox'^+_2\models \psi$, there is $i'\le i$ with $\lambda_{i'} = \top$.

If $B\in\rho_i$ then, since $\Abox''^+_1\models \psi$, we have $i=1$ and $\lambda_1=\top$. Let $\psi'= \lambda'_{2}\U(\rho'_{2}\land \lambda'_{3}\U ( \dots(\rho'_{l-1}\land (\lambda_{l}\U \rho_{l}))\dots))$ where $\rho'_j=\rho_j\setminus\{C\}$ and $\lambda'_j=\lambda_j\setminus\{C\}$. We see that in this case $\Abox'^+_j\models \psi'(m)$, and so $\Abox^+_j\models\psi'$ for all $j$. Also since $\Abox'^-\not\models\psi$, we have $\Abox^-_j\not\models\psi'$, and so $\psi'$ separates $E$.

If $B\notin \rho_i$, let $\psi'= \Diamond(\rho_i\land\lambda_{i+1}\U(\rho_{i+1}\land \lambda_{i+2}\U ( \dots(\rho_{l-1}\land (\lambda_{l}\U \rho_{l}))\dots)))$. Then $\Abox^+\models\psi'$ for all $\Abox^+\in E'^+$. If $\Abox'^-\models\psi'$, then we have $\Abox'^-\models\psi$ as $\lambda_{i'} = \top$, and so $\psi'$ also separates $E'$. Consider the instance $\Abox_l'^-$ (corresponding to some $\Abox^-_l$) shown below:\\ 
$$\Abox_l'^-=\{X((2l-1)m+j)\mid X(j)\in\Abox_l^-\}\cup\{B((2l-1)m)\}\cup \{C(j)\mid (2l-1)m<j<2lm \}.
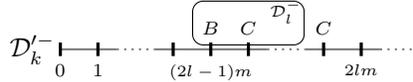
$$
\centerline{
\begin{tikzpicture}[nd/.style={draw,thick,circle,inner sep=0pt,minimum size=1.5mm,fill=white},xscale=0.5]
\node[label=left:{$\Abox_k'^-$}\!\!\!] at (0,0) {};
\node   at (4,0.3) (x1)  {};
\node at (6,0.5) {\tiny $\Abox_l^-$};
\node  at (6,0.4) (x2) {};
\node[fit = (x1)(x2), basic box = black]  {};
\draw[thick,gray,-] (0,0) -- (1.5,0);
\draw[thick,gray,dotted] (1.5,0) -- (2.5,0);
\draw[thick,gray,-] (2.5,0) -- (5.5,0);
\draw[thick,gray,dotted] (5.5,0) -- (6.5,0);
\draw[thick,gray,-] (6.5,0) -- (8.5,0);
\draw[thick,gray,dotted] (8.5,0) -- (9.5,0);
\slin{0,0}{}{$0$};
\slin{1,0}{}{$1$};
\slin{3,0}{}{};
\slin{4,0}{$B$}{\tiny$(2l-1)m$};
\slin{5,0}{$C$}{};
\slin{7,0}{$C$}{};
\slin{8,0}{}{\tiny$2lm$};
\end{tikzpicture}
}

\noindent
Since $\Abox_l'^-\subseteq\Abox'^-$, we have $\Abox_l'^-\not\models\psi'$, and so $\Abox_l^- \not\models\psi''$, where $\psi''$ is $\psi'$ with the $\rho_j$ replaced by $\rho_j\setminus\{C\}$ and $\lambda_j$ replaced by $\lambda_j\setminus\{B,C\}$, for $\lambda_j\ne\emptyset$. Clearly, $\Abox^+\models\psi''$ for all $\Abox^+\in E^+$, and so $\psi''$ separates $E$.

\medskip

$(ii.2)$ To show  $\mathsf{QBE}(\mathcal{Q}[\nxt,\Diamond]) \le_p \QBE(\mathcal{Q}[\Diamond])$, suppose that $E=(E^{+},E^{-})$ is given. Let $m$ be the maximum over all $\max \mathcal{D}$ with $\mathcal{D}\in E^{+}$. Introduce, for every $A$ such that $A(\ell) \in \mathcal{D}$ for some $\mathcal{D}\in E^{+}$, a
fresh atom $A_{k}$, $0<k\leq m$, and extend any $\mathcal{D}\in E^{+}\cup E^{-}$ to a data instance $\mathcal{D}'$ by adding $A_{k}(\ell)$ to $\mathcal{D}$ if $A(k+\ell)\in \mathcal{D}$. Let $F^{+}=\{\mathcal{D}' \mid \mathcal{D}\in E^{+}\}$ and
$F^{-}=\{\mathcal{D}' \mid \mathcal{D}\in E^{-}\}$. Then clearly $E$ is $\mathcal{Q}[\Diamond\nxt]$-separable iff $F$ is $\mathcal{Q}[\Diamond]$-separable.

The converse reduction and $\QBE(\mathcal{Q}[\Diamond]) \le_p \QBE(\mathcal Q[\Us])$ are proved similarly to (\emph{ii}.1).
\end{proof}

\section{Proofs for Section~\ref{sec:no-ont}}
We show the complexity results in Table~\ref{table:free}. To this end, we first introduce some notation for sequence problems.
Let $\Sigma$ be an alphabet of symbols. A \emph{word over $\Sigma$} is a finite sequence of symbols from $\Sigma$. A word $\alpha$ is a \emph{subsequence} of a word $\beta$ if $\alpha$ can be obtained from $\beta$ by removing zero or more symbols anywhere in $\beta$. For a set $S$ of words, we call a word $\alpha$ a \emph{common subsequence of $S$} if it is a subsequence of every word in $S$.
The \emph{consistent subsequence problem (CSSP)} is formulated as follows:
\begin{description}
	\item[Given:] sets $S^{+}$ and $S^{-}$ of words over an alphabet $\Sigma$.	
	\item[Problem:] decide whether there exists a common subsequence of $S^{+}$ that is a not subsequence of any word in $S^{-}$.
\end{description}

The following is shown in~\cite{DBLP:journals/tcs/Fraser96}:

\begin{theorem}\label{thm:Fraser1}
$(i)$ CSSP is \NP{}-hard even if both the alphabet and $S^{+}$ have cardinality two.
	
$(ii)$ CSSP is \NP{}-hard even if $S^{-}$ is a singleton.
\end{theorem}

Another problem of interest for us is the following \emph{common subsequence problem (KsubS)}:
\begin{description}
	\item[Given:] a set $S$ of words over an alphabet $\Sigma$ and a number $k$.	
	\item[Problem:] decide whether there exists a common subsequence of $S$ of length at least $k$.
\end{description}

The following is shown in~\cite{DBLP:journals/jacm/Maier78}:

\begin{theorem}\label{Meier1}
  KsubS is \NP{}-hard even if the alphabet has cardinality $2$.
\end{theorem}

We are now in a position to prove the results for $\nxt\Diamond$-queries
in Table~\ref{table:free}. We start by proving the \NP{}-lower bounds for  $\mathcal{Q}_{p}[\Diamond]$. We actually show a slightly stronger result than in the table.

\begin{lemma}\label{lem:np}
	 QBE$(\mathcal{Q}_{p}[\Diamond])$ with two positive examples or a single  negative example is \NP{}-hard.
\end{lemma}
\begin{proof}
	The proof by polynomial-time reduction of CSSP (as formulated in Theorem~\ref{thm:Fraser1}) is trivial. It is also of interest to give a proof of the second claim (a single negative example) via a polynomial-time reduction of KsubS.
	The proof also works directly for $\mathcal{Q}_{p}^{\circ}[\Diamond]$-queries.
	Suppose that an instance $S,k$ of KsubS over alphabet $\{A,B\}$ is given.
	We define $E$ of the form $(E^{+},\{\mathcal{D}^{-}\})$ such that the following conditions are equivalent:
	\begin{itemize}
		\item there exists a common subsequence of $S$ of length $k$;
		\item there exists $\q\in \mathcal{Q}_{p}[\Diamond]$ that separates $(E^{+},\{\mathcal{D}^{-}\})$;
		\item there exists $\q\in \mathcal{Q}_{p}^{\circ}[\Diamond]$ that separates $(E^{+},\{\mathcal{D}^{-}\})$.
	\end{itemize}
	We represent each word $w\in S$ as a data instance $\Abox_{w}$ starting at time point $1$ (for example, the word $w=ABBA$ is represented as $\Abox_w=\{A(1), B(2), B(3), A(4)\}$). Now let
	$$
	\Abox^{+} =\{ A(i(k+2)),B(i(k+2))\mid 1 \leq i \leq k\}
	$$
	and 
	$$
	\Abox^{-} =\Abox^{+}\setminus \{A(k(k+2)),B(k(k+2))\}
	$$
	and let $E^{+}=\{\Abox_{w} \mid w\in S\}\cup \{\Abox^+\}$. Assume first that there exists a common subsequence $C_{1}\cdots C_{k}$ of $S$ of length $k$. Then
	$\Diamond(C_{1}\wedge \Diamond(C_{2} \wedge \cdots \wedge \Diamond C_{k}))$ separates $(E^{+},\{\mathcal{D}^{-}\})$. Now assume that a query
	\begin{align*}
		\varkappa& = \rho_0 \land \Diamond (\rho_1 \land \Diamond (\rho_2 \land \dots \land \Diamond \rho_n) ),
	\end{align*}
	where every $\rho_{i}$ is a $\mathcal{Q}[\nxt]$-query, separates $(E^{+},\{\mathcal{D}^{-}\})$.
	As $\Abox_{w}\models \varkappa(0)$ for some $w\in S$, we have that $n\leq k$ and all $\rho_{i}$ have depth bounded by $k$. Then $\rho_{0}=\top$ and also, as there are `gaps' of length $k+1$ between any two entries in $\mathcal{D}^{+}$ and since $\mathcal{D}^{+}\models \varkappa(0)$ we may assume that each $\rho_{i}$, $i>0$, is of the form $\nxt^{m_{i}}\rho_{i}'$ with $0\leq m_{i} \leq k$ and
	$\rho_{i}'$ a conjunction of atoms. Observe that we can satisfy, in $\mathcal{D}^{+}$,
	\begin{itemize}
		\item $\rho_{1}$ in the interval $\{1,\ldots,k+2\}$;
		\item $\rho_{2}$ in the interval $\{(k+1)+1,\ldots,2(k+2)\}$;
		\item and so on, with $\rho_{n}$ satisfied in the interval $\{(n-1)(k+2)+1,\ldots,n(k+2)\}$.
	\end{itemize}
	In particular, if $\rho_{i}$ is a conjunction of atoms, then it can be satisfied in $i(k+2)$.
	If $n<k$, then it follows directly that $\mathcal{D}^{-}\models \varkappa(0)$, and we have derived a contradiction.
	Hence $n=k$. Then, as the depth of $\varkappa$ is bounded by $k$, $\rho_{k}$ is a conjunction of atoms. In fact, one can now show by induction starting with $\rho_{k-1}$ that all $\rho_{i}$, $i>0$, are nonempty conjunctions of atoms. Otherwise a shift to the left shows that $\mathcal{D}^{-}\models \varkappa(0)$ and we have derived a contradiction.
	Thus $\varkappa$ takes the form
	$\Diamond (\rho_1 \land \Diamond (\rho_2 \land \dots \land \Diamond \rho_k) )$ with all $\rho_{i}$ non-empty. It follows from $\Abox_{w}\models \varkappa(0)$ for all $w\in S$ that
	$\varkappa$ defines a common subsequence of $S$ of length $k$, as required.
	%
\end{proof}
The \NP{}-lower bound for QBE$(\mathcal{Q}_{p}[\nxt,\Diamond])$ with a bounded number of positive examples or a single negative example follows from Lemma~\ref{lem:np} and Theorem~\ref{th:data-reduction} $(ii.1)$.

We next obtain the \NP{}-lower bound for QBE$(\mathcal{Q}[\Diamond])$ with a single negative example from Lemma~\ref{lem:np} by observing that it follows from the proof of the first part of Theorem~\ref{th:data-reduction} $(i.3)$ that any $(E^{+},E^{-})$ with $E^{-}$ a singleton is $\mathcal{Q}_{p}[\Diamond]$-separable if, and only if, it is $\mathcal{Q}[\Diamond]$-separable. The \NP{}-lower bound for QBE$(\mathcal{Q}[\nxt,\Diamond])$ with a single negative example follows from
the \NP{}-lower bound for QBE$(\mathcal{Q}[\Diamond])$ with a single negative example using Theorem~\ref{th:data-reduction} $(ii.2)$.

We come to the $\NP{}$-upper bounds. Recall that a query language $\mathcal{Q}$ has the \emph{polynomial separation property (PSP) under an ontology language $\mathcal{L}$} if any $\mathcal{Q}$-separable example is separated by a query in $\mathcal{Q}$ of polynomial size.
The \NP{}-upper bounds for query languages using $\Diamond$ (and $\nxt$) in Table~\ref{table:free} follow trivially from the following result.

\begin{lemma}\label{lem:psp}
	Let $\mathcal{Q}\in \{\mathcal{Q}_{p}[\Diamond],\mathcal{Q}_{p}[\nxt,\Diamond],
	\mathcal{Q}[\Diamond],\mathcal{Q}[\nxt,\Diamond]\}$.
	Then $\mathcal{Q}$ has the PSP under the empty ontology.
\end{lemma}
\begin{proof}
	The proof for $\mathcal{Q}\in \{\mathcal{Q}_{p}[\Diamond],\mathcal{Q}_{p}[\nxt,\Diamond]\}$ is trivial: if $\Abox\models \varkappa(0)$ for some $\varkappa \in \mathcal{Q}$, then $\varkappa$ is clearly equivalent to a query in $\mathcal{Q}$ whose temporal depth does not exceed the maximal timestamp in $\Abox$, and so is of linear size in $\Abox$.
	
	For $\mathcal{Q}\in \{	\mathcal{Q}[\Diamond],\mathcal{Q}[\nxt,\Diamond]\}$, the argument is as follows. Assume that $\varkappa$ separates $(E^{+},E^{-})$. We may assume that $\varkappa$ is a conjunction of at most $|E^{-}|$-many queries in $\mathcal{Q}_{p}^{\circ}[\Diamond]$ of the form
	\begin{align*}
		\varkappa& = \rho_0 \land \Diamond (\rho_1 \land \Diamond (\rho_2 \land \dots \land \Diamond \rho_n) ),
	\end{align*}
	where every $\rho_{i}$ is a query in $\mathcal{Q}_{p}[\nxt]$. Then the conjuncts of $\varkappa$ are equivalent to queries in which $n$ does not exceed the maximal timestamps in data instances in $E^{+}$ and each $\rho_{i}$ is a
	query in $\mathcal{Q}_{p}[\nxt]$ whose temporal depth also does not exceed the maximal timestamps in data instances in $E^{+}$.
\end{proof}

We next complete the description of the polynomial-time algorithm solving  $\QBEba(\mathcal{Q}_{p}[\nxt,\Diamond])$ for $E^{+}=\{\mathcal{D}_{1}^{+},\mathcal{D}_{2}^{+}\}$ and
$E^{-}=\{\mathcal{D}_{1}^{-},\mathcal{D}_{2}^{-}\}$. The extension to arbitrary $E^{+},E^{-}$ is straightforward.
Recall that we assume that $\varkappa$ takes the form~\eqref{dnpath} with $\rho_{n}\ne \top$. Also recall that $S_{i,j}$ is the set of tuples $(k,\ell_1,\ell_2,n_{1},n_{2})$ such that
\begin{enumerate}
	\item $\ell_{1} \leq i\leq \max \mathcal{D}_{1}^{+}$,
	\item $\ell_{2} \leq j\leq \max \mathcal{D}_{2}^{+}$,
\end{enumerate}
and there is
$
\varkappa=\rho_0 \land \op_1 (\rho_1 \land \dots \land \op_k \rho_k)
$
for which
\begin{enumerate}
	\item there are satisfying assignments $f_{1},f_{2}$ in $\mathcal{D}_{1}^{+}$ and $\mathcal{D}_{2}^{+}$ with $f_{1}(k)=\ell_{1}$ and $f_{2}(k)=\ell_{2}$, respectively, and
	\item $n_{1}$ is minimal with a satisfying assignment $f$ for $\varkappa$ in $\mathcal{D}_{1}^{-}$ such that $f(k)=n_{1}$, and $n_{1}=\infty$ if there is no such $f$; $n_{2}$ is minimal with a satisfying assignment $f$ for $\varkappa$ in $\mathcal{D}_{2}^{-}$ such that $f(k)=n_{2}$, and $n_{2}=\infty$ if there is no such $f$.
\end{enumerate}
Then clearly there is a $\varkappa\in \mathcal{Q}_{p}[\nxt,\Diamond]$ separating
$(E^{+},E^{-})$ if there are $k,\ell_{1},\ell_{2}$ such that $(k,\ell_{1},\ell_{2},\infty,\infty) \in S_{\max \mathcal{D}_{1}^{+},\max \mathcal{D}_{2}^{+}}$.

So it suffices to compute $S_{\max \mathcal{D}_{1}^{+},\max \mathcal{D}_{2}^{+}}$
in polytime incrementally, starting with $S_{0,0}$. We have computed $S_{0,j}$ and $S_{i,0}$ already. Recall that $t_{\mathcal{D}}(i)= \{A \mid A(i)\in \mathcal{D}\}$.
To obtain $S_{i+1,j+1}$, we add to $S_{i+1,j}\cup S_{i,j+1}$ any tuple
$(k,\ell_{1},\ell_{2},n_{1},n_{2})$ for which there is $(k',\ell_{1}',\ell_{2}',n_{1}',n_{2}')\in S_{i+1,j}\cup S_{i,j+1}$
with $k'<k$, $\ell_{1}'<\ell_{1}\leq i+1$, $\ell_{2}'<\ell_{2}\leq j+1$
such that
$\ell_{1}=i+1$ or $\ell_{2}=j+1$ and, for $m=k-k'+1\geq 0$, we have $\ell_{1}-m>\ell_{1}'$, $\ell_{2}-m>\ell_{2}'$ and some sets of atoms $\rho_{1},\ldots,\rho_{m}$ with
$$
\rho_{1}\subseteq t_{\mathcal{D}_{1}^{+}}(\ell_{1}-m) \cap t_{\mathcal{D}_{2}^{+}}(\ell_{2}-m),
\dots, \rho_{m}\subseteq t_{\mathcal{D}_{1}^{+}}(\ell_{1}) \cap t_{\mathcal{D}_{2}^{+}}(\ell_{2})
$$
such that
\begin{itemize}
\item either $n_{1}$ is minimal with $n_{1}-n_{1}'>m$ and
$$
\rho_{1}\subseteq t_{\mathcal{D}_{1}^{-}}(n_{1}-m), \dots,
\rho_{m}\subseteq t_{\mathcal{D}_{1}^{-}}(n_{1})
$$
or, if no such $n_{1}$ exists, $n_{1}=\infty$, and
\item either $n_{2}$ is minimal with $n_{2}-n_{2}'>m$ and
$$
\rho_{1}\subseteq t_{\mathcal{D}_{2}^{-}}(n_{2}-m), \dots,
\rho_{m}\subseteq t_{\mathcal{D}_{2}^{-}}(n_{2})
$$
or, if no such $n_{2}$ exists, $n_{2}=\infty$.
\end{itemize}
Thus, we obtain $S_{i+1,j+1}$ from $S_{i+1,j}\cup S_{i,j+1}$ by adding any tuple that describes a query obtained from a query $\varkappa$ described by a tuple in $S_{i+1,j}\cup S_{i,j+1}$ by attaching the query $\Diamond(\rho_{1} \wedge \nxt(\rho_{2}\wedge\cdots \wedge \nxt \rho_{m}))$ with $m\geq 0$ to it. Clearly $S_{i+1,j+1}$ can be computed in polynomial time from $S_{i+1,j}$ and $S_{i,j+1}$. This finishes the proof for  $\QBEba(\mathcal{Q}_{p}[\nxt,\Diamond])$.

The proof for $\QBEba(\mathcal{Q}_{p}[\Diamond])$ is obtained by dropping $\nxt$ from the proof above. The \PTime{}-upper bound for QBE$(\mathcal{Q}[\Diamond])$ with a bounded number of positive examples can be proved in two steps:
(1) by Theorem~\ref{th:data-reduction} $(i.1)$ it suffices to prove the \PTime{}-upper bound for $\QBEba(\mathcal{Q}[\Diamond])$;
(2) by Theorem~\ref{th:data-reduction} $(i.3)$, $\QBEba(\mathcal{Q}[\Diamond])\leq_{p} \QBEba(\mathcal{Q}_{p}[\Diamond])$.
Finally, the \PTime{}-upper bound for QBE$(\mathcal{Q}[\nxt,\Diamond])$ with a bounded number of positive examples follows from the \PTime{}-upper bound for QBE$(\mathcal{Q}[\Diamond])$ with a bounded number of positive examples
by Theorem~\ref{th:data-reduction} $(ii.2)$).

\medskip

We now prove the results for query languages with $\U$ in Table~\ref{table:free}. We start with the \NP{}-lower bounds.
The \NP{}-lower bound for QBE$(\mathcal{Q}_{p}[\U])$ with a bounded number of positive and negative examples follows from the \NP{}-lower bound for
QBE$(\mathcal{Q}_{p}[\Diamond])$ with a bounded number of positive examples (shown above) and the second part of Theorem~\ref{th:data-reduction} $(ii.1)$ which reduces the number of negative examples from unbounded to a singleton.
The \NP{}-lower bound for QBE$(\mathcal{Q}[\Us])$ with a bounded number of negative examples follows from the \NP{}-lower bound for
QBE$(\mathcal{Q}[\Diamond])$ with a bounded number of negative examples (shown above) and Theorem~\ref{th:data-reduction} $(ii.2)$.
This completes the proof of the $\NP{}$-lower bounds.

The \NP{}-upper bound for QBE$(\mathcal{Q}_{p}[\U])$ follows from its PSP under the empty ontology which is proved in the same way as Lemma~\ref{lem:psp}:
\begin{lemma}\label{lem:pspU}
	$\mathcal{Q}_{p}[\U]$ has the PSP under the empty ontology.
\end{lemma}
\begin{proof}
If $\Abox\models \varkappa(0)$ for some $\varkappa \in \Qp[\U]$, then $\varkappa$ is clearly equivalent to a $\Qp[\U]$-query  whose temporal depth does not exceed the maximal timestamp in $\Abox$, and so is of linear size in $\Abox$.
\end{proof}

To obtain the \PTime{} and \PSpace{} upper bounds for queries with $\U$, we require the machinery and separability criterion that will be developed in the next section.


\section{Separability Criteria for $\mathcal{Q}_p[\U]$, $\mathcal{Q}[\Us]$, and $\mathcal{Q}[\U]$}\label{appC}

A \emph{transition system} is a tuple $S =  (\Sigma_1, \Sigma_2, W, L, R, W_0)$, where $\Sigma_1$ (respectively, $\Sigma_2$) is a \emph{state} (respectively, \emph{transition}) \emph{label alphabet}, $W$ is a set of \emph{states} and $W_0 \subseteq W$ is a set of \emph{initial} states. A \emph{state labelling}, $L$, is a map $W \to \Sigma_1$; a \emph{transition labelling}, $R$, is a \emph{partial} map $W \times W \to \Sigma_2$. We write $s \to_b s'$, for $s, s' \in W$, if $R(s,s')=b$ and we write $s \to s'$ if $R(s,s')$ is defined. A \emph{run} or \emph{computation} on $S$ is a finite sequence $\mathfrak s = s_0 \to s_1 \to \dots \to s_n$, for $n \geq 0$, such that $s_{i-1} \to s_i$ for all $i$ and $s_0 \in W_0$. A computation tree $\mathfrak T_S$ of $S$ is (an infinite) tree---forest, to be more precise---in which the vertices are runs $\mathfrak s$ on $S$ and the successor relation is $\mathfrak s \to \mathfrak s'$ for all $\mathfrak s = s_0 \to \dots \to s_n$ and $\mathfrak s' = s_0 \to \dots \to s_n \to s_{n+1}$. The vertices $\mathfrak s$ of the tree are labelled with $L(s_n)$, while the edges $\mathfrak s \to \mathfrak s'$ are labelled with $b$ such that $s_n \to_b s_{n+1}$. A tree $\mathfrak T$ is a \emph{subtree} of $\mathfrak T_S$ if the set of vertices of $\mathfrak T$ is a \emph{convex} subset of the set of vertices of $\mathfrak T_S$ containing a root (from $W_0$).

Let $S$ be a transition system such that $\Sigma_1 = 2^\Sigma$ and $\Sigma_2 = 2^{\Sigma \cup \{\bot\}}$ for some signature $\Sigma$. We then say that $S$ is a transition system \emph{over the signature} $\Sigma$. For a pair $S, T$ of transition systems over $\Sigma$, we say that $S$ is \emph{simulated by} $T$ if every finite subtree $\mathfrak T'$ of $\mathfrak T_S$ is homomorphically embeddable into $\mathfrak T_T$, i.e., there is a map $h$ from the set of vertices of $\mathfrak T'$ to the set of vertices of $\mathfrak T_T$ such that $(i)$ $\mathfrak s$ is labelled by $a$ implies $h(\mathfrak s)$ is labelled by $a' \supseteq a$, $(ii)$ $\mathfrak s \to \mathfrak s'$ in $\mathfrak T'$ labelled by $b$ implies $h(\mathfrak s) \to h(\mathfrak s')$ is in $\mathfrak T_T$ and labelled by $b' \supseteq b$. We say that $S$ is \emph{contained} in $T$ if every finite path in $\mathfrak T_S$, (i.e., a run in $S$) is homomorphically embeddable into $\mathfrak T_T$.

Let $T = (\Sigma_1, \Sigma_2, W', L', R', W_0')$. We define the \emph{direct product} (aka synchronous composition) of $S$ and $T$ as a transition system $S \times T = (\Sigma_1, \Sigma_2, W'', L'', R'', W_0'')$ with $W'' = W \times W'$, $W_0'' = W_0 \times W_0'$, $L''((s,s')) = L(s) \cap L'(s')$ for all $(s,s') \in W''$. Then $R''((s, s'), (t, t'))$ is defined iff both $R(s,t)$ and $R'(s', t')$ are defined, in which case $R''((s, s'), (t, t')) = R(s,t) \cap R'(s', t')$. The \emph{disjoint union} of $S$ and $T$ is a transition system $S \uplus T = (\Sigma_1, \Sigma_2, W'', L'', R'', W_0'')$ that is obtained by renaming states in $T$ if necessary to make $W$ and $W'$ disjoint, and then taking $W'' = W \cup W'$, $L'' = L \cup L'$, $R'' = R \cup R'$, and $W_0'' = W \cup W'$. The definitions of the product and disjoint union are straightforwardly extended to a collection of transition systems $S_1, \dots, S_n$.

\subsection{Representations for $\Q[\Us]$}

Let $\varkappa$ be a $\Q[\Us]$-query over a signature $\Sigma$. We can naturally associate $\varkappa$ with a tree $\mathfrak T_\varkappa$ as follows. Let $\varkappa = \varrho_0' \land \bigwedge_i (\varrho_i \U \psi_i)$, where $\varrho_0', \varrho_i$ is a conjunction of $\Sigma$-atoms, $\psi_i = \varrho'_i \land \bigwedge_j (\varrho_j \U \psi_j)$, and $\varrho_i'$ is a conjunction of $\Sigma$-atoms. We do not distinguish between a conjunction and a set of atoms. The root of the tree is $r$ and the tree has $r \to (\varrho_i \U \psi_i)$ for each $i$, i.e., there are vertices $\varrho_i \U \psi_i$. The root $r$ is labelled with $\varrho'_0$ and each $\varrho_i \U \psi_i$ is labelled with $\varrho_i'$. Each edge $r \to (\varrho_i \U \psi_i)$ is labelled with $\varrho_i$. The tree then contains $(\varrho_i \U \psi_i) \to (\varrho_j \U \psi_j)$ for each $j$, where each such edge is labelled with ${\varrho_j}$ and $(\varrho_j \U \psi_j)$ is labelled with the set of atoms of $\psi_j$, and so on. Thus, we will treat any $\varkappa \in \Q[\Us]$ as a tree. The other way round, every finite tree $\mathfrak T$ with vertices labelled with subsets of $\Sigma$ and edges labelled with subsets of $\Sigma \cup \{\bot\}$ corresponds to a $\Q[\Us]$-query. Indeed, let $x$ be any leaf of $\mathfrak T$. Then we define a $\Q[\Us]$-query $\varkappa_x = \bigwedge \varrho$, where $\varrho$ is a label of $x$. Suppose now we have $x \to y_i$, for $i \in I$ in $\mathfrak T$, and the label of $x$ is $\varrho$ while the label of $x \to y_i$ is $\varrho_i$. We define $\varkappa_x = (\bigwedge \varrho) \land \bigwedge_{i \in I} ((\bigwedge \varrho_i) \U \varkappa_{y_i})$. The query $\varkappa_r$, where $r$ is the root of $\mathfrak T$, is the required query representing $\mathfrak T$. We denote it by $\varkappa_{\mathfrak T}$.

Let $\Abox$ be a data instance and $\TO$ an $\LTL{}$-ontology over a signature $\Sigma$. Let $S$ be a transition system over $\Sigma$. We say that $S$ \emph{represents} $\TO, \Abox$ if the following conditions hold: $(i)$ $\TO, \Abox \models \varkappa_{\mathfrak T'}$ for every finite subtree $\mathfrak T'$ of $\mathfrak T_S$; $(ii)$ $\mathfrak T_\varkappa$ is homomorphically embeddable into $\mathfrak T_S$ for each $\varkappa$ with $\TO, \Abox \models \varkappa$.

\begin{lemma}\label{lemmain}
Let $E = (E^+, E^-)$, $E^+ = \{\Abox_i \mid i \in I^+ \}$, $E^- = \{\Abox_i \mid i \in I^- \}$, and let $\TO$ be an $\LTL{}$ ontology. Let $S^i$ represent $\TO, \Abox_i$, for $i \in I^+ \cup I^-$. Then $(i)$ $E$ is not $\Q[\Us]$-separable under $\TO$ iff $\prod_{i \in I^+} S^i$ is simulated by $\uplus_{i \in I^-} S^i$\textup{;} $(ii)$ $E$ is not $\Qp[\U]$-separable under $\TO$ iff $\prod_{i \in I^+} S^i$ is contained in $\uplus_{i \in I^-} S^i$.
\end{lemma}
\begin{proof}
We show $(i)$. For $(\Rightarrow)$, suppose $S^+ = \prod_{i \in I^+} S^i$ is not simulated by $S^- = \uplus_{i \in I^-} S^i$. It follows that there exists a finite subtree $\mathfrak T$ of $\mathfrak T_{S^+}$ that is not homomorphically embeddable into $\mathfrak T_{S^-}$. We claim that $\varkappa_{\mathfrak T}$ separates $E$ under $\TO$. First, we show that $\TO, \Abox_i \models \varkappa_{\mathfrak T}$ for each $i \in I^+$. Indeed, for any such $i$, let $\mathfrak T^i$ be a projection of $\mathfrak T$ to the runs of $S^i$. Clearly, $\mathfrak T^i$ is a finite subtree of $\mathfrak T_{S^i}$ and $\varkappa_{\mathfrak T^i} \models \varkappa_{\mathfrak T}$ ($\mathfrak T$ is homomorphically embeddable into $\mathfrak T_i$). Because $S^i$ represents $\TO, \Abox_i$, we obtain $\TO, \Abox_i \models \varkappa_{\mathfrak T}$. Second, we show that $\TO, \Abox_i \not \models \varkappa_{\mathfrak T}$ for each $i \in I^-$. For the sake of contradiction, suppose $\TO, \Abox_i \models \varkappa_{\mathfrak T}$ for some such $i$. Because $S^i$ represents $\TO, \Abox_i$, it follows that $\mathfrak T$ is homomorphically embeddable into $\mathfrak T_{S^i}$, and so $\mathfrak T$ is homomorphically embeddable into $\mathfrak T_{S^-}$, which is a contradiction. The proofs of $(\Rightarrow)$ and $(ii)$ are similar.
\end{proof}

\paragraph{Constructing representations for queries without an ontology.} Given $\Abox$, we construct a transition system $S$ with the states $0, \dots, (\max \D+1)$, where $(\max \D+1)$ is labelled with $\emptyset$ and the remaining $j$ by $\{A \mid A(j) \in \D \}$. Transitions are $j \to k$, for $0 \leq j < k \leq \max \D+1$, that are labelled by $\{ A \in \Sigma \cup \{\bot\} \mid A(n) \in \D, n \in (j,k)  \}$ and $(\max \D+1) \to (\max \D+1)$ with label $\Sigma^\bot =\Sigma \cup \{\bot\}$.

\begin{lemma}
$S$ represents $\emptyset, \Abox$.
\end{lemma}
\begin{proof}
First, we show that $\Abox \models \varkappa_{\mathfrak T'}$ for every finite subtree $\mathfrak T'$ of $\mathfrak T_S$. Let $\I_\D$ be an $\LTL{}$ interpretation such that $\I_\Abox, n \models A$ iff $A(n) \in \D$, for any atom $A$. We observe that any $\LTL{}$ interpretation $\I$ can be viewed as a transition system with the states $n \in \mathbb N$ that are labelled with (sets of) atoms $A$ holding at $n$. The transitions hold between any pair of states $n < m$ and each such transition is labelled with atoms $A$ or $\bot$ that hold at each $i \in (n,m)$. It is clear that $\Abox \models \varkappa_{\mathfrak T'}$ iff $\mathfrak T'$ is homomorphically embeddable into $\mathfrak T_{\I_\Abox}$. We define an embedding $h$ of $\mathfrak T'$ into $\I_\Abox$ as follows. We set $h(0) = 0$. Suppose $h(\mathfrak s)$ for $\mathfrak s = 0 \to s_1 \to \dots \to s_n$ has been defined and let $\mathfrak s' = 0 \to s_1 \to \dots \to s_{n+1}$. If $s_{n+1} < \max \D+1$, then we set $h(\mathfrak s') = s_{n+1}$. If $s_{n+1} = \max \D+1$, then we set $h(\mathfrak s') = h(\mathfrak s')+1$. Clearly, $h$ is a homomorphism. Therefore, $\mathfrak T'$ is homomorphically embeddable into $\mathfrak T_{\I_\Abox}$ and $\Abox \models \varkappa_{\mathfrak T'}$.

Second, we show that $\mathfrak T_\varkappa$ is homomorphically embeddable into $\mathfrak T_S$ for each $\varkappa$ such that $\Abox \models \varkappa$. Take any $\varkappa$ such that $\Abox \models \varkappa$. It follows that $\mathfrak T_\varkappa$ is homomorphically embeddable into $\mathfrak T_{\I_\Abox}$. It remains to observe that $\mathfrak T_{\I_\Abox}$ is homomorphically embeddable into $\mathfrak T_S$ (the definition of $h$ is left to the reader).
\end{proof}

The criterion of Theorem~\ref{criterionUsA} now follows immediately from the two previous lemmas. The remaining \PTime{} and \PSpace{} upper bounds from Table~\ref{table:free} are explained in the main part of the paper.

\paragraph{Constructing representations for queries with an $\LTL_\horn\Xallop$-ontology.}

Let $\Abox$ be a data instance and $\TO$ an $\LTL_\horn\Xallop$-ontology. Let $\C_{\TO, \D}$ be the canonical model of $\TO, \Abox$ and $s_{\TO, \Abox}$, $p_{\TO, \D}$ the numbers from Proposition~\ref{prop:period}. We define $S$ with the states $\{0, \dots, \max \Abox+s_{\TO, \Abox}+p_{\TO, \Abox}-1\}$. The label of each state $n$ is $\{ A \in \Sigma \mid \C_{\TO, \Abox}, n \models A \}$. There are transitions from $n$ to $m$, for each pair of states $n < m$ labelled with $\{ A \in \Sigma \cup \{ \bot \} \mid \C_{\TO, \Abox}, k \models A \text{ for all }k \in (n, m)\}$. Moreover, there are transitions from $n$ to $m$, for $n,m \in [\max \Abox+s_{\TO, \Abox}, \max \Abox+s_{\TO, \Abox}+p_{\TO, \Abox})$ such that $n \geq m$. A label for such a transition is $\{ A \in \Sigma \cup \{ \bot \} \mid \C_{\TO, \Abox}, k \models A \text{ for all }k \in (n, \max \Abox+s_{\TO, \Abox}+p_{\TO, \Abox}) \cup [\max \Abox+s_{\TO, \Abox}, m)\}$.

\begin{lemma}\label{th:horn+qus-repr}
$S$ represents $\TO, \Abox$.
\end{lemma}
\begin{proof}
First, we show that $\Abox \models \varkappa_{\mathfrak T'}$ for every finite subtree $\mathfrak T'$ of $\mathfrak T_S$. It is clear from the properties of $\C_{\TO, \Abox}$ (see Section~\ref{sec:horn}) that $\Abox \models \varkappa_{\mathfrak T'}$ iff $\mathfrak T'$ is homomorphically embeddable into $\mathfrak T_{\C_{\TO, \Abox}}$. We define an embedding $h$ of $\mathfrak T'$ into $\C_{\TO, \Abox}$ as follows. We set $h(0) = 0$. Suppose $h(\mathfrak s)$ for $\mathfrak s = 0 \to s_1 \to \dots \to s_n$ has been defined and let $\mathfrak s' = 0 \to s_1 \to \dots \to s_{n+1}$. If $s_{n+1} > s_n$, we set $h(\mathfrak s') = h(\mathfrak s)+(s_{n+1} - s_n)$. Otherwise, we set $h(\mathfrak s') = h(\mathfrak s)+ p_{\TO, \Abox} - (s_{n} - s_{n+1})$. It is readily verified that $h$ is a homomorphism. Therefore, $\mathfrak T'$ is homomorphically embeddable into $\mathfrak T_{\C_{\TO, \Abox}}$ and $\TO, \Abox \models \varkappa_{\mathfrak T'}$.

Second, we show that $\mathfrak T_\varkappa$ is homomorphically embeddable into $\mathfrak T_S$ for each $\varkappa$ such that $\TO, \Abox \models \varkappa$. Take any $\varkappa$ such that $\TO, \Abox \models \varkappa$. It follows that $\mathfrak T_\varkappa$ is homomorphically embeddable into $\mathfrak T_{\C_{\TO,\Abox}}$. It remains to show that $\mathfrak T_{\C_{\TO, \Abox}}$ is homomorphically embeddable into $S$. To this end, we define a map $r \colon \mathbb N \to [0, \max \Abox+s_{\TO, \Abox}+p_{\TO, \Abox})$ by setting $r(n) = n$ if $n \in [0, \max \Abox+s_{\TO, \Abox}+p_{\TO, \Abox})$ and $r(n) = ((n - \max \Abox-s_{\TO, \Abox}) \text{ mod } p_{\TO, \Abox}) + \max \Abox+s_{\TO, \Abox}$, otherwise. Now, we define $h(\mathfrak s)$ for $\mathfrak s = 0 \to s_1 \to \dots \to s_n$ (note that $s_i \in \mathbb N$ and $s_{i+1} > s_i$) to be equal to $r(s_{n})$. It is readily verified that $h$ is a homomorphism from $\mathfrak T_{\C_{\TO, \Abox}}$ into $S$.
\end{proof}

Now we can explain the data complexity upper bounds from Theorem~\ref{thm:datahorn} (the upper bounds from Theorem~\ref{thm:Horn-U} are explained in the main paper) for $\QBE(\LTL_\horn\Xallop, \mathcal{Q})$ with $\mathcal Q \in \{\Q[\Us], \Qp[\U]\}$. The result for $\QBEbp(\LTL_\horn\Xallop,\Q[\Us])$ (and so for $\QBEba(\LTL_\horn\Xallop,\Q[\Us])$) follows from the fact that $\mathfrak P = \prod_{i \in I^+} S^i$ and $\mathfrak U = \uplus_{i \in I^-} S^i$ are constructible in $\PTime{}$ in the size of $E$. This is not the case for $\QBE(\LTL_\horn\Xallop,\Q[\Us])$; here, we use the observation that if there exists a finite subtree $\mathfrak T$ of $\mathfrak T_{\mathfrak P}$ that is not homomorphically embeddable into $\mathfrak U$, then there exists such $\mathfrak T$ satisfying the property that every $s$ from $\mathfrak P$ occurs on each path of $\mathfrak T$ at most $|\mathfrak U|$-many times. 
Let $N^+ = \prod_{i \in I^+}p_{\TO, \Abox_i}$. We claim that $\mathfrak T$ is a subtree of $\mathfrak T_{\mathfrak P}^M$ for $M = \max_{i \in I^+} \{ \max \D_i + s_{\TO, \Abox_i} \} + N^+ |\mathfrak U|$. Indeed, in the required $\mathfrak T$, if there is a path that is longer than $M$, the property above would be violated. By our construction of $S^i$, any $n$-th element, for $n \geq \max_{i \in I^+} \{ \max \D_i + s_{\TO, \Abox_i} \}$, of any path of $\mathfrak T_{\mathfrak P}$ is of the form $(t_1, \dots, t_{|I^+|})$, where $t_i \in [\max \D_i + s_{\TO, \Abox_i}, \max \D_i + s_{\TO, \Abox_i} + p_{\TO, \Abox_i})$. Observe that $(t_1, \dots, t_{|I^+|}) \to (s_1, \dots, s_{|I^+|})$ in $\mathfrak P$, for $(t_1, \dots, t_{|I^+|})$ as above, implies $s_i \in [\max \D_i + s_{\TO, \Abox_i}, \max \D_i + s_{\TO, \Abox_i} + p_{\TO, \Abox_i})$ and any sequence $(t_1, \dots, t_{|I^+|}) \to \dots \to (s_1, \dots, s_{|I^+|})$ as above in $\mathfrak P$ longer than $N^+ |\mathfrak U|$ will have some $(t_1, \dots, t_{|I^+|})$ repeated more than $|\mathfrak U|$ times.
Thus, in order to decide $\QBE(\LTL_\horn\Xallop,\Q[\Us])$, we need to check if $\mathfrak T_{\mathfrak P}^M$ is homomorphically embeddable into $\mathfrak U$. The latter can be checked by constructing $\smash{\mathfrak T_{\mathfrak P}^M}$ branch-by-branch while checking all possible embeddings of these branches into $\mathfrak U$. Since $M$ is polynomial in $E$, this algorithm works in $\PSpace{}$ in the size of $E$.

It remains to explain the $\NP{}$ upper bound for $\QBE(\LTL_\horn\Xallop,\Qp[\U])$. From Lemma~\ref{lemmain} and the argument above, it follows that $E$ is separable under $\TO$ iff there exists a path in $\mathfrak T_{\mathfrak P}^M$ that is not embeddable into $\mathfrak U$. Such a path (if exists) is of the size polynomial in $E$. Embeddability of such a path into $\mathfrak U$ can be checked in \PTime{} from $E^-$ and the size of the path.

\paragraph{Constructing representations for queries with an $\LTL{}$-ontology.}

Let $\Abox$ be a data instance and $\TO$ an $\LTL{}$-ontology. We can assume that $\max \Abox = 0$. Indeed, for a given $\TO$ and $E$, we can construct in polytime an $\LTL{}$-ontology $\TO'$ and $E'$ such that $\max \Abox' = 0$ for each $\Abox'$ in $E'$ and $E$ is $\Q[\Us]$-separable under $\TO$ iff $E'$ is $\Q[\Us]$-separable under $\TO'$. Let $\avec{T}_\TO$ be the set of $\TO$-types. For $\avec{T} \subseteq \avec{T}_\TO$, we say that $\avec{T}$ is \emph{realisable in} $\TO, \D$ if there are instants $n_\I$ in all models $\I$ of $\TO, \Abox$ such that $\{\tp_\I(n_\I) \mid \I \models \TO, \Abox\} = \avec{T}$. For $\avec{T}_1, \avec{T}_2 \subseteq \avec{T}_\TO$ realisable in $\TO, \D$ and $\Gamma \subseteq \Sigma \cup \{\bot\}$, we define $\avec{T}_1 \to_\Gamma \avec{T}_2$ if there are instants $n_\I < m_\I$ in all models $\I$ of $\TO, \Abox$ such that $(i)$ $\{\tp_\I(n_\I) \mid \I \models \TO, \Abox\} = \avec{T}_1$, $(ii)$ $\{\tp_\I(m_\I) \mid \I \models \TO, \Abox\} = \avec{T}_2$, $(iii)$ $\{ \tp_\I(k) \mid \I \models \TO, \Abox, n_\I < k < m_\I \} \cap \avec{T}_2 = \emptyset$, $(iv)$ $\Gamma = \{ A \in \Sigma^\bot \mid \I,k \models A \text{ for all } \I \models \TO, \Abox, n_\I < k < m_\I \}$. We observe that condition $(iii)$ ensures that there exists at most one $\Gamma$ for given $\avec{T}_1, \avec{T}_2$ such that $\avec{T}_1 \to_\Gamma \avec{T}_2$. We define $S$ with the states $\avec{T} \subseteq \avec{T}_\TO$ realisable in $\TO, \Abox$. A single initial state of $S$ is $\avec{T}_0 = \{\tp_\I(0) \mid \I \models \TO, \Abox \}$. The states $\avec{T}$ are labelled with $\{ A \in \Sigma \mid A \in \tp \text{ for all }\tp \in \avec{T}\}$. There are transitions from $\avec{T}_1$ to $\avec{T}_2$ where $\avec{T}_1 \to_\Gamma \avec{T}_2$ holds for some $\Gamma$, labelled with $\Gamma$.

\begin{lemma}\label{th:bool-repr}
$S$ represents $\TO, \Abox$.
\end{lemma}
\begin{proof}
Let $\avec{\I} = \{\I \mid \I \models \TO, \Abox \}$. We treat every $\I \in \avec{\I}$ as a transition system as we did above. Take the product $\prod_{\I \in \avec{\I}} \I$ and denote it (slightly abusing notation) by $\avec{\I}$. Note that the states $s$ of $\avec{\I}$ are maps $s \colon \avec{\I} \to \mathbb N$.

First, we show that $\TO, \Abox \models \varkappa_{\mathfrak T'}$ for every finite subtree $\mathfrak T'$ of $\mathfrak T_S$. It should be clear that $\TO, \Abox \models \varkappa_{\mathfrak T'}$ iff $\mathfrak T'$ is homomorphically embeddable into $\mathfrak T_{\avec{\I}}$. We define an embedding $h$ of $\mathfrak T'$ into $\avec{\I}$ as follows. We set $h(\avec{T}_0) = s_0$, where $s_0$ is the initial state of $\avec{\I}$ satisfying $s_0(\I) = 0$ for every $\I \in \avec{\I}$. Suppose $h(\mathfrak s)$ for $\mathfrak s = \avec{T}_0 \to \avec{T}_1 \to \dots \to \avec{T}_n$ has been defined equal to $s$. Our induction hypothesis will be that $\avec{T}_n = \{ \tp_\I(s(\I)) \mid \I \in \avec{\I}\}$. It can be readily verified that it holds for $\mathfrak s = \avec{T}_0$. Let $\mathfrak s' = \avec{T}_0 \to \avec{T}_1 \to \dots \to \avec{T}_{n+1}$ and $\avec{T}_{n} \to_\Gamma \avec{T}_{n+1}$. For each $\I \in \avec{\I}$, we select $m_\I > s(\I)$ such that $\avec{T}_{n+1} = \{ \tp_\I(m_\I)) \mid \I \in \avec{\I}\}$ and $\Gamma = \{ A \in \Sigma^\bot \mid \I,k \models A \text{ for all } \I \in \avec{\I}, s(\I) < k < m_\I \}$. That this selection is always possible follows from the IH. We then set $h(\mathfrak s') = s'$ such that $s'(\I) = m_\I$ for $\I \in \avec{\I}$.

Now, we show that $\mathfrak T_\varkappa$ is homomorphically embeddable into $\mathfrak T_S$ for each $\varkappa$ such that $\TO, \Abox \models \varkappa$. Take any $\varkappa$ such that $\TO, \Abox \models \varkappa$. It follows that $\mathfrak T_\varkappa$ is homomorphically embeddable into $\mathfrak T_{\avec{\I}}$. It remains to show that $\mathfrak T_{\avec{\I}}$ is homomorphically embeddable into $S$. We define $h$ so that $h(s_0) = \avec{T}_0$. Consider now $\mathfrak s_1 = s_0 \to s_1$. Instead of $s_1$, we can always select $s_1'$ such that $s_0 \to s_1'$, $\{ \tp_\I(s_1'(\I)) \mid \I \in \avec{\I}\} \subseteq \{ \tp_\I(s_1(\I)) \mid \I \in \avec{\I}\}$, $\{ \tp_\I(k) \mid \I \in \avec{\I}, s_0(\I) < k < s_1'(\I) \} \cap \{ \tp_\I(s_1'(\I)) \mid \I \in \avec{\I}\} = \emptyset$ and, finally, $\{ \tp_\I(k) \mid \I \in \avec{\I}, s_0(\I) < k < s_1'(\I) \} \subseteq \{ \tp_\I(k) \mid \I \in \avec{\I}, s_0(\I) < k < s_1(\I) \}$. We define $h(\mathfrak s_1) = \{ \tp_\I(s_1'(\I)) \mid \I \in \avec{\I}\}$. It can be readily verified, given the subsumptions above, that (subsumption of the) label of the edge from $s_0$ to $\mathfrak s_1$ is preserved under $h$ and (subsumption of) the node $\mathfrak s_1$ is preserved under $h$.

Consider now $\mathfrak s_2 = s_0 \to s_1 \to s_2$. We can always select $s_2'$ so that $s_1' \to s_2'$, $\{ \tp_\I(s_2'(\I)) \mid \I \in \avec{\I}\} \subseteq \{ \tp_\I(s_2(\I)) \mid \I \in \avec{\I}\}$, $\{ \tp_\I(k) \mid \I \in \avec{\I}, s_1'(\I) < k < s_2'(\I) \} \cap \{ \tp_\I(s_2'(\I)) \mid \I \in \avec{\I}\} = \emptyset$ and, finally, $\{ \tp_\I(k) \mid \I \in \avec{\I}, s_1'(\I) < k < s_2'(\I) \} \subseteq \{ \tp_\I(k) \mid \I \in \avec{\I}, s_1(\I) < k < s_2(\I) \}$. We define $h(\mathfrak s_2) = \{ \tp_\I(s_2'(\I)) \mid \I \in \avec{\I}\}$.

Clearly, we can extend this argument to arbitrarily-many steps to define $h$ for any $\mathfrak s = s_0 \to \dots \to s_n$. This completes the proof of the lemma.
\end{proof}


\subsection{Representations for $\Q[\U]$}

Let $\varkappa$ be a $\Q[\U]$-query over a signature $\Sigma$. Let all the subformulas of $\varkappa$, which are either conjunctions (sets) of atoms $\gamma, \lambda$, or $\varphi \U \psi$, or  conjunctions thereof, be enumerated. We assume that there are no subformulas of the form $\varphi \U \varphi$ (such formulas are equivalent to $\bot \U \varphi$).
We can associate $\varkappa$ with a tree $\mathfrak T_\varkappa$ having edges of two types: black and red. Let $\varkappa = \gamma_0 \land \bigwedge_{i \in J_0} \varkappa_i$ and let $\varkappa_i = (\lambda_i \land \bigwedge_{j \in I_i} \varkappa_j) \U (\gamma_i \land \bigwedge_{j \in J_i} \varkappa_j))$, where $\varkappa_j$, for $j \in I_i \cup J_i$, is of the form $\varphi \U \psi$. Then the root $r$ of the tree is labelled with $\gamma_0$ and there are black edges $r \to \varkappa_i$, for $i \in J_0$. Each such edge is labelled with $\lambda_i$ and node $\varkappa_i$, $i \in J_0$, is labelled with $\gamma_i$. Now take any $i \in J_0$. There is a black edge $\varkappa_i \to \varkappa_{j}$, for each $j \in J_i$, and there is a red edge $\varkappa_i \to \varkappa_{j}$, for each $j \in I_i$. To define the label of each such black or red $\varkappa_i \to \varkappa_{j}$ edge and the label of the corresponding $\varkappa_{j}$, we look at the form of the $\varkappa_j$. Suppose $\varkappa_j = (\lambda_j \land \bigwedge_{k \in I_j} \varkappa_k) \U (\gamma_j \land \bigwedge_{k \in J_j} \varkappa_k))$. Then the label of $\varkappa_i \to \varkappa_{j}$ is $\lambda_j$ and the label of $\varkappa_{j}$ is $\gamma_j$. The construction of the edges $\varkappa_j \to \varkappa_k$, their labels, and the the construction of the subsequent tree is done analogously (by treating $j$ as $i$ in the previous construction). Thus, we can and will treat any $\varkappa \in \Q[\Us]$ as the tree $\mathfrak T_\varkappa$.

Let a \emph{black/red tree} be a tree where each edge has either black or red colour, but not both.
Every finite black/red tree $\mathfrak T$, where vertices are labelled with subsets of $\Sigma$ and edges are labelled with subsets of $\Sigma \cup \{\bot\}$, corresponds to a $\Q[\U]$-query. Indeed, let $x \to y$ be any edge such that $y$ is a leaf of $\mathfrak T$. Then we define a $\Q[\U]$-query $\varkappa_{x \to y}$ as $\lambda \U \gamma$, where $\lambda$ is the label of $x \to y$ while $\gamma$ is the label of $y$. Suppose now we have black (respectively, red) transitions $x \to y_i$, for $i \in J$ (respectively, $i \in I$), for an edge $z \to x$ in $\mathfrak T$ labelled with $\lambda$ for $x$ labelled with $\gamma$. We define $\varkappa_{z \to x} = (\lambda \land \bigwedge_{i \in I} \varkappa_{x \to y_i}) \U (\gamma \land \bigwedge_{i \in J} \varkappa_{x \to y_i})$. Let $r$ be the root of $\mathfrak T$ labelled with $\gamma$. Then the required $\varkappa_\mathfrak T$ representing $\mathfrak T$ is $\gamma \land \bigwedge_{r \to y \text{ in }\mathfrak T} \varphi_{r \to y}$.

Further, we define a \emph{black/red transition system} $S$ by adding either black or red colour, but not both, to each transition $s \to s'$ of the transition system $S$ defined above. The computation tree $\mathfrak T_S$ of $S$ is defined as before, however, every edge $\mathfrak s \to \mathfrak s'$ in $\mathfrak T_S$ has either red or black (but not both) colour that is equal to the colour of $s_n \to s_{n+1}$. In the definition of the direct product $S \times T$, we now require that $R''((s, s'), (t, t'))$ is red (respectively, black) iff both $R(s,t)$ and $R'(s', t')$ are red (respectively, black) (the labels are defined as before). Finally, in the definition of a homomorphic embedding of a black/red (labelled) tree $\mathfrak T'$ to another such tree $\mathfrak T''$ we require, additionally, that $h(\mathfrak s) \to h(\mathfrak s')$ is black (respectively, red) in $\mathfrak T''$ if $\mathfrak s \to \mathfrak s'$ is black (respectively, red) in $\mathfrak T'$. For a data instance $\Abox$, an $\LTL{}$ ontology $\TO$, a signature $\Sigma$, and a black/red transition system $S$, the definition of $S$ representing $\TO, \Abox$ continues to hold (with $\Q[\Us]$ changed to $\Q[\U]$ in  $(ii)$). Moreover, the same proof as in Lemma~\ref{th:repr-char} $(i)$ shows that we have:
\begin{lemma}\label{th:repr-char}
Let $E = (E^+, E^-)$, $E^+ = \{\Abox_i \mid i \in I^+ \}$, $E^- = \{\Abox_i \mid i \in I^- \}$, and let $\TO$ be an $\LTL{}$ ontology. Let $S^i$ represent $\TO, \Abox_i$, for $i \in I^+ \cup I^-$. Then $E$ is not $\Q[\U]$-separable under $\TO$ iff $\prod_{i \in I^+} S^i$ is simulated by $\uplus_{i \in I^-} S^i$.
\end{lemma}

\paragraph{Constructing representations for queries without an ontology.}

Let $\avec{d},\avec{e} \subseteq \mathbb N$ be finite and nonempty. For any $d \in \avec{d}$, let $\mu(d) = \min \{e \in \avec{e} \mid d < e \}$. If $\mu$ is a surjective $\avec{d} \to \avec{e}$ function, we write $\avec{d} \sm \avec{e}$ and set
$$
\unt(\avec{d}, \avec{e}) ~=~ \bigcup_{d\in\avec{d}} \{d' \in \mathbb N \mid d < d' < \mu(d)\}.
$$

\begin{example}\em
Let $\avec{d} = \{1,2,3\}$ and $\avec{e} = \{3,4\}$. Then $\avec{d} \sm \avec{e}$ with $\unt(\avec{d}, \avec{e}) = \{2\}$. However, for $\avec{d} = \{1,2\}$ and $\avec{e} = \{3,4\}$, we have neither $\avec{d} \sm \avec{e}$ (because $\mu$ is not a surjection) nor $\avec{e} \sm \avec{d}$ (because $\mu$ is not defined).
\end{example}
Given a data instance $\Abox$, $\avec{e} \subseteq \mathbb N$ and an atom $A$ (possibly $\bot$), we write $\Abox, \avec{e} \models A$ if $A(e) \in \Abox$ for all $e \in \avec{e}$.
We construct a black/red transition system $S$ with a set of states $\{ 0, z, u \} \cup \{ \avec{d}\avec{e} \mid \avec{d},\avec{e} \subseteq \{0, \dots, \max \Abox \}\}$. The label of $0$ is $\{ A \mid A(0) \in \Abox\}$, the label of $z$ is $\emptyset$, the label of $u$ is $\Sigma^\bot$, and the label of $\avec{d}\avec{e}$ is $\{ A \mid A(e) \in \Abox \text{ for all }e \in \avec{e} \}$.
The alphabet of the transition labels is $2^{\Sigma \cup \{\bot\}}$. From $0$, we have
\begin{itemize}
\item[$(i)$] a black transitions to every $\avec{e}\avec{d}$ such that $\{0\} \sm \avec{d}$ (this implies that $|\avec{d}| = 1$) and $\avec{e} = \unt(\{ 0 \}, \avec{d})$, labelled with the set $\{ A \in \Sigma^\bot \mid \Abox, \avec{e} \models A\}$.

\end{itemize}
From each $\avec{e}\avec{d}$, we have
\begin{itemize}
\item[$(ii)$] a black transition to every $\avec{f}\avec{g}$ such that $\avec{d} \sm \avec{g}$ and $\avec{f} = \unt(\avec{d}, \avec{g})$, labelled with $\{ A \in \Sigma^\bot \mid \Abox, \avec{f} \models A\}$,
\item[$(iii)$] a red transition to every $\avec{f}\avec{g}$ such that $\avec{e} \sm \avec{g}$ and $\avec{f} = \unt(\avec{e}, \avec{g})$, labelled with $L = \{A \in \Sigma^\bot \mid \Abox, \avec{f} \models A\}$.
\end{itemize}
The state $z$ has a black and a red transition to itself labelled with $\Sigma^\bot$ and the same holds for $u$. We have a black transition to $z$ from every $\avec{e}\avec{d}$ labelled with $\{ A \in \Sigma \cup \{\bot\} \mid \Abox, \{ \max \avec{d}, \dots, \max \Abox - 1 \} \models A\}$, and we have a red transition to $z$ from every $\avec{e}\avec{d}$ labelled with $\{ A \in \Sigma^\bot \mid \Abox, \{ \max \avec{e}, \dots, \max \Abox - 1 \} \models A\}$. Finally, we have a red transition from every $\emptyset\avec{d}$ to $u$ as well as from $z$ to $u$ labelled with $\Sigma^\bot$.

\begin{lemma}
$S$ represents $\emptyset, \Abox$.
\end{lemma}
\begin{proof}

First, we show that $\Abox \models \varkappa_{\mathfrak T'}$, for every finite subtree $\mathfrak T'$ of $\mathfrak T_S$. Let $\I_\D$ be an $\LTL{}$ interpretation such that $\I_\Abox, n \models A$ iff $A(n) \in \D$, for any atom $A$. We regard any $\LTL{}$ interpretation $\I$ as a black/red transition system with the states $\{ 0, u  \} \cup \{ \avec{d}\avec{e} \mid \avec{d},\avec{e} \subseteq \mathbb N, \avec{e} \neq \emptyset\}$. The state $0$ is labelled with $\{ A \in \Sigma \mid \Abox, 0 \models A \}$, $u$ is labelled with $\Sigma^\bot$, while each state $\avec{d}\avec{e}$ is labelled with $\{ A \in \Sigma \mid \Abox, \avec{e} \models A \}$. From $0$, there are black transitions according to $(i)$. From each $\avec{e}\avec{d}$, we have black and red transitions according to $(ii)$ and $(iii)$, respectively. The state $u$ has a black and a red transition to itself and a transition from each state $\emptyset \avec{e}$ all labelled with $\Sigma^\bot$. It should be clear that $\Abox \models \varkappa_{\mathfrak T'}$ iff $\mathfrak T'$ is homomorphically embeddable into $\mathfrak T_{\I_\Abox}$. We define an embedding $h$ of $\mathfrak T'$ into $\I_\Abox$ as follows. We set $h(0) = 0$. Suppose $h(\mathfrak s)$ for $\mathfrak s = 0 \to s_1 \to \dots \to s_n$ has been defined and let $\mathfrak s' = 0 \to s_1 \to \dots \to s_{n+1}$. Suppose, first, $s_{n+1} = \avec{d}\avec{e}$ for $\avec{e} \subseteq [0, \max \Abox]$. Then we set $h(\mathfrak s') = \avec{d}\avec{e}$. Suppose $s_{n+1} = z$. Then $h(\mathfrak s_n) = \avec{d}\avec{e}$, for some $\avec{d}, \avec{e} \subseteq \mathbb N$. If $s_n \to s_{n+1}$ is black, we set $h(\mathfrak s') = \emptyset\{e+1 \mid e \in \avec{e}\}$ and if it is red, we set $h(\mathfrak s') = \emptyset\{d+1 \mid d \in \avec{d}\}$. Finally, if $s_{n+1} = u$, then we set $h(\mathfrak s') = u$. It is straightforwardly verified that $h$ is a homomorphism. Therefore, $\mathfrak T'$ is homomorphically embeddable into $\mathfrak T_{\I_\Abox}$ and $\Abox \models \varkappa_{\mathfrak T'}$.

Second, we show that $\mathfrak T_\varkappa$ is homomorphically embeddable into $\mathfrak T_S$ for each $\varkappa$ such that $\Abox \models \varkappa$. Take any $\varkappa$ such that $\Abox \models \varkappa$. It follows that $\mathfrak T_\varkappa$ is homomorphically embeddable into $\mathfrak T_{\I_\Abox}$. It remains to observe that $\mathfrak T_{\I_\Abox}$ is homomorphically embeddable into $S$. Indeed, we define $h(0) = 0$. Let $\mathfrak s = 0 \to s_1 \to \dots \to s_n$ for $n \geq 1$. If $s_n = \avec{d}\avec{e}$ for $\avec{e} \subseteq [0, \max \Abox]$, then $h(\mathfrak s) = s_n$. If $\max \avec{e} > \max \Abox$, we set $h(\mathfrak s) = z$. Finally, if $s_n = u$, we set $h(\mathfrak s) = u$. It is readily verified that $h$ is a homomorphism from $\mathfrak T_{\I_\Abox}$ into $S$.
\end{proof}

Now we explain why $\QBE(\Q[\U])$ is in \PSpace{}. To this end we observe that every run of $S$ of length $> \max \D$ results in either $s = z$ or $s = u$. Moreover, if $s = u$ then all the subsequent states of the run are also $u$. Thus, any run of $\mathfrak P = \prod_{i \in I^+} S^i$ of  length $> \max_{i \in I^+}\{\max \D\}$ is in a state $s = (t_1, \dots, t_{|I^+|})$ where $t_i \in \{u, z\}$. Then $\mathfrak T_{\mathfrak P}^M$ for $M = \max_{i \in I^+ \cup I^-}\{\max \D\} + 1$ is mapped into $\mathfrak U = \uplus_{i \in I^-} S^i$, if a map $h$ exists, in such a way that $h(s_0 \to \dots \to s_M)$ is either $z$ or $u$ (in the corresponding $S$ representing $\Abox_i$, $i \in I^-$). So, we obtain that if $\mathfrak T_{\mathfrak P}^M$ is homomorphically embeddable into $\mathfrak U$, then any finite subtree of $\mathfrak T_{\mathfrak P}$ is homomorphically embeddable into $\mathfrak U$. To decide $\QBE(\Q[\U])$, we can check the embeddability of $\mathfrak T_{\mathfrak P}^M$ in a branch-by-branch fashion similarly to the case of $\Q[\Us]$. Note, however,  that the existence of a polynomial algorithm for $\QBEbp(\Q[\U])$ and $\QBEba(\Q[\U])$ remains open as bounding the number of positive examples does not result in $\mathfrak P$ of polynomial size.

\paragraph{Constructing representations for queries with an $\LTL_\horn\Xallop$-ontology.}

Let $\avec{d},\avec{e}$ be finite and nonempty subsets of the interval $[0, P)$, for some $P \in \mathbb N$, and $M \in \mathbb N$.
For any $d \in \avec{d}$, let $\text{succs}(d, \avec{e}) = \{e \in \avec{e} \mid d < e \}$ and let
$$
\mu(d) =
\begin{cases}
  \min \text{succs}(d, \avec{e}), & \mbox{if either } d \in [0, M) \text{ or both } d \in [M, P) \text{ and } \text{succs}(d, \avec{e}) \neq \emptyset, \\
  \min (\avec{e}), & \mbox{if }d \in [M, P) \text{ and } \text{succs}(d, \avec{e}) = \emptyset.
\end{cases}
$$
If $\mu$ is a surjective $\avec{d} \to \avec{e}$ function, we write $\avec{d} \sm_{M,P} \avec{e}$ and set
$
\unt_{M, P}(\avec{d}, \avec{e}) ~=~ \bigcup_{d\in\avec{d}} \{d' \in \text{bwn}(d,\mu(d))\},
$
where $\text{bwn}(d,e) = (d, e)$ if $d < e$ and $(e, P) \cup [M, d)$ if $d \geq e$.

\begin{example}\em
Let $M = 2$, $P = 8$,  $\avec{d} = \{1,4,6,7\}$ and $\avec{e} = \{3,5\}$. Then $\avec{d} \sm_{M,P} \avec{e}$ with $\unt_{M,P}(\avec{d}, \avec{e}) = \{2, 7\}$.
\end{example}

Let $\Abox$ be a data instance and $\TO$ an $\LTL_\horn\Xallop$-ontology. Let $\C_{\TO, \D}$ be the canonical model of $\TO, \Abox$ and $s_{\TO, \Abox}$, $p_{\TO, \D}$ be the numbers from Proposition~\ref{prop:period}. Let $\max \Abox+s_{\TO, \Abox} = M$ and $\max \Abox+s_{\TO, \Abox}+p_{\TO, \Abox} = P$.  We define $S$ with the states $\{0, u\} \cup \{\avec{d}\avec{e} \mid \avec{d},\avec{e} \subseteq [0,P) \}$. The label of $0$ is $\{ A \in \Sigma \mid \C_{\TO, \Abox}, 0 \models A \}$, the label of $\avec{d}\avec{e}$ is $\{ A \in \Sigma \mid \C_{\TO, \Abox}, \avec{e} \models A \}$ and the label of $u$ is $\Sigma^\bot$. We define the (red and black) transitions between $0$ and $\avec{d}\avec{e}$ as specified by $(i)$--$(iii)$ above but using $\sm_{M,P}$ instead of $\sm$, $\unt_{M,P}$ instead of $\unt$, and $\C_{\TO, \Abox}, \avec{e}$ instead of $\Abox, \avec{e}$ (the same applies to $\avec{f}$). Finally, we define the transitions between $\emptyset \avec{e}$ and $u$ as defined above.

\begin{lemma}
$S$ represents $\TO, \Abox$.
\end{lemma}
\begin{proof}
First, we show that $\Abox \models \varkappa_{\mathfrak T'}$ for every finite subtree $\mathfrak T'$ of $\mathfrak T_S$.
It is clear from the properties of $\C_{\TO, \Abox}$ that $\TO, \Abox \models \varkappa_{\mathfrak T'}$ iff $\mathfrak T'$ is homomorphically embeddable into $\mathfrak T_{\C_{\TO, \Abox}}$. We define an embedding $h$ of $\mathfrak T'$ into $\C_{\TO, \Abox}$ as follows. We set $h(0) = 0$. Suppose $h(\mathfrak s)$ for $\mathfrak s = 0 \to s_1 \to \dots \to s_n$ has been defined and $\mathfrak s' = 0 \to s_1 \to \dots \to s_{n+1}$. If $s_n = u$, we set $h(\mathfrak s') = u$. Suppose $s_n = \avec{d}\avec{e}$ and $s_{n+1} = \avec{f}\avec{g}$. We will have an IH that $r(\avec{d}') = \avec{d}$ and $r(\avec{e}') = \avec{e}$ for the map $r$ from Lemma~\ref{th:horn+qus-repr}.
First, we assume $s_n \to s_{n+1}$ is a black transition. Then $\avec{e} \sm_{M, P} \avec{g}$ and let $\mu_{M,P} \colon \avec{e} \to \avec{g}$ be the corresponding (surjective) map. We construct a map $\mu' \colon \avec{e}' \to \mathbb N$ by taking
$$
\mu'(e) =
\begin{cases}
  e +  \mu_{M,P}(r(e)) - r(e), & \mbox{if } \mu_{M,P}(r(e)) > r(e); \\
  e + p_{\TO, \Abox} - r(e) +\mu_{M,P}(r(e)), & \mbox{otherwise},
\end{cases}
$$
for each $e \in \avec{e}'$.
We set $h(\mathfrak s') = \avec{f}'\avec{g}'$, where $\avec{g}' = \mu'(\avec{e'})$ and $\avec{f}' = \unt(\avec{e}, \avec{g}')$. We note that $r(\avec{f}') = \avec{f}$ and $r(\avec{g}') = \avec{g}$, so IH continues to hold. The case when $s_n \to s_{n+1}$ is a red transition is similar and left to the reader. It can be readily verified that $h$ is a homomorphism from $\mathfrak T'$ to $\mathfrak T_{\C_{\TO, \Abox}}$.

Now, we show that $\mathfrak T_\varkappa$ is homomorphically embeddable into $\mathfrak T_S$ for each $\varkappa$ such that $\TO, \Abox \models \varkappa$. Take any $\varkappa$ such that $\TO, \Abox \models \varkappa$. It follows that $\mathfrak T_\varkappa$ is homomorphically embeddable into $\mathfrak T_{\C_{\TO, \Abox}}$. It remains to show that $\mathfrak T_{\C_{\TO, \Abox}}$ is homomorphically embeddable into $S$. We define $h$ so that $h(0) = 0$, $h(\mathfrak s) = u$ for $\mathfrak s = s_0 \to s_1 \to \dots \to u$. For $\mathfrak s = s_0 \to s_1 \to \dots \to \avec{d}\avec{e}$, we set $h(\mathfrak s) = r(\avec{d})r(\avec{e})$. It is readily verified that $h$ is homomorphism.
\end{proof}

To justify the $2\ExpTime{}$ upper bound for $\QBE(\LTL_\horn\Xallop, \Q[\U])$, we observe that $S$ above representing $\TO, \D$ can be constructed in time $O(2^{2^{|\TO|+|\Abox|}})$ as $S$ has such number of states.
To justify the $\PSpace{}$ upper bound, let $N_i$, for $i \in I^+ \cup I^-$, be the number of states in $S$ representing $\TO, \Abox_i$ of the form either $u, z$ or $\avec{d}\avec{e}$, for $(\avec{d} \cup \avec{e}) \cap [0, \max \Abox_i + s_{\TO, \Abox_i}) = \emptyset$. We set $N^+ = \prod_{i \in I^+} N_i$.
Similarly to the argument after Lemma~\ref{th:horn+qus-repr},
we observe that if there exists a finite subtree $\mathfrak T$ of $\mathfrak T_{\mathfrak P}$ that is not homomorphically embeddable into $\mathfrak U$, then there exists such $\mathfrak T$ satisfying the property that every $s$ from $\mathfrak P$ occurs on each path of $\mathfrak T$ at most $K^-=\max_{i \in I^-} \{ \max \Abox_i + s_{\TO, \Abox_i} + N_i\}$-many times. It follows that the required $\mathfrak T$, if exists, is a subtree of $\mathfrak T_{\mathfrak P}^M$ for $M = \max_{i \in I^+} \{ \max \D_i + s_{\TO, \Abox_i} \} + N^+ K^-$.
Indeed, in the required $\mathfrak T$, if there is a path that is longer than $M$, the property above would be violated. By our construction of $S^i$, any $n$-th element, for $n \geq \max_{i \in I^+} \{ \max \D_i + s_{\TO, \Abox_i} \}$, of any path of $\mathfrak T_{\mathfrak P}$ is of the form $(t_1, \dots, t_{|I^+|})$, where $t_i$ is either $u,z$ or $\avec{d}\avec{e}$ satisfying $(\avec{d} \cup \avec{e}) \cap [0, \max \Abox_i + s_{\TO, \Abox_i}) = \emptyset$. Observe that $(t_1, \dots, t_{|I^+|}) \to (s_1, \dots, s_{|I^+|})$ in $\mathfrak P$, for $(t_1, \dots, t_{|I^+|})$ as above, implies $s_i$ is either $u,z$ or $\avec{d}\avec{e}$ satisfying $(\avec{d} \cup \avec{e}) \cap [0, \max \Abox_i + s_{\TO, \Abox_i}) = \emptyset$. Any sequence $(t_1, \dots, t_{|I^+|}) \to \dots \to (s_1, \dots, s_{|I^+|})$ as above in $\mathfrak P$ longer than $N^+ K^-$ will have some $(t_1, \dots, t_{|I^+|})$ repeated more than $K^-$ times.


\section{Proofs for Section~\ref{sec:horn}}

Let $E^+=\{\mathcal D_1^+,\dots,\mathcal D^+_n\}$ and $E^-=\{\mathcal D_1^-,\dots,\mathcal D_l^-\}$ and let $\TO$ be a $\LTL_\horn\Xallop$ ontology.
For every $\mathcal{C}_{\TO,\mathcal{D}}$ with $\mathcal{D} \in E^+ \cup E^-$, let $s_{\TO,\mathcal{D}}\leq 2^{|\TO|}$ and $p_{\TO,\mathcal{D}}\leq 2^{2|\TO|}$ be the length of the `handle' and the length of the `period' in $\mathcal{C}_{\TO,\mathcal{D}}$, respectively (provided by Proposition~\ref{prop:period}). Set
$$
k = \max_{\mathcal{D} \in E^+ \cup E^-} (\max \mathcal{D} + s_{\TO,\mathcal{D}}), \qquad m = \prod_{\mathcal{D} \in E^+ \cup E^-} p_{\TO,\mathcal{D}}.
$$
\begin{lemma}\label{lem:hornsize}
	$(i)$ If $E$ is $\mathcal{Q}[\nxt,\Diamond]$-separable under $\TO$,
	then it is separated by a conjunction of at most $l$-many $\varkappa\in\mathcal{Q}_{p}^{\circ}[\Diamond]$
	of $\Diamond$-depth $\leq k+1$ and $\nxt$-depth $\leq k+m$.
	
	$(ii)$ If $E$ is $\mathcal{Q}[\Diamond]$-separable under $\TO$,
	then it is separated by a conjunction of at most $l$-many $\varkappa\in\mathcal{Q}_{p}[\Diamond]$
	of $\Diamond$-depth $\leq k+1$.
	
	$(iii)$ If $E$ is $\mathcal{Q}_{p}[\nxt,\Diamond]$-separable under $\TO$,
	then it is separated by some $\varkappa\in\mathcal{Q}_{p}[\nxt,\Diamond]$
	of $\Diamond$-depth $\leq k+l$ and $\nxt$-depth $\leq k+m$.
	
	$(iv)$ If $E$ is $\mathcal{Q}_{p}[\Diamond]$-separable under $\TO$,
	then it is separated by some $\varkappa\in\mathcal{Q}_{p}[\Diamond]$
	of $\Diamond$-depth $\leq k+l$.
\end{lemma}
\begin{proof}
	Recall that the types of any $\mathcal{C}_{\TO,\mathcal{D}}$ form a sequence
	$$
	\tp_0,\dots,\tp_k, \tp_{k+1},\dots, \tp_{k+m}, \dots, \tp_{k+1},\dots, \tp_{k+m}, \dots .
	$$
$(i)$ Recall from the proof of Theorem~\ref{th:data-reduction} $(i.3)$ that we may assume that for any $\mathcal{D}^{-}\in E^{-}$ there is a query $\varkappa\in\mathcal{Q}_{p}^{\circ}[\Diamond]$ that separates $(E^{+},\{\mathcal{D}^{-}\})$. So let $\mathcal{D}^{-}\in E^{-}$ and assume that $(E^{+},\{\mathcal{D}^{-}\})$ is $\Qp^\circ[\Diamond]$-separable under $\TO$. Then there is a separator
	$$
	\varkappa = \rho_0 \land \Diamond (\rho_1 \land \Diamond (\rho_2 \land \dots \land \Diamond \rho_n))
	$$
	in which each $\rho_r$ has $\nxt$-depth $\le k+m$. Indeed, in view of the form of the canonical models, if $\rho_r = \bigwedge_{i=0}^\ell \nxt^i \lambda_i$ with $\ell > k+m$, then one can replace $\rho_r$ with
	$$
	\rho'_r = \bigwedge_{i=0}^k \nxt^i \lambda_i \land \bigwedge_{j=1}^m \nxt^{k+j} \hspace*{-5mm} \bigwedge_{\substack{0 \le i \le \ell\\ j=(i-k) \ \text{mod} \ m}} \hspace*{-5mm} \lambda_i.
	$$
	In addition, if $n > k$, then the query
	$$
	\rho_0 \land \Diamond (\rho_1 \land \Diamond (\rho_2 \land \dots \land \Diamond \rho_k)) \land \bigwedge_{i=k+1}^n \Diamond^{k+1} \rho_i
	$$
	still separates $(E^{+},\{\mathcal{D}^{-}\})$ under $\TO$, and so some
	$\rho_0 \land \Diamond (\rho_1 \land \Diamond (\rho_2 \land \dots \land \Diamond (\rho_k \wedge \Diamond \rho_{j}))$ with $k<j\leq n$ separates $(E^{+},\{\mathcal{D}^{-}\})$ under $\TO$.
	
	\medskip
	
	$(ii)$ is proved by dropping the $\nxt$-queries from the proof of $(i)$.
	
	\medskip
	
	$(iii)$ The proof of $(i)$ shows that if $E$ is $\mathcal{Q}_p[\nxt,\Diamond]$-separable under $\TO$, then there is a separator
	$$
	\rho_0 \land \Diamond (\rho_1 \land \Diamond (\rho_2 \land \dots \land \Diamond \rho_k)) \land \bigwedge_{i=k+1}^n \Diamond^{k+1} \rho_i
	$$
	in which each $\rho_r$ has $\nxt$-depth $\le k+m$. Now we can select, for each
	negative example $\mathcal{D}^{-}$, a $j$ such that $$\rho_0 \land \Diamond (\rho_1 \land \Diamond (\rho_2 \land \dots \land \Diamond (\rho_k \wedge \Diamond \rho_{j}))$$ separates $(E^{+},\{\mathcal{D}^{-}\})$ under $\TO$. Let $j_{1},\ldots,j_{l}$ be thus selected. Then
	$$
	\rho_0 \land \Diamond (\rho_1 \land \Diamond (\rho_2 \land \dots \land \Diamond (\rho_k \wedge \Diamond (\rho_{j_{1}} \land \dots \land \Diamond \rho_{j_{l}})))
	$$
	separates $E$ under $\TO$.
	
	\medskip
	
	$(iv)$ is proved by dropping the $\nxt$-queries from the proof of $(iii)$.
\end{proof}

\bigskip
\noindent
\textbf{Theorem~\ref{thm:Horn-diam}.}
{\em Let $\mathcal{Q}\in \{\mathcal{Q}[\nxt,\Diamond],\mathcal{Q}[\Diamond],\mathcal{Q}_{p}[\nxt, \Diamond],\mathcal{Q}_{p}[\Diamond]\}$. Then $\mathsf{QBE}(\LTL_{\smash{\horn}}\Xallop,\mathcal{Q})$ and $\QBEba\!(\LTL_{\smash{\horn}}\Xallop,\mathcal{Q})$ are both \PSpace-complete for combined complexity.
}
\begin{proof}
	The lower bound follows from~\cite{Chen199495}. We first give the upper bound proof for $\mathcal{Q}_{p}[\Diamond]$.

	
	Let $E^+=\{\mathcal D_1^+,\dots,\mathcal D^+_n\}$ and $E^-=\{\mathcal D_1^-,\dots,\mathcal D_l^-\}$. We use Lemma~\ref{lem:hornsize} $(iv)$.
	Let $k$ and $m$ be as in Lemma~\ref{lem:hornsize}.
	The nondeterministic algorithm starts by guessing a conjunction of atoms $\rho_0$ and checking in \PSpace{} that $\TO,\mathcal D^+_i\models\rho_0(0)$ for all $i \in [1,n]$. We use numbers $d^+_i,d^-_j \le k+m$, for $i \in [1,n]$, $j \in [1,l]$, and a set $N\subseteq[1,l]$ that will keep track of the negative examples yet to be separated. Initially, we set all $d^+_i,d_j^- = 0$ and $N=\{j\in[1,l]\mid\TO,\mathcal D^-_j\models\rho_0(0)\}$. Then we repeat the following steps until $N=\emptyset$, in which case the algorithm terminates accepting the input:
	\begin{itemize}
		\item Guess a conjunction $\rho$ of atoms in the signature of $\TO$ and $E$.
		
		\item For every $i\in [1,n]$, check in \PSpace{} that $\TO,\mathcal D^+_i \models \Diamond\rho(d^+_i)$ and reject if this is not so.
		
		\item Guess ${d_i^+}'$ such that $\min(d^+_i,k) < {d^+_i}' \le k+m$ and $\TO,\mathcal D^+_i \models \rho({d^+_i}')$.
		
		\item For each, $j\in N$ check that $\TO,\mathcal D^-_j \models \Diamond\rho(d_i^-)$. If no, remove $j$ from $N$. Otherwise, find in \PSpace{} the smallest ${d^-_i}'$ such that $\min(d^-_i,k)< {d^-_i}' \le k+m$ and $\TO,\mathcal D^-_i \models \rho({d^-_i}')$.
		
		\item Set $d^+_i := {d^+_i}'$ and, for all $j$ still in $N$, set $d^-_j := {d^-_i}'$.
	\end{itemize}
	Let $\varphi_i=\rho_0\land\Diamond(\rho_1\land\Diamond(\dots(\rho_{i-1}\land\Diamond\rho_i)\dots)$,  where $\rho_i$ is the conjunction of atoms guessed in the $i$-th iteration. Let $N_i$ be the set $N$  after the $i$-th iteration. Then, for all $j\in[1,n]$, we have $\TO,\mathcal D^+_j\models\varphi_i(0)$, for all $j\in N_i$ we have $\TO,\mathcal D^-_j\models\varphi_i(0)$, and for all $j\in [1,l]\setminus N_i$ we have $\TO,\mathcal D^-_j\not\models\varphi_i(0)$. So the algorithm accepts after the $\ell$-th iteration iff $\varphi_\ell$ separates $(E^+,E^-)$.

By Lemma~\ref{lem:hornsize} $(ii)$, this also gives a \PSpace{} algorithm for $\mathcal{Q}[\Diamond]$. By Lemma~\ref{lem:hornsize} $(i)$, for $\mathcal{Q}[\nxt,\Diamond]$ it suffices to give a \PSpace{} algorithm for $\mathcal{Q}_{p}^{\circ}[\Diamond]$, which can be obtained by modifying the algorithm above. 	

It starts by guessing a conjunction of atoms $\lambda_0$  and checking that $\TO, \mathcal{D}_i^+ \models \lambda_0(0)$ for all $i \in [1, n]$, which can be done in \PSpace{}. We use numbers $d^+_i ,d^-_i \le k+m$  (for $\Diamond$-subformulas) and $c \le k+m$ (for $\nxt$-formulas), and a set $N \subseteq [1, l]$ that will keep track of the negative examples yet to be separated. Initially, we set all $d^+_i ,d^-_j = 0$, $c = 0$, and $N = [1,l]$. Then we repeat (1) or (2) until $N = \emptyset$, in which case the algorithm terminates accepting the input:
	\begin{enumerate}
		\item[(1)]
		\begin{itemize}
			\item[--] Set $c=0$.
			
			\item[--] Guess a conjunction $\lambda$ of atoms in the signature of $\TO$ and $E$.
			
			\item[--] For every $i \in [1, n]$, check in \PSpace{} that $\TO, \mathcal{D}_i^+ \models \Diamond \lambda(d+i)$ and reject if this is not so.
			
			\item[--] Guess $d^{+'}_i$ such that $\min (d+,k) < d^{+'} \le k+m$ and $\TO, \mathcal{D}^+_i \models \lambda(d^{+'})$.
			
			\item[--] For each $j \in N$, check that $\TO, \mathcal{D}_j^- \models \Diamond \lambda( d^-_i)$. If no, remove $j$ from $N$. Otherwise, find in \PSpace{} the smallest $d^{-'}$ such that $\min (d^-_i,k) < d^{-'} \le k+m$ and $\TO, \mathcal{D}_i^- \models  \lambda( d^{-'})$.
			
			\item[--] Set $d^+_i = d^{+'}_i$ and, for all $j$ still in $N$, set $d^-_j = d_i^{-'}$.
		\end{itemize}
		
		\item[(2)]
		\begin{itemize}
			\item[--] Increment $c$, provided $c < m + k$.
			
			\item[--] Guess a conjunction $\lambda$ of atoms in the signature of $\TO$ and $E$.
			
			\item[--] For every $i \in [1, n]$, check in \PSpace{} that $\TO, \mathcal{D}_i^+ \models \lambda(d^+_i + c)$ and reject if this is not so.
			
			\item[--] For each $j \in N$, check that $\TO, \mathcal{D}_j^- \models \lambda(d^-_i +c)$. If no, remove $j$ from $N$.
		\end{itemize}
	\end{enumerate}

A \PSpace{} algorithm for $\mathcal{Q}_p[\nxt,\Diamond]$ is similar to the one above: it uses  Lemma~\ref{lem:hornsize} $(iii)$ for guessing the next temporal operator $\nxt$ or $\Diamond$ in the query.
\end{proof}

\bigskip
\noindent
\textbf{Theorem~\ref{thm:Horn-U}.}
{\em $\mathsf{QBE}(\LTL_\horn\Xallop, \mathcal{Q}[\Us])$ is in \ExpTime{} for combined complexity,
$\mathsf{QBE}(\LTL_\horn\Xallop, \mathcal{Q}_{p}[\U])$ is in \ExpSpace{}, and
$\smash\QBEba(\LTL_\horn\Xallop, \mathcal{Q}_{p}[\U])$ is \NExpTime{}-hard.
}
\begin{proof}
The upper bound follows from Lemmas~\ref{lemmain} and \ref{th:horn+qus-repr} above as explained in the main part of the paper.

Now we establish the \NExpTime{} lower bound. Let $\M$ be a non-deterministic Turing machine that accepts words $\boldsymbol{x}$ over its tape alphabet in at most $N=2^{p(|\boldsymbol{x}|)}$ steps, for some polynomial $p$. Given such an $\M$ and an input $\boldsymbol{x}$, our aim is to define an $\LTL_\horn\Xallop$ ontology $\TO$ and an example set $E=(E^+=\{\Abox_1^+,\Abox_2^+\},E^-=\{\Abox^-\})$ of size polynomial in $\M$ and $\boldsymbol{x}$ such that $E$ is separated by a $\Qp[\U]$-query under $\TO$ iff $\M$ accepts $\boldsymbol{x}$.

Suppose $\M$ has a set $Q$ of states, tape alphabet $\Sigma$ with $\B$ for blank, initial state $q_0$, and accepting state $\qa$.
Without loss of generality we assume that
$\M$ erases the tape before accepting and its head is at the left-most cell in any accepting configuration.

%
%
Given an input word $\boldsymbol{x}=x_1\dots x_n$ over $\Sigma$, we represent configurations $\conf$ of a computation of $\M$ on $\boldsymbol{x}$ by
the $(N-1)$-long word written on the tape (with sufficiently many blanks at the end), in which the symbol $y$ in the active cell is replaced by the pair $(q,y)$ with the current state $q$.
An accepting computation of $\M$ on $\boldsymbol{x}$ is encoded by the word
$w=\sharp \conf_1 \, \sharp \, \conf_2 \, \sharp\,  \dots \, \sharp \, \conf_{N-1} \, \sharp \, \conf_{N} $ over
the alphabet $\Xi=\Sigma\cup(Q\times\Sigma)\cup\{\sharp\}$,
where $\conf_1,\conf_2,\dots,\conf_N$ are the subsequent configurations in the computation. In particular, $\conf_1$ is the initial configuration $(q_0,x_1)x_2\dots x_n\B\dots\B$, and $\conf_N$ is the accepting configuration $\conf_{acc}=(\qa,\B)\B\dots\B$. Thus, any accepting computation is encoded by a word of length $N^2$ in the alphabet $\Xi$ (we allow $\conf_{acc}$ to follow $\conf_{acc}$).

A tuple $\lt = (a,b,c,d,e,f)\in (\Xi)^6$ is called  \emph{legal}~\cite[Theorem 7.37]{DBLP:books/daglib/0086373} if there exist two consecutive configurations $\conf_1$ and $\conf_2$ of $\M$ and a number $i$ such that
$$
abcdef = \conf_1[i]\conf_1[i+1]\conf_1[i+2]\conf_2[i]\conf_2[i+1]\conf_2[i+2],
$$
where $\conf_j[i]$ is the $i$th symbol in $\conf_j$.
Let $\leg \subseteq (\Xi)^6$ be the set of all legal tuples (plus a few additional 6-tuples to take care of $\sharp$) with the following property: a word $w$ encodes an accepting computation iff it starts with the initial configuration preceded by $\sharp$, ends with the accepting configuration, and
every two length 3 subwords at distance $N$ apart form a legal tuple. Let $\bar \leg = (\Xi)^6\setminus \leg$.


For any $k>0$, by a \emph{$k$-counter} we mean a set $\mathbb A=\{A^i_j \mid i=0,1, \ j=1,\dots,k\}$  of atomic concepts that will be used to store values between $0$ and $2^k-1$, which can be different at different time points. The counter $\mathbb A$ is \emph{well-defined} at a time point $n \in \mathbb N$ in an interpretation $\I$ if $\I,n \models A^0_j \land A^1_j \to \bot$ and $\I,n \models A^0_j \lor A^1_j$, for any $j=1,\dots,k$. In this case, the \emph{value of} $\mathbb A$ at $n$ in $\I$ is given by the unique binary number $b_{k} \dots b_1$ for which $\I,n \models A^{b_1}_1\wedge\dots\wedge A^{b_k}_k$.
We require the following formulas, for $c = b_{k} \dots b_1$ (provided that $\mathbb A$ is well-defined):
\begin{itemize}
\item $[\mathbb A={c}] = A^{b_1}_1\wedge\dots\wedge A^{b_k}_k$, for which $\I,n \models [\mathbb A={c}]$ iff the value of $\mathbb A$ is $c$;

\item $[\mathbb A{<c}] = \bigvee_{\substack{k\geq i\geq 1\\ b_i=1}}\big(A_i^0\wedge\bigwedge_{j=i+1}^kA_j^{b_j}\big)$ with $\I,n \models [\mathbb A < {c}]$ iff the value of $\mathbb A$ is $<c$;

\item $[\mathbb A{>c}] = \bigvee_{\substack{k\geq i\geq 1\\ b_i=0}}\big(A_i^1\wedge\bigwedge_{j=i+1}^kA_j^{b_j}\big)$ with $\I,n \models [\mathbb A > {c}]$ iff the value of $\mathbb A$ is $>c$.
\end{itemize}
%
%
%
We regard the set $(\Rnext \mathbb A)=\{\Rnext A^i_j \mid i=0,1, \ j=1,\dots,k\}$ as another counter that stores at $n$ in $\I$ the value stored by $\mathbb A$ at $n+1$ in $\I$.
Thus, we can use formulas like $[\mathbb A >c_1]\to[(\Rnext \mathbb A)={c_2}]$, which says that if the value of $\mathbb A$ at $n$ in $\I$ is greater than $c_1$, then the  value of $\mathbb A$ at $n+1$ in $\I$ is $c_2$. Also, for $l \le k$, we can use formulas like $[\mathbb A=i \,  (\text{mod}\, 2^l)]$ with self-explaining meaning. Another important formula we need is defined by:
%
\begin{multline*}
 [\mathbb A={\mathbb B+1}] = \bigwedge_{i=1}^k\big(( B_i^0\wedge B_{i-1}^1\wedge\dots\wedge B_{1}^1\to A_i^1\wedge A_{i-1}^0\wedge\dots\wedge A_{1}^0)
\wedge \bigwedge_{j<i}((B_i^0\wedge B_{j}^0\to A_i^0)\wedge(B_i^1\wedge B_{j}^0\to A_i^1))\big).
\end{multline*}
It says that the value of $\mathbb A$ is one greater than the value of $\mathbb B$.

To define $\TO$ and $E = (E^+, E^-)$ for given $\M$ and $\boldsymbol{x}=x_1\dots x_n$, we assume that $\Xi=\{a_1,\dots,a_{2^m}\}$
 and $k=6m+2\lceil \log N\rceil+1$.
We use the following atomic concepts in $\TO$ and $E$: the symbols in $\Xi$, the atoms $C$,
$S$, $N$, $T$, and those atoms that are needed in $k$-counters $\mathbb S$, $\mathbb N$, $\mathbb T$.

We set $\Abox_1^+=\{T(0)\}$, $\Abox_2^+=\{S(0)\}$, and $\Abox^-=\{N(0)\}$.

The following axioms initialise the corresponding $m$-counters:
$$
T\to [(\nxt \mathbb T)=0]\quad S\to [(\nxt \mathbb S)=0]\quad N\to [(\nxt \mathbb N)=0].
$$
These and all other axioms of $\TO$ can be easily transformed to equivalent sets of polynomially-many $\LTL_\horn\Xallop$ axioms.

The behaviour of each counter is specified by the axioms below whose meaning is illustrated by the structure of the canonical model of the corresponding example restricted to $\Xi\cup\{C\}$.

The $T$-axioms
\begin{align*}
&[\mathbb T<N^2]\to [\nxt \mathbb T=\mathbb T+1],\\
& [\mathbb T=0]\to \sharp, \quad [\mathbb T=1]\to (q_1,x_1),\\
&[\mathbb T=2]\to x_2, \dots , [\mathbb T=n]\to x_n,\\
&[\mathbb T>n]\land[\mathbb T<N]\to \B, \\
&[\mathbb T>N]\land[\mathbb T<N^2-N]\to \Xi,\\
&[\mathbb T=N^2-N]\to \sharp, \\
& [\mathbb T=N^2-N+1]\to (\qa,\B),\\
& [\mathbb T>N^2-N+1]\to \B
\end{align*}
together with the data instance $\Abox_1^+$ give rise to the canonical model of the form
%
$$\emptyset, \sharp,(q_1,x_1),x_2,\dots,x_n,\B^{N-n-1},\sharp,\Xi^{N^2-2N-1},\sharp,(\qa,\B),\B^{N-2},\dots$$
%
%
%
%
The $S$-axioms
\begin{align*}
&[\mathbb S<(2|\Xi|+1)N^2]\to[\nxt \mathbb S=\mathbb S+1],\\
&[\mathbb S>N^2]\to C,\\
&[\mathbb S> N^2\land\mathbb S=2i \ (\text{mod}\ 2^{m+1})]\to a_{i+1}, \text{ for all $i\in[1,2^m]$}
\end{align*}
and $\Abox_2^+$ generate the canonical model
%
%
$$
\emptyset, \emptyset^{N^2},\emptyset,(a_1C,C,\dots,a_{2^k}C,C)^{N^2},\emptyset,\emptyset,\dots
$$
Let $\bar\leg=\{\lt_1=(a_1,b_1,c_1,d_1,e_1,f_1),\ldots,\lt_l\}$. Let $t=N^2-N-3$. The $N$-axioms comprise the following, for each $i\in[1,l]$:
\begin{align*}
&[\mathbb N < (2l+3)N^2]\to [\nxt \mathbb N = \mathbb N +1],\\
&[0<\mathbb N<N^2]\to \Xi,\\
&[N^2<\mathbb N<3N^2]\to C,\\
&[N^2<\mathbb N<3N^2]\land N_0^1\to \Xi,\\
&[(2i+1)N^2<\mathbb N < (2i+1)N^2+t+1]\to \Xi, \\
&[\mathbb N =(2i+1)N^2+t+1]\to a_i,\\
&[\mathbb N =(2i+1)N^2+t+2]\to b_i,\\
&[\mathbb N =(2i+1)N^2+t+3]\to c_i,\\
&[(2i+2)N^2-N<\mathbb N <(2i+2)N^2-2]\to \Xi,\\
&[\mathbb N =(2i+2)N^2-2]\to d_i,\\
&[\mathbb N =(2i+2)N^2-1]\to e_i,\\
&[\mathbb N =(2i+2)N^2]\to f_i,\\
&[(2i+2)N^2<\mathbb N <(2i+2)N^2+t+1]\to \Xi.
\end{align*}
The data instance $\Abox^-$ gives the canonical model
$$\emptyset,\emptyset,\Xi^{N^2-1}, \emptyset,(\Xi C,C)^{N^2-2},\Xi C,\emptyset,\Abox_{\lt_1},\emptyset^{N+2},\Abox_{\lt_2}\ldots,\Abox_{\lt_l},\emptyset,\dots.$$
where $\Abox_{\lt_i}=\emptyset,\Xi^{t},a_i,b_i,c_i,\Xi^{N-3},d_i,e_i,f_i,\Xi^{t}$.

We denote the set of the axioms above by $\TO$ and show that $E$ is separated by a $\Qp[\U]$-query $\varkappa$ under $\TO$ iff $\M$ accepts $\boldsymbol{x}$.

\smallskip

($\Leftarrow$) Suppose $\rho_1\dots\rho_{N^2}$ encodes an accepting computation of $\M$ on $\boldsymbol{x}$. Consider the $\Qp[\U]$-query
$$
\varkappa =\Diamond(\rho_1\land C\U(\rho_2 \land C \U (\dots(\rho_{N^2-1}\land (C\U \rho_{N^2}))\dots))).
$$
It is not hard to show by inspecting the respective canonical models described above that
$$
 \TO,\Abox_1^+ \models\varkappa(0), \quad \TO,\Abox_2^+ \models\varkappa(0), \quad \TO,\Abox_1^- \not\models\varkappa(0).
$$
To prove the last one, we first notice that $\emptyset,\Xi^{N^2-1}\not\models\varkappa(0)$, and $\emptyset,(\Xi C,C)^{N^2-2},\Xi C\not\models\varkappa(0)$.
We have $\Abox_\lt \not\models C(j)$ for all $\lt$ and $j$. So, if $\Abox_\lt\models\varkappa(0)$ for some $\lt$, then there is $i< t$ such that $\emptyset,\Abox_\lt \models \rho_{j}(i+j)$ for all $j\in [1,N^2]$. But then $\rho_{t-i}\rho_{t-i+1}\rho_{t-i+2}\rho_{t+N-i}\rho_{t+N-i+1}\rho_{t+N-i+2}=\lt\in\bar\leg$, which is a contradiction. So we have $\Abox_\lt\not\models\varkappa(0)$ for all $\lt\in\bar\leg$, and therefore $\TO,\Abox^- \not\models\varkappa(0)$.

\smallskip

($\Rightarrow$) Suppose the query
$$
\varkappa = \lambda_1\U(\rho_1\land \lambda_2\U(\rho_2  \dots(\rho_{K-1}\land (\lambda_{K}\U \rho_{K}))\dots))
$$
with $\rho_K\ne\top$ separates $E$ under $\TO$. Since $\TO,\Abox_1^+\models\varkappa(0)$, we have $K\le N^2$ and $\rho_i\subseteq\Xi$ for all $i$.
Since $\emptyset,\emptyset,\Xi^{N^2-1}\not\models\varkappa(0)$ we have $\rho_1\ne\emptyset$. Since $\TO,\Abox_2^+\models\varkappa(0)$ we have $\lambda_1=\top$. Now if $K<N^2$, then $\emptyset,\emptyset,\Xi^{N^2-1}\models\varkappa(0)$, so $K=N^2$.

Since $\TO,\Abox_2^+ \models\varkappa(0)$, we have $|\rho_i|\le 1$ for all $i$.
Let $y_1<\ldots <y_{N^2}$ be such that $\TO,\Abox_2^+\models \rho_j(y_j)$ and $\TO,\Abox_2^+\models \lambda(i)$ for all $j\in[1,N^2]$ and $i\in(y_j,y_{j+1})$. We see that if $y_j$ is odd, then $\rho_j=\emptyset$ and if $y_j$ is even we can assume that $\rho_j=a\in\Xi$ where $\TO,\Abox_2^+\models a(y_j)$. Let construct $z_j$ in the following way: $z_1=N^2+2$ and if we already have $z_j$, then $z_{j+1}$ is the smallest number bigger than $z_j$ with the same parity as $y_{j+1}$. We can see that, for all $j<N^2$, we have $\TO,\Abox^-\models\rho_j(z_j)$ with $\TO,\Abox^-\models\lambda_j(y)$ for all $y\in(z_j,z_{j+1})$ and if there is an odd $y_j$, then $z_{N^2}<3N^2-1$, $\TO,\Abox^-\models \rho_{N^2}(z_{N^2})$, and therefore $\TO,\Abox^-\models\varkappa(0)$ which cannot happen. So there are no odd $y_j$'s and $|\rho_i|=1$ for all $i$. 

%
%

In view of $\TO,\Abox_1^+ \models\varkappa(0)$, the word $\rho_1\ldots\rho_{N^2}$ starts with the starting configuration preceded by $\sharp$ and ends with the accepting one. Suppose there is some $i$ such that $(\rho_i,\rho_{i+1},\rho_{i+2},\rho_{N+i},\rho_{N+i+1},\rho_{N+i+2})=\lt\in \bar\leg$. Let $y_j=t-i+j$ for $j\in[1,N^2]$. We have $\Abox_\lt \models \rho_j(y_j)$, and so $\Abox_\lt \models \varkappa(0)$, and therefore $\TO,\Abox^- \models \varkappa(0)$. So every two length 3 subwords at distance $N$ apart form a legal tuple and $\rho_1\ldots\rho_{N^2}$ encodes a successful computation of $\M$ on $\boldsymbol{x}$.
\end{proof}

\bigskip
\noindent
{\bf Theorem~\ref{thm:datahorn}.}
{\em For data complexity, the results of Theorem~\ref{thm:qbewithout} continue to hold for queries mediated by an $\LTL_{\smash{\horn}}\Xallop$-ontology.
}

\bigskip
\noindent
We first consider $\nxt\Diamond$-queries, and then come to $\U$-queries.

\medskip
\noindent
$\nxt\Diamond$-queries. The \NP{}-lower bounds are inherited from the ontology-free case. For the \NP{}-upper bounds observe that by Lemma~\ref{lem:hornsize} and since $\TO$ is fixed, we always have a
separating query of polynomial size whenever a separating query exists. The
\NP{}-upper bounds follow in the standard way. We now come to the  \PTime{}-upper bounds. We prove the \PTime{}-upper bound for $\QBEba(\LTL_\horn\Xallop,\mathcal{Q}_{p}[\nxt,\Diamond])$ by modifying the dynamic programming algorithm we gave in the ontology-free case. The \PTime{}-upper bound for $\QBEba(\LTL_\horn\Xallop,\mathcal{Q}_{p}[\Diamond])$ is obtained by dropping $\nxt$ from the proof. The \PTime{}-upper bound for QBE$(\LTL_\horn\Xallop,\mathcal{Q}[\Diamond])$ with a bounded number of positive examples can be again proved in two steps:
(1) by Theorem~\ref{th:data-reduction} $(i.1)$ it suffices to prove the \PTime{}-upper bound for $\QBEba(\LTL_\horn\Xallop,\mathcal{Q}[\Diamond])$;
(2) by Theorem~\ref{th:data-reduction} $(i.3)$, $\QBEba(\LTL_\horn\Xallop,\mathcal{Q}[\Diamond])\leq_{p} \QBEba(\LTL_\horn\Xallop,\mathcal{Q}_{p}[\Diamond])$.
Finally, the \PTime{}-upper bound for QBE$(\LTL_\horn\Xallop,\mathcal{Q}[\nxt,\Diamond])$ with a bounded number of positive examples follows from the \PTime{}-upper bound for QBE$(\LTL_\horn\Xallop,\mathcal{Q}[\Diamond])$ with a bounded number of positive examples
using the same `trick' as in the proof of Theorem~\ref{th:data-reduction} $(ii.2)$ in the ontology-free case: we modify the models $\mathcal{C}_{\TO,\mathcal{D}}$ by adding fresh atoms $A_{i}$ encoding  $\nxt^{i}A$
and interpreting them in $\mathcal{C}_{\TO,\mathcal{D}}$ in the same way as as $\nxt^{i}A$. By Lemma~\ref{lem:hornsize} it suffices to do this for $i\leq k+m$ (which is polynomial in $E$ as $|\TO|$ is fixed).

To prove the \PTime{}-upper bound for $\mathcal{Q}_{p}[\nxt,\Diamond]$, we extend the notion of a satisfying assignment for a query $\varkappa$ in a  data instance to a satisfying assignment in the canonical model $\mathcal{C}_{\TO,\mathcal{D}}$ in the obvious way: suppose $\varkappa$ takes the form~\eqref{dnpath} with $\rho_{n}\ne \top$.
Then $\mathcal{C}_{\TO,\mathcal{D}},0\models \varkappa$ iff there is a \emph{satisfying assignment $f$ for $\varkappa$ in $\mathcal{C}_{\TO,\mathcal{D}}$} in the sense that $f$ is a strictly monotone map $f \colon [0,n] \rightarrow \mathbb{N}$ with $f(0)=0$, $f(i+1)=f(i)+1$
if $\op_{i}=\nxt$, and $\rho_{i}\subseteq \tp(f(i)) = \{A \mid \mathcal{C}_{\TO,\mathcal{D}},f(i) \models A\}$, for all $i\leq n$.
We first observe the following lemma (using the notation and numbers $k,m$ introduced for Lemma~\ref{lem:hornsize}):
\begin{lemma}
$(E^{+},E^{-})$ is $\mathcal{Q}_{p}[\nxt,\Diamond]$-separable under $\TO$ iff there exists $\varkappa$ of the form \eqref{dnpath} with $\rho_{n}\not=\top$, of $\Diamond$-depth $\leq k+l$ and $\nxt$-depth $\leq k+m$ such that
\begin{enumerate}
	\item for any $\mathcal{D}\in E^{+}$ there is a satisfying assignment for $\varkappa$ into $\mathcal{C}_{\TO,\mathcal{D}}$ with $f(n)\leq N:=k+(k+l+1)(k+m)$;
	\item for any $\mathcal{D}\in E^{-}$ there is no satisfying assignment for $\varkappa$ into $\mathcal{C}_{\TO,\mathcal{D}}$ with $f(n)\leq N$.
	\end{enumerate}
\end{lemma}

\begin{proof}
 First assume that $(E^{+},E^{-})$ is $\mathcal{Q}_{p}[\nxt,\Diamond]$-separable under $\TO$. By Lemma~\ref{lem:hornsize} there is a query $\varkappa$ of the form \eqref{dnpath} with $\rho_{n}\not=\top$, of $\Diamond$-depth $\leq k+l$ and $\nxt$-depth $\leq k+m$ that separates $(E^{+},E^{-})$ under $\TO$.
 Take any satisfying assignment $f$ for $\varkappa$ in $\mathcal{C}_{\TO,\mathcal{D}}$. Clearly then we can assume that $f(i+1)-f(i) \leq k+m$ for any $i$ with $\op_{i}=\Diamond$. Point~1 follows directly.
 Point~2 follows from $\TO,\mathcal{D}^{-}\not\models \varkappa(0)$ for $\mathcal{D}^{-}\in E^{-}$.

 Conversely, assume that there exists $\varkappa$ of the form \eqref{dnpath} with $\rho_{n}\not=\top$, of $\Diamond$-depth $\leq k+l$ and $\nxt$-depth $\leq k+m$ such that Points~1 and 2 hold. We show that $\varkappa$ separates $(E^{+},E^{-})$ under $\TO$. But $\TO,\mathcal{D}^{+}\models \varkappa(0)$ for $\mathcal{D}^{+}\in E^{+}$ follows from Point~1 and $\TO,\mathcal{D}^{-}\not\models \varkappa(0)$ for $\mathcal{D}^{-}\in E^{-}$
 follows from Point~2 using the same argument as in the proof of Point~1
 in the converse direction.
\end{proof}
We explain the modifications of the dynamic programming algorithm for $\QBEba(\mathcal{Q}_{p}[\nxt,\Diamond])$ for   $E^{+}=\{\mathcal{D}_{1}^{+},\mathcal{D}_{2}^{+}\}$ and
$E^{-}=\{\mathcal{D}_{1}^{-},\mathcal{D}_{2}^{-}\}$.

We modify the parameters stored in the tuples in the set $S_{i,j}$ slightly. Instead of the length of the query a tuple describes, we store its $\Diamond$-depth $K$ and its $\nxt$-depth $M$. Thus, let $S_{i,j}$ be the set of tuples $(K,M,\ell_1,\ell_2,n_{1},n_{2})$ such that
\begin{enumerate}
	\item $K \leq k+l$;
	\item $M \leq k+m$;
	\item $\ell_{1} \leq i\leq N$,
	\item $\ell_{2} \leq j\leq N$,
\end{enumerate}
and there is
$
\varkappa=\rho_0 \land \op_1 (\rho_1 \land \dots \land \op_k \rho_k)
$
of $\Diamond$-depth $K$ and $\nxt$-depth $M$ for which
\begin{enumerate}
	\item there are satisfying assignments $f_{1},f_{2}$ in $\mathcal{C}_{\TO,\mathcal{D}_{1}^{+}}$ and $\mathcal{C}_{\TO,\mathcal{D}_{2}^{+}}$ with $f_{1}(1+K+M)=\ell_{1}$ and $f_{2}(1+K+M)=\ell_{2}$, respectively, and
	\item $n_{1}$ is minimal with a satisfying assignment $f$ for $\varkappa$ in $\mathcal{C}_{\TO,\mathcal{D}_{1}^{-}}$ such that $f(1+K+M)=n_{1}\leq N$, and $n_{1}=\infty$ if there is no such $f$; and $n_{2}$ is minimal with a satisfying assignment $f$ for $\varkappa$ in $\mathcal{C}_{\TO,\mathcal{D}_{2}^{-}}$ such that $f(k)=n_{2}\leq N$, and $n_{2}=\infty$ if there is no such $f$.
\end{enumerate}
It suffices to compute $S_{N,N}$ in polynomial time because there exists a query
in $\mathcal{Q}_{p}[\nxt,\Diamond]$ separating $(E^{+},E^{-})$ iff there are
$K\leq k+\ell$, $M\leq k+m$, $\ell_{1}\leq N$, and $\ell_{2}\leq N$ such that $(K,M,\ell_{1},\ell_{2},\infty,\infty)\in S_{N,N}$. $S_{i,j}$ with $i\leq N$ and $j\leq N$ can be computed in essentially the same way as in the ontology-free case incrementally starting with $S_{0,0}$.

\medskip
\noindent
The bounds for $\U$-queries were explained in Section~\ref{appC} after Lemma~\ref{th:horn+qus-repr}.


\section{Proofs for Section~\ref{sec:boxdiamond}}
\bigskip
\noindent
\textbf{Theorem~\ref{thm:dim}.}
{\em Let $\mathcal{Q} \in \{\mathcal{Q}_{p}[\Diamond],\mathcal{Q}[\Diamond]\}$. If $E$ is $\mathcal{Q}$-separable under an $\LTL\Xbd$-ontology $\TO$, then $E$ can be separated under $\TO$ by a $\mathcal{Q}$-query of polysize in $E$ and $\TO$. $\mathsf{QBE}(\LTL\Xbd,\mathcal{Q})$ and $\QBEba(\LTL\Xbd,\mathcal{Q})$ are $\Sigma_{2}^{p}$-complete for combined complexity. The presence of $\LTL\Xbd$-ontologies has no effect on the data complexity, which remains the same as in Theorem~\ref{thm:qbewithout}.
}
\begin{proof}
	We start by giving a few more details of the $\Sigma_{2}^{p}$-lower bound proof.
	Recall that we reduce the validity problem for
	fully quantified Boolean formulas of the form
	$$
	\exists \avec{p}\, \forall \avec{q}\,\psi,
	$$
	where $\psi$ is a propositional formula, and $\avec{p}=p_{1},\dots,p_{k}$ and
	$\avec{q}= q_{1},\dots,q_{m}$ are lists of propositional variables.
	We assume w.l.o.g. that $\psi$ is not a tautology. We also assume that
	$\lnot\psi\not\models x$ for $x\in\{p_i, \lnot p_i, q_j, \lnot q_j \mid
	1\leq i\leq k, 1\leq j \leq m\}$. Indeed, if $\lnot\psi\models x$ then
	$\psi\equiv \lnot x\lor \psi'$, for some $\psi'$,  and when $x\in\{p_i, \lnot
	p_i\}$ the QBF formula $\exists \avec{p}\, \forall \avec{q}\,\psi$ is vacuously
	valid whereas when $x\in\{q_j\lnot q_j\}$ the QBF formula $\exists \avec{p}\,
	\forall \avec{q}\,\psi$ is valid  iff $\exists \avec{p}\, \forall
	\avec{q'}\,\psi'$ is, where $\avec{q'}$ is obtained from $\avec{q}$ by
	removing $q_j$.
	%
	We regard propositional variables as atoms and also use fresh atoms
	$A_{1},\dots,A_{k}$, $\bar A_{1},\dots,\bar{A}_{k}$ and $B$.
	
	Let $E = (E^+,E^-)$ with $E^+ = \{\Abox_1, \Abox_2\}$, $E^- = \{\Abox_3\}$, where
	$$
	\mathcal{D}_{1}=\{B_1(0)\}, \quad \mathcal{D}_{2}=\{B_2(0)\}, \quad \mathcal{D}_{3}=\{q_{1}(0),q_{2}(0),\ldots,q_{m}(0)\},
	$$
	and let $\mathcal{O}$ contain (the normal forms of) the following axioms, for all $i=1,\ldots,k$:
	\begin{align}\label{axioms0}
		& B_1 \rightarrow \neg\psi,\quad
		B_2 \rightarrow \neg\psi,\\
		\label{axioms1}
		& p_{i} \rightarrow \Diamond \big(\bar{A}_{i} \wedge \bigwedge_{j\not=i} (A_{j} \wedge \bar{A}_{j})\big), \qquad
		\neg p_{i} \rightarrow \Diamond \big(A_{i} \wedge \bigwedge_{j\not=i} (A_{j} \wedge \bar{A}_{j})\big),
		%
	\end{align}
	We show that $\exists \avec{p}\, \forall \avec{q}\, \psi$ is valid iff
	$E$ is $\Qp[\Diamond]$-separable under $\mathcal{O}$.
	
	$(\Rightarrow)$ Suppose $\exists \avec{p}\, \forall \avec{q}\, \psi$ is valid.
	Take an assignment $\mathfrak a$ for the variables $\avec{p}$  such that
	under all assignments $\mathfrak b$ for the variables $\avec{q}$ formula $\psi$ is true.
	Let $C$ be the conjunction of all $A_{i}$ with $\mathfrak a(p_{i})=1$ and all
	$\bar{A}_{i}$ with $\mathfrak a(p_{i})=0$, and let $\varkappa = \Diamond C$. We
	show that $\varkappa$ separates $E$.
	Define an interpretation $\mathcal{J}$ by taking
	\begin{itemize}
		\item $\J,0\models p_{i}$ iff $\mathfrak a(p_{i})=1$, for $i = 1,\dots,k$ and $\J,0\models q_j$, for $j = 1,\dots,m$;
		
		\item if $\J,0 \models p_{i}$, then $\J,i \models \bar{A}_{i} \wedge \bigwedge_{j\not=i} (A_{j} \wedge \bar{A}_{j})$;
		
		\item if $\J,0 \not\models p_{i}$, then $\J,i \models A_{i} \wedge \bigwedge_{j\not=i} (A_{j} \wedge \bar{A}_{j})$.
	\end{itemize}
	By the definition, $\mathcal{J}$ is a model of $\mathcal{O}$ and
	$\mathcal{D}_{3}$ with $\Jmc,0 \not\models \varkappa$. On the other hand, let
	$\mathcal{I}$ be a model of $\mathcal{O}$ and some $\mathcal{D}_l$, $l=1,2$.
	By~\eqref{axioms0}, $\I,0\not\models  \psi$. Then the truth values of the
	$p_{i}$ in $\I$ at $0$ cannot reflect the truth values of the $p_{i}$ under
	$\mathfrak a$ (for otherwise $\psi$ would be true at $0$ in $\I$). Take some
	$i_{0}$ for which these truth values of $p_{i_0}$ differ, say $\mathfrak
	a(p_{i_{0}})=1$ but $\I,0\not\models p_{i_{0}}$. Then $\mathcal{I},0 \models
	\Diamond (A_{i_{0}} \wedge \bigwedge_{j\not=i_0} (A_{j} \wedge \bar{A}_{j}))$,
	and so $\mathcal{I},0\models \varkappa$.
	
	$(\Leftarrow)$ Suppose a $\Qp[\Diamond]$-query $\varkappa$ separates $E$ but
	$\exists \avec{p}\, \forall \avec{q}\, \psi$ is not valid.  From our conditions
	on $\psi$, it is easy to see by considering possible models of $\TO$ and
	$\Abox_l$, $l=1,2,3$, that $\varkappa$ does not contain occurrences of $B_1$,
	$B_2$, $p_i$, $q_j$, $1\leq i \leq k$, $1\leq j\leq m$.
	Let $\J$ be a model of $\TO$ and $\Abox_3$ such that $\J,0 \not\models
	\varkappa$.
	Let $\mathfrak a$ be the assignment for $\avec{p}$ given by $\J$ at $0$.  As
	$\exists \avec{p}\, \forall \avec{q}\, \psi$ is not valid, there is an
	assignment $\mathfrak b$ for $\avec{q}$ such that $\psi$ is false under
	$\mathfrak a$ and $\mathfrak b$.
	Consider an interpretation $\I$ such that $\I,0\models B_1$, the truth values
	of $\avec{p}$ and $\avec{q}$ at $0$ are given by $\mathfrak a$ and $\mathfrak
	b$, and all other atoms are interpreted as in $\J$. Then $\I$ is a model of
	$\TO$ and $\Abox_1$, and so $\I,0 \models \varkappa$.  But then $\J\models
	\varkappa$, as $\varkappa$ can only contain atoms $A_i$ and $\bar A_i$, which
	is a contradiction showing that $\exists \avec{p}\, \forall \avec{q}\, \psi$ is
	valid.	
	
\medskip
	
	We now prove the results for data complexity. The \NP{}-lower bounds are inherited from the ontology-free case. We show the \NP{}-upper bound for
	QBE$(\LTL\Xbd,\mathcal{Q}_{p}[\Diamond])$ and the \PTime{}-upper bound
	for $\QBEba(\LTL\Xbd,\mathcal{Q}_{p}[\Diamond])$. The \PTime{}-upper bound for QBE$(\LTL\Xbd,\mathcal{Q}[\Diamond])$ with a bounded number of positive examples can be again proved in two steps:
	(1) by Theorem~\ref{th:data-reduction} $(i.1)$ it suffices to prove the \PTime{}-upper bound for $\QBEba(\LTL\Xbd,\mathcal{Q}[\Diamond])$;
	(2) by Theorem~\ref{th:data-reduction} $(i.3)$, $\QBEba(\LTL\Xbd,\mathcal{Q}[\Diamond])\leq_{p} \QBEba(\LTL\Xbd,\mathcal{Q}_{p}[\Diamond])$.
	
	Assume an $\LTL\Xbd$-ontology $\TO$ is given. We show that one can construct in \emph{polynomial} time for any data instance $\mathcal{D}$ a set $\mathcal{M}_{\TO,\mathcal{D}}$ of models of $\mathcal{D}$ whose types form a sequence
	\begin{equation}\label{periodA2}
		\!\tp_{0},\dots,\tp_{k_{0}},\tp_{k_{0}+1},\dots,\tp_{k_{0}+l},\dots,\tp_{k_{0}+1},\dots,\tp_{k_{0}+l},\dots
	\end{equation}
	with $\max \mathcal{D} \leq k_{0} \leq \max \mathcal{D}+|\TO|$ and $l\leq |\TO|$ such that for any $\varkappa\in \mathcal{Q}_{p}[\Diamond]$, $\mathcal{D},\TO\models \varkappa(0)$ iff $\I,0\models \varkappa$ for all $\I \in \mathcal{M}_{\TO,\mathcal{D}}$. Note that, in particular, every set $\mathcal{M}_{\TO,\mathcal{D}}$ is of polynomial size in $\mathcal{D}$. Interestingly, the models in $\mathcal{M}_{\TO,\mathcal{D}}$ are not necessarily  models of $\TO$ (unless $\TO$ is a Horn ontology). Then, to
	show the \NP{}-upper bound for
	QBE$(\LTL\Xbd,\mathcal{Q}_{p}[\Diamond])$ and the \PTime{}-upper bound
	for $\QBEba(\LTL\Xbd,\mathcal{Q}_{p}[\Diamond])$
	one constructs for any $\mathcal{D}\in E^{+}\cup E^{-}$ the set $\mathcal{M}_{\TO,\mathcal{D}}$ and then decides, using that polysize separating queries exist if separating queries exist at all,
    whether there exists $\varkappa\in \mathcal{Q}_{p}[\Diamond]$ such that
	\begin{itemize}
		\item for all $\mathcal{D}\in E^{+}$: $\I,0 \models \varkappa$, for all $\I \in \mathcal{M}_{\TO,\mathcal{D}}$;
		\item for all $\mathcal{D}\in E^{-}$: $\I,0 \not\models \varkappa$, for some $\I \in \mathcal{M}_{\TO,\mathcal{D}}$
	\end{itemize}
    in either \NP{} (by guessing the polysize query and then verifying it in polynomial time) or \PTime{} (by applying essentially the same dynamic programming algorithm as for $\QBEba(\mathcal{Q}_{p}[\Diamond])$.

	We come to the construction of $\mathcal{M}_{\TO,\mathcal{D}}$. Let $\mathcal{D}$ be a data instance. A type $\tp$ is \emph{consistent with $\mathcal{D}$ at $k$} if $A(k)\in \mathcal{D}$ implies $\neg A\not\in \tp$, for any atom $A$. We next define the notion of a decoration. Let $I_{0},\ldots,I_{n}$ be a partition
	of $\mathbb{N}$ into nonempty intervals $I_{0},\ldots,I_{n}$ with $I_{n}$ of the form $[m,\infty]$ for some $m$ with $\max \mathcal{D} < m \leq \max \mathcal{D} + |\TO|+1$ and $\max I_{k} +1 = \min I_{k+1}$ for all $k<n$.
	Let $f$ be a function that associates with each interval $k\leq n$ a nonempty set $f(k)$ of $\TO$-satisfiable types. Intuitively, the types in $f(k)$ are types that we aim to satisfy in the interval $I_{k}$. We then call $D=(I_{0},\ldots,I_{n},f)$ a \emph{pre-decoration} of $\mathcal{D}$.
	We say that a model $\I$ is \emph{consistent with $D=(I_{0},\ldots,I_{n},f)$}
	if it is defined by a sequence
	$$
	\tp_{0},\tp_{1},\ldots
	$$
	of types $\tp_{i}$ such that
	\begin{enumerate}
		\item if $i\in I_{k}$, then $\tp_{i} \in f(k)$ and $\tp_{i}$ is consistent with $\mathcal{D}$ at $i$, for all $i\geq 0$;
		\item each $\tp\in f(n)$ occurs infinitely often as $\tp_{i}$ in $\I$ for $i\geq m$.
	\end{enumerate}
	Then $D=(I_{0},\ldots,I_{n},f)$ is a \emph{decoration of $\mathcal{D}$ for $\TO$} if every model $\tp_{0},\tp_{1},\ldots$ that is consistent with $D$
	satisfies $\tp_{i}$ at timepoint $i$ (and this ia, in particular, a model of $\TO$). Note that models that are consistent with $D$ are trivially models of $\mathcal{D}$. Thus, any $D$ defines the set $\mathcal{M}_{D}$ of models that are consistent with $D$ and these are also always models of $\TO$ if $D$ is a decoration of $\mathcal{D}$ for $\TO$. $D=(I_{0},\ldots,I_{n},f)$ also defines a \emph{canonical model} $\I_{D}$ as follows: fix any ordering $\tp_{0},\ldots,\tp_{j-1}$ of $f(n)$ and assume $I_{n}=[m_{D},\infty]$. Then let $\I_{D}$ be defined by setting
	\begin{itemize}
		\item for $i\in I_{k}$ with $k<n$, $i\in A^{\I_{D}}$ if $A\in \tp$ for all $\tp\in f(k)$ that are consistent with $\mathcal{D}$ at $i$;
		\item for $i=m_{D}+j_{0}+kj$ with $j_{0}<j$, $i\in A^{\I_{D}}$ if $A\in \tp_{j_{0}}$.
	\end{itemize}
	Thus, for $i<m_{D}$, $\I_{D}$ is defined as the intersection of all models that are consistent with $D$ and for $i\geq m_{D}$ we repeat the pattern $\tp_{0},\ldots,\tp_{j-1}$ again and again. Note that $\I_{D}$ is of the form defined in \eqref{periodA2}. We show the following lemma connecting $\mathcal{M}_{D}$ and $\I_{D}$.
	\begin{lemma}\label{lem:eqcompl}
		For every $\varkappa\in \mathcal{Q}[\Diamond]$ and every $i<m_{D}$, we have $\mathcal{J},i\models \varkappa$ for all $\mathcal{J}\in \mathcal{M}_{D}$ iff $\mathcal{I}_{D},i\models \varkappa$.
	\end{lemma}
	\begin{proof}
		Obtain $\mathcal{M}$ from $\mathcal{M}_{D}$ by replacing for each $\mathcal{J}\in \mathcal{M}_{D}$ the final part of $\J$ based on the interval $I_{n}$ by the final part of $\I_{D}$ based on $I_{n}$. Then clearly $\mathcal{J},0\models \varkappa$ for all $\mathcal{J}\in \mathcal{M}_{D}$ iff $\mathcal{J},0\models \varkappa$ for all $\mathcal{J}\in \mathcal{M}$. It is therefore sufficient to prove the claim for $\mathcal{M}$ instead of $\mathcal{M}_{D}$.
		
		The proof is by induction on $\ell$ for $\varkappa$ of the form
		$
		\rho_0 \land \Diamond (\rho_1 \land \Diamond (\rho_2 \land \dots \land \Diamond \rho_{\ell}) )$.
		For $\ell=0$ the claim follows from the definition.
		
		Assume that the claim has been proved for $\ell\geq 0$, $\varkappa = \rho_0 \land \Diamond (\rho_1 \land \Diamond (\rho_2 \land \dots \land \Diamond \rho_{\ell+1}) )$, and $\J,i\models \varkappa$ for some $i<m_{D}$ and all $\J\in \mathcal{M}$. We have to show that $\I_{D},i\models \varkappa$.
		
		If there exists $i\geq m_{D}$ such that $\J,i\models \rho_1 \land \Diamond (\rho_2 \land \dots \land \Diamond \rho_{\ell}) )$, then we are done.
		
		Otherwise, we show that there exists $i'$ with $i<i'<m_{D}$ such that $\J,i'\models \rho_1 \land \Diamond (\rho_2 \land \dots \land \Diamond \rho_{\ell}) )$ for all $\J\in \mathcal{M}$. Then the claim follows by IH.
		
		We first observe that there exists $i'$ with $i<i'<m_{D}$ such that for $i'\in I_{k}$ we have $\rho_{1} \subseteq \tp$ for all $\tp\in f(k)$ that are consistent with $\mathcal{D}$ at $i'$.
		
		For assume that this is not the case. Then construct a model $\J\in \mathcal{M}$ by choosing for every $j$ with $i<j<m_{D}$
		such that $j\in I_{k}$ a $\tp\in f(k)$ that is consistent with $\mathcal{D}$ at $j$ such that $\rho_{1} \not\subseteq \tp$. Define $\J$ using these $\tp_{j}$. Then $\J,i\not\models \varkappa$, a contradiction.
		
		Let $i'$ be minimal $i<i'<m_{D}$ such that for $i'\in I_{k}$ we have $\rho_{1} \subseteq \tp$ for all $\tp\in f(k)$ that are consistent with $\mathcal{D}$ at $i'$.
		
		We next show that $\J,i'\models \Diamond (\rho_2 \land \dots \land \Diamond \rho_{\ell}) )$ for all $\J\in \mathcal{M}$. Assume that this is not the case.
		Let $\J$ be a witness. Then we construct a new model $\J'\in \mathcal{M}$ by
		refuting $\rho_{1}$ between $i$ and $i'$ (possible by minimality of $i'$) and then adding $\J$ from $i'$. Then $\J',i\not\models \varkappa$, a contradiction.
		
		It follows that $\J,i'\models \rho_{1} \land \Diamond (\rho_2 \land \dots \land \Diamond \rho_{\ell}) )$ for all $\J\in \mathcal{M}$, as required.
	\end{proof}
	We next define the decorations we work with. Given any model $\I$ of $\mathcal{D}$ and $\TO$ of the form \eqref{periodA2},
	we obtain a decoration $D_{\I}=(I_{0},\ldots,I_{n},f)$ with $n\leq 2|\TO|+2$ as follows. Call a node \emph{$i$ maximal in $\I$ for $\TO$} if there exists $C$ with $\Box C \in \sub(\TO)$ such that $i\models \Box C \wedge \neg C$.
	
	Assume $I_{0},\ldots,I_{\ell}$ and $f(0),\ldots,f(\ell)$ have been defined
	already and $I_{\ell}$ is not of the form $[m,\infty]$ (if $I_{\ell}$ is of the form $[m,\infty]$ we are done).
	We next define $I_{\ell+1}$ (and possibly $I_{\ell+2}$).
	\begin{enumerate}
		\item If $\max I_{\ell}<k_{0}$, then we proceed as follows:
		if $\max I_{\ell}+1$ is either maximal for $\TO$ in $\I$ or $\max I_{\ell}+1=k_{0}$, then set $I_{\ell+1}=\{\max I_{\ell}+1\}$ and $f(\ell+1)=\{\tp_{\I}(\max I_{\ell}+1)\}$.
		
		Otherwise let
		$$
		k:= \min\{ k>\max I_{\ell} \mid \text{$k$ is maximal for $\TO$ in $\I$ or $k=k_{0}$}\}
		$$
		and set $I_{\ell+1}=[\max I_{\ell},k-1]$, $I_{\ell+2}=\{k\}$,
		$f(\ell+1) = \{\tp_{\I}(k) \mid k\in I_{\ell+1}\}$, and $f(\ell+2)=\{\tp_{\I}(k)\}$.
		\item Otherwise $\max I_{\ell}\geq k_{0}$. Then let $I_{\ell+1}=[k_{0},\infty]$ and $f(\ell+1)=\{\tp_{k_{0}+1},\ldots,\tp_{k_{0}+l}\}$.
	\end{enumerate}
	One can easily show that $D_{\I}=(I_{1},\ldots,I_{n},f)$ is indeed a decoration of $\mathcal{D}$ for $\TO$ and $n\leq 2|\TO|+2$. Note also that $\I$ itself is consistent with $D_{\I}$. The following lemma summarises our findings.
	\begin{lemma}
		For any $\mathcal{D}$ one can construct in polynomial time a set $\mathcal{F}_{\TO,\mathcal{D}}$ of decorations $D=(I_{0},\ldots,I_{n},f)$ of $\mathcal{D}$ for $\TO$ such that
		$n\leq 2|\TO|+2$ and the following are equivalent for any $\varkappa \in \mathcal{Q}[\Diamond]$:
		\begin{enumerate}
			\item $\TO,\mathcal{D}\models \varkappa(0)$;
			\item $\I,0 \models \varkappa$ for every $\I \in \mathcal{M}_{D}$ and $D\in\mathcal{F}_{\TO,\mathcal{D}}$;
			\item $\I_{D},0 \models \varkappa$ for every $D\in \mathcal{F}_{\TO,\mathcal{D}}$.
		\end{enumerate}
	\end{lemma}
	\begin{proof}
		Models of the form \eqref{periodA2} satisfying $\TO$ and $\mathcal{D}$ are complete in the sense that the following conditions are equivalent for all $\varkappa\in \mathcal{Q}_{p}[\Diamond]$:
		\begin{itemize}
			\item $\TO,\mathcal{D}\models \varkappa(0)$;
			\item $\I,0 \models \varkappa$ for all models $\I$ of $\TO$ and $\mathcal{D}$ of the form \eqref{periodA2}.
		\end{itemize}
		Hence the class of models $\mathcal{M}_{D_{\I}}$ with $\I$ a model of $\TO$ and $\mathcal{D}$ of the form \eqref{periodA2} is also complete. Hence the equivalence of Points~1. to 2. holds if we define $\mathcal{F}_{\TO,\mathcal{D}}$ as the class of decorations $D=(I_{0},\ldots,I_{n},f)$ of $\mathcal{D}$ for $\TO$ with $n\leq 2|\TO| +2$. The equivalence of Points~2. and 3. follows from Lemma~\ref{lem:eqcompl}.
		It remains to show that $\mathcal{F}_{\TO,\mathcal{D}}$ can be constructed in polynomial time. The set of pre-decorations $(I_{0},\ldots,I_{n},f)$ of $\mathcal{D}$ with $n\leq 2|\TO|+2$ can clearly be constructed in polynomila time in $|\mathcal{D}|$. It thus remains to check in polynomial time whether a pre-decoration is a decoration. But such a check is straightforward as a pre-decoration $D=(I_{0},\ldots,I_{n},f)$ is a decoraton if, and only if, the following condition holds: for any subformula $\Box C$ of $\TO$, any $i\leq n$, any $\tp\in f(i)$, and any $k\in f(i)$ with $\tp$ consistent with $\mathcal{D}$ at $k$: $\Box C \in \tp$ iff no $\tp'\in f(i)$ with $C\not\in \tp'$ is consistent with $\mathcal{D}$ at any $k'\in I_{i} \cap [k+1,\infty]$ and no $\tp'\in f(j)$ with $C\not\in \tp'$ and $j>i$ is consistent with $\mathcal{D}$ at any $k'\in I_{j}$.
\end{proof}
	
The set $\mathcal{M}_{\TO,\mathcal{D}}$ of models required for the construction of the algorithms is now defined by setting $\mathcal{M}_{\TO,\mathcal{D}}=
\bigcup_{D\in \mathcal{F}_{\TO,\mathcal{D}}}\mathcal{M}_{D}$.
\end{proof}


\section{Proofs for Section~\ref{sec:full}}

\textbf{Theorem~\ref{thm:fullLTL}.}
{\em $(i)$ $\QBE(\LTL,\mathcal{Q})$ is in $2\ExpTime{}$, for any $\mathcal{Q} \in \{\, \mathcal{Q}[\Diamond], \mathcal{Q}[\nxt, \Diamond], \mathcal{Q}[\Us]\,\}$.
	$(ii)$ $\mathsf{QBE}(\LTL, \mathcal{Q})$ is in $2\ExpSpace{}$, for any $\mathcal{Q} \in \{\, \mathcal{Q}_p[\Diamond], \mathcal{Q}_p[\nxt, \Diamond], \mathcal{Q}_p[\U]\,\}$.
}	
\begin{proof}
Let $E$ be an example set and $\TO$ an $\LTL{}$ ontology.
To show the results for $\QBE(\LTL,\mathcal{Q}[\Us])$ and $\mathsf{QBE}(\LTL, \mathcal{Q}_p[\U])$, it is enough, by Lemma~\ref{lemmain}, to show that the construction of $S$ in Lemma~\ref{th:bool-repr}, representing $\Abox$ from $E$ and $\TO$, can be done in $2\ExpTime{}$. Consider an unlabelled transition system $\mathcal S$ (which can be defined as a transition system of the kind we have with the unary alphabet $\Sigma_1 = \Sigma_2 = \{  \emptyset \}$) with the states $\tp$, where $\tp$ is a type realisable in $\TO, \D$. Given $\mathcal S$ and a set of realisable types $\avec{T}$, we set $\mathcal S_{\avec{T}}$ to be $\mathcal S$ restricted to $\tp$ that are reachable from some $\tp' \in \avec{T}$ and the initial states $\avec{T}$. For given $\avec{T}_1$, $\Gamma$ and $\avec{T}_2$, we have $\avec{T}_1 \to_\Gamma \avec{T}_2$ iff, for each path $\mathfrak s$ in $\mathcal T_{S_{\avec{T}_1}}$, there is a position $p_{\mathfrak s} > 0$ satisfying the following conditions: $(i)$ the set of types at all $p_{\mathfrak s}$ coincides with $\avec{T}_2$; $(ii)$ $A \in \tp$ for every $\tp$ at a position $p \in (0, p_{\mathfrak s})$ for every $\mathfrak s$ iff $A \in \Gamma$, for each $A \in \Sigma^\bot$; $(iii)$ the set of all $\tp$ from $(ii)$ does not intersect with $\avec{T}_2$. Next, we observe that if the positions $p_{\mathfrak s}$ satisfying $(i)$, $(ii)$, $(iii)$ exist, then there exist such positions $p_{\mathfrak s} \leq |\mathcal S|$. (Intuitively, this is because whenever $p_{\mathfrak s} > |\mathcal S|$, there exists a type in $\mathfrak s$ that repeats itself.) Thus, we need to check conditions $(i)$--$(iii)$ in a tree of depth $|\mathcal S|$. This can be done in a branch-by-branch fashion using a (non-deterministic) algorithm working in $\PSpace{}$ in $|\mathcal S|$. It remains to observe that $\mathcal S$ itself can be constructed in $\ExpTime{}$ in $|\TO|$ and that to construct $S$ we need to check $\avec{T}_1 \to_\Gamma \avec{T}_2$ for $O(2^{2^{|\TO|}})$-many pairs $(\avec{T}_1, \avec{T}_2)$.

To obtain the results of the theorem for $\mathcal{Q}[\nxt, \Diamond]$, we construct $S$ (cf.\ Lemma~\ref{th:bool-repr}) as above but using the transition relation $\avec{T}_1 \to_{\Gamma}' \avec{T}_2$ for $\Gamma \in \{ \emptyset, \Sigma^\bot\}$ only (i.e., $\Sigma_2 = \{ \emptyset, \Sigma^\bot\}$ in the definition of a transition system). We set $\avec{T}_1 \to_{\Sigma^\bot}' \avec{T}_2$ if $\avec{T}_1 \to_{\Sigma^\bot} \avec{T}_2$ and $\avec{T}_1 \to_{\emptyset}' \avec{T}_2$ if $\avec{T}_1 \to_{\Gamma} \avec{T}_2$ for $\Gamma \neq \Sigma^\bot$. It is easy to verify that $S$ represents $\TO, \Abox$ for the class of $\mathcal{Q}[\nxt, \Diamond]$ queries. The case $\mathcal{Q}[\Diamond]$ is left to the reader.
\end{proof}


\end{document}